%% file: InexSchwarzschild2.tex
\documentclass[reqno]{article}

\usepackage{amsthm}
\usepackage{amsmath}
\usepackage{amssymb}
\usepackage{graphicx}
\usepackage{color}
\usepackage{bbm}
\usepackage[top=2.5cm, bottom=2.5cm, left=2.5cm, right=2.5cm]{geometry}
\usepackage{enumerate}
\usepackage{faktor}
\usepackage{nicefrac}
\usepackage[nottoc]{tocbibind}

\newcommand{\N}{\mathbb{N}}
\newcommand{\R}{\mathbb{R}}

\newcommand{\ed}{(-\varepsilon, \varepsilon)^{d+1}}
\newcommand{\diam}{\mathrm{diam}}
\newcommand{\D}{\mathcal{D}}
\newcommand{\ux}{\underline{x}}
\newcommand{\Sd}{\mathbb{S}^{d-1}}
\newcommand{\Mint}{M_{\mathrm{int}}}
\newcommand{\gint}{g_{\mathrm{int}}}
\newcommand{\Mext}{M_{\mathrm{ext}}}
\newcommand{\gext}{g_{\mathrm{ext}}}
\newcommand{\Mmink}{M_{\mathrm{Mink}}}
\newcommand{\Mmax}{M_{\mathrm{max}}}
\newcommand{\gmax}{g_{\mathrm{max}}}
\newcommand{\hext}{h_\mathrm{ext}}
\newcommand{\ovg}{\overline{\gamma}}
\newcommand{\oMe}{\overline{M}_\mathrm{ext}}

\begin{document}

\numberwithin{equation}{section}
\newtheorem{theorem}[equation]{Theorem}
\newtheorem{remark}[equation]{Remark}
\newtheorem{claim}[equation]{Claim}
\newtheorem{lemma}[equation]{Lemma}
\newtheorem{definition}[equation]{Definition}
\newtheorem{premiss}[equation]{Premiss}
\newtheorem{corollary}[equation]{Corollary}
\newtheorem{proposition}[equation]{Proposition}
\newtheorem*{theorem*}{Theorem}

\title{The $C^0$-inextendibility of the Schwarzschild spacetime and the spacelike diameter in Lorentzian geometry}
\author{Jan Sbierski\thanks{Department for Pure Mathematics and Mathematical Statistics, University of Cambridge,
Wilberforce Road,
Cambridge,
CB3 0WA,
United Kingdom}}
\date{\today}

\newgeometry{top=2.0cm, bottom=2.5cm, left=2.5cm, right=2.5cm}
\maketitle

\begin{abstract}
The maximal analytic Schwarzschild spacetime is manifestly inextendible as a Lorentzian manifold with a twice continuously differentiable metric. 
In this paper, we prove the stronger statement that it is even inextendible as a Lorentzian manifold with a continuous metric. To capture the obstruction to continuous extensions through the curvature singularity, we introduce the notion of the \emph{spacelike diameter} of a globally hyperbolic region of a Lorentzian manifold with a merely continuous metric and give a sufficient condition for the spacelike diameter to be finite. 
The investigation of low-regularity inextendibility criteria is motivated by the strong cosmic censorship conjecture.
\end{abstract}

\tableofcontents

\restoregeometry
\section{Introduction}

A connected Lorentzian manifold $(M,g)$ is called $C^k$-inextendible, if there does not exist a connected Lorentzian manifold $(\tilde{M}, \tilde{g})$ (of the same dimension as $M$) with a $C^k$-regular metric $\tilde{g}$ in which $M$ isometrically embeds as a proper subset.
In this paper we prove:
\begin{theorem*}
The $(d+1)$-dimensional maximal analytic Schwarzschild spacetime is $C^0$-inextendible for all $d \geq 3$.
\end{theorem*} 
%the $C^0$-inextendibility of the maximal analytic Schwarzschild spacetime in spatial dimensions $d \geq 3$. 

The interest of the study of low-regularity inextendibility criteria for Lorentzian manifolds stems from the \emph{strong cosmic censorship conjecture} in general relativity, which was conceived by Penrose and states, in physical terms, that, generically, the theory of general relativity should uniquely predict the fate of all local (classical) observers. A widely accepted mathematical formulation of this conjecture is the following: 
\begin{equation}
\label{SCC}
\parbox{0.75\textwidth}{
For \emph{generic} asymptotically flat initial data for the vacuum Einstein equations $Ric(g) = 0$, the maximal globally hyperbolic development is inextendible as a suitably regular Lorentzian manifold.}
\end{equation} 
However, there is still no consensus on the exact meaning of `suitably regular' (and neither is there on the exact definition of `generic' - but in the following we focus on the question which regularity class one should impose).
In a series of papers \cite{Chris91}, \cite{Chris93}, \cite{Chris99}, Christodoulou proved a corresponding formulation of the strong cosmic censorship conjecture for the spherically symmetric Einstein-scalar field system. He showed that `generically', the maximal globally hyperbolic development is $C^0$-inextendible in the class of \emph{spherically symmetric} Lorentzian manifolds\footnote{However, no argument is given which would show the $C^0$-inextendibility in the class of Lorentzian manifolds without symmetry assumptions.}. This very strong form of inextendibility of the maximal globally hyperbolic development is, however, particular to the spherically symmetric Einstein-scalar field system: in \cite{Daf03}, \cite{Daf05a}, Dafermos showed that for solutions to the spherically symmetric Einstein-Maxwell-scalar field system, arising from initial data close to that of the Reissner-Nordstr\"om solution,  the maximal globally hyperbolic development is in fact $C^0$-\emph{extendible}. Moreover, he showed that `generically', the Hawking mass blows up which implies inextendibility as a \emph{spherically symmetric} Lorentzian manifold with Christoffel symbols locally in $L^2$. 

Recently, Dafermos and Luk announced the $C^0$-stability of the Kerr-Cauchy horizon (without any symmetry assumptions), \cite{DafLuk}. More precisely, they show that for initial data on the event horizon, which approach the geometry of the Kerr event horizon in a manner compatible with the expected behaviour of the Kerr exterior under generic perturbations, the solution exists all the way up to the Cauchy horizon `in a neighborhood of timelike infinity', and, moreover, is continuous and $C^0$-close to the unperturbed Kerr solution.  
This result, together with the expected stability of the Kerr exterior, shows that the strong cosmic censorship conjecture \eqref{SCC} does not hold if one imposes the strong $C^0$-inextendibility on the maximal globally hyperbolic development\footnote{However, the result of Dafermos and Luk does not rule out the possibility that some part of the generic singularity inside a black hole is indeed $C^0$-inextendible.}. So what regularity should one choose in the mathematical formulation of the strong cosmic censorship conjecture \eqref{SCC}?

Let us recall that the physical motivation of the conjecture is the belief in determinism in classical physics. Until one has gained a better understanding for the regime in which the classical theory of gravity needs to be replaced by a quantum theory, it seems sensible, in order to do justice to the physical content of the conjecture, to rule out non-unique extensions of the maximal globally hyperbolic development as, possibly weak, \emph{solutions to the Einstein equations}. Such non-unique extensions could be constructed from a local existence result for the Einstein equations. Hence, this suggests that one should \emph{at least} require inextendibility of the maximal globally hyperbolic development in all regularity classes which admit a local existence result for the Einstein equations (in a weak form). In the light of the recent low-regularity existence results, \cite{KlRodSzef12}, \cite{LukRod12}, this regularity class is, in particular, below $C^2$.   

The currently favoured formulation of the strong cosmic censorship conjecture, introduced by Christodoulou in \cite{Chris09}, is to require generic inextendibility as a Lorentzian manifold with \emph{Christoffel symbols locally in $L^2$}, since this even rules out the mere formulation of the Einstein equations (in a weak form). Moreover, the results \cite{Daf03}, \cite{Daf05a} by Dafermos lend support to the veracity of this version of the strong cosmic censorship conjecture.

It now seems, however, that the methods developed so far for proving inextendibility of Lorentzian manifolds have not penetrated regularities lower than $C^2$ (or $C^{1,1}$).\footnote{That is, if one disregards the spherically symmetric case where one can use the behaviour of the radius and the Hawking mass of the spheres of symmetry in a straightforward way as obstructions.} Most of the known inextendibility results exploit geodesic completeness or the blow-up of curvature scalars as obstructions to $C^2$-extensions. 

The aim of this paper is to initiate the study of low-regularity extensions of Lorentzian manifolds. In particular we show how one can use the infiniteness of the timelike diameter of the exterior of the Schwarzschild spacetime together with the infiniteness of a newly introduced geometric quantity, called the \emph{spacelike diameter}, in the interior, to show that the maximal analytic Schwarzschild spacetime is $C^0$-inextendible. Although, as explained above, the singularity structure of the Schwarzschild solution is not expected to be generic, it is nevertheless an important and instructive illustrative model in general relativity. Furthermore, we feel that the techniques developed in this paper might also be helpful for further physically more relevant investigations.

\subsection{$C^0$-extendibility and $C^2$-inextendibility}

There are various examples of Lorentzian manifolds which are $C^2$-inextendible while being $C^0$-extendible. The physically most relevant example is probably given by the spacetimes constructed by Dafermos in \cite{Daf03}, \cite{Daf05a}. Another, very simple, example is provided by the Lorentzian manifold $\big((0, \infty) \times \R ,g\big)$, where the Lorentzian metric $g$ is given by 
\begin{equation*}
g = e^{2 \sqrt{t}} \big(-\, dt^2 + dx^2\big)\;.
\end{equation*}
Clearly, this Lorentzian manifold can be continuously extended to $[0, \infty) \times \R$. On the other hand, the scalar curvature of $g$ is found to be $R = - \frac{1}{2 e^{2\sqrt{t}}} \cdot \frac{1}{t^{\nicefrac{3}{2}}}$, which shows that no $C^2$-extension to $[0, \infty) \times \R$ is possible.

Let us also briefly remark that there have been attempts to deduce from the strength of a \emph{curvature singularity} whether it forms an obstruction to $C^0$-extensions or not. The example \cite{Ori00} by Ori of a \emph{strong curvature singularity} which admits a $C^0$-extension showed, however, that the relation of these two concepts is not as straightforward as originally expected. 

While there is an abundance of examples of $C^2$-inextendible Lorentzian manifolds, we are not aware, however, of any previous results showing $C^0$-inextendibility of a Lorentzian manifold. In particular, the methods presented in this paper can be used to show that the toy-example $(M,g)$, where $M = \R \times (0, \infty)$ and
\begin{equation*}
g = \frac{1}{h(r)} \, dt^2 - h(r) \, dr^2 \,,
\end{equation*}
is $C^0$-inextendible. Here, $t$ is the coordinate on $\R$, $r$ the coordinate on $(0,\infty)$, and $h(r)$ is a smooth and monotonically increasing function with $h(r) = r$ for $0<r< \frac{1}{2}$ and $h(r) = 1$ for $r \geq 1$.

\subsection{Sketch of the proof}

In the following we give a brief sketch of the proof of the $C^0$-inextendibility of the Schwarzschild spacetime. We refer the reader unfamiliar with any notation or definition in the following overview to Section \ref{Defs} (and Section \ref{SecSchw} for the definition of the Schwarzschild spacetime).

One assumes that there is a $C^0$-extension $\iota : \Mmax \hookrightarrow \tilde{M}$ of the maximal analytic Schwarzschild spacetime $\Mmax$ (in this overview we will not distinguish between $\Mmax$ and $\iota(\Mmax) \subset \tilde{M}$). Lemma \ref{ExtensionLemma} shows that there is then a  timelike curve $\tilde{\gamma}$ in $\tilde{M}$ leaving the Schwarzschild spacetime $\Mmax$. This timelike curve can only `leave through' either the curvature singularity at $\{r=0\}$, or  timelike or null infinity. These two cases are separately discussed and ruled out. Here, an important ingredient is Lemma \ref{NormalForm}, which introduces near-Minkowskian coordinates in a neighbourhood $\tilde{U}$ of the point in the boundary of $\Mmax$, through which $\tilde{\gamma}$ leaves $\Mmax$.  In particular, by estimating the future and pasts of points in $\tilde{U}$ by cones with respect to the chosen coordinates, we can choose $s_0$ such that the closure of $I^+\big(\tilde{\gamma}(s_0), \tilde{U}\big) \cap I^-\big(\tilde{\gamma}(0), \tilde{U}\big)$ in $\tilde{M}$ is contained in $\tilde{U}$. 
\begin{center}
\def\svgwidth{5cm}
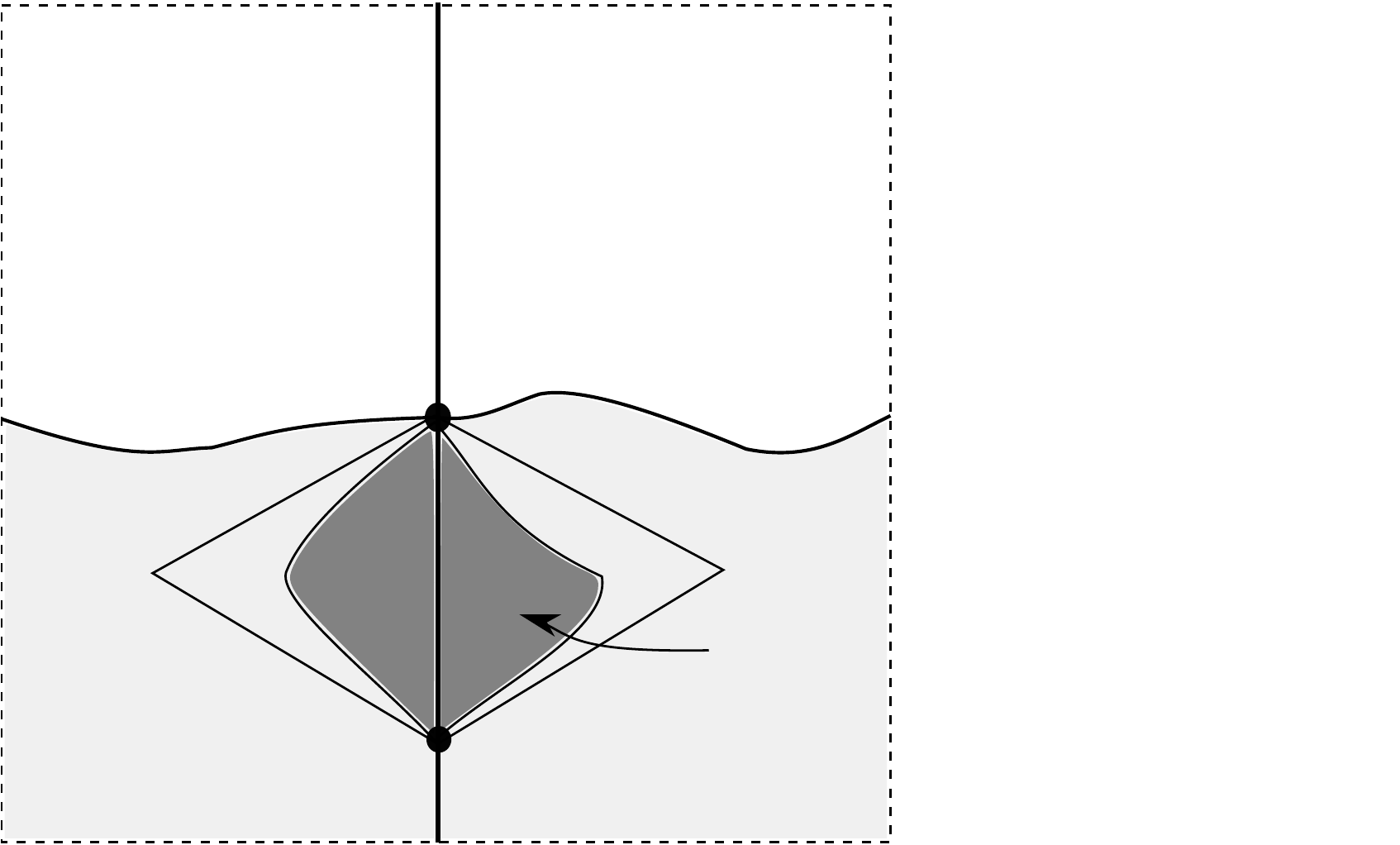
\end{center}
The important next step in deriving a contradiction is to identify subsets of $\tilde{U}$ with subsets in $\Mmax$. Let us remark here, that this step would be much simpler in Riemannian geometry, since one could use the distance function to do so.

We first discuss the case that $\tilde{\gamma}$ leaves $\Mmax$ through timelike or null infinity. Without loss of generality, we can assume that $\tilde{\gamma}$ is future directed in $\Mmax$. Ideally, we would like to have the following identification for all $s_0 < s < 0$:
\begin{equation}
\label{Wish}
I^+\big(\tilde{\gamma}(s_0), \tilde{U}\big) \cap I^-\big(\tilde{\gamma}(s), \tilde{U}\big) = I^+\big(\tilde{\gamma}(s_0), \Mmax\big) \cap I^-\big(\tilde{\gamma}(s), \Mmax\big) \;.
\end{equation}
However, since on the left we consider a causality relation with respect to a small neighbourhood of $\tilde{M}$, and on the right with respect to $\Mmax$, such an identification is not justified in general (and indeed, neither of the two inclusions is)! Here, however, we exploit that the exterior of the Schwarzschild black hole is future one-connected, i.e., in particular, any future directed timelike curve in $\Mmax$ from $\tilde{\gamma}(s_0)$ to $\tilde{\gamma}(s)$ is homotopic to $\tilde{\gamma}|_{[s_0,s]}$ via timelike curves with fixed endpoints. Since the closure of $I^+\big(\tilde{\gamma}(s_0), \tilde{U}\big) \cap I^-\big(\tilde{\gamma}(0), \tilde{U}\big)$ in $\tilde{M}$ is contained in $\tilde{U}$, a homotopy as above cannot escape $I^+\big(\tilde{\gamma}(s_0), \tilde{U}\big) \cap I^-\big(\tilde{\gamma}(0), \tilde{U}\big)$. This yields the relation $``\supseteq"$ in \eqref{Wish}. A contradiction is now obtained since the timelike diameter of $I^+\big(\tilde{\gamma}(s_0), \tilde{U}\big) \cap I^-\big(\tilde{\gamma}(s), \tilde{U}\big)$ is bounded for $s \nearrow 0$, while the timelike diameter of $I^+\big(\tilde{\gamma}(s_0), \Mmax\big) \cap I^-\big(\tilde{\gamma}(s), \Mmax\big)$ grows beyond bound for $s \nearrow 0$.

In order to show that $\tilde{\gamma}$ cannot leave through $\{r=0\}$ either, we introduce the  \emph{spacelike diameter} of a globally hyperbolic  Lorentzian manifold $N$ with a $C^0$-regular metric, defined by
\begin{equation*}
\diam_s(N) := \sup_{\substack{\Sigma \textnormal{ Cauchy} \\ \textnormal{hypersurface of } N}} \diam\, \Sigma \;.
\end{equation*} 
Moreover, we call a chart $\psi : N \supseteq U \to D \subseteq \R^{d=1}$ a \emph{regular flow chart} for the Lorentzian manifold $N$ if the following three properties are satisfied: $i)$ the absolute value of the metric components in this chart are uniformly bounded and, moreover, $g_{00}$ is negative and bounded away from $0$; $ii)$ the domain $D$ is of the form $D = \bigcup_{\ux \in B} I_{\ux} \times \{\ux\}$, where $B \subseteq \R^d$ has finite diameter with respect to the Euclidean metric on $\R^d$ and $I_{\ux}$ is an open and connected interval; $iii)$ the timelike curves $I_{\ux} \ni s \mapsto \psi^{-1}(s, \ux)$ are inextendible in $N$.

We then establish the following
\begin{theorem*}
Let $(N,g)$ be a connected and globally hyperbolic Lorentzian manifold with a $C^0$-regular metric $g$ and let $\psi_k : U_k \to D_k = \bigcup_{\ux \in B_k} I_{\ux} \times \{\ux\}$, $k = 1, \ldots, K$, be a finite collection of regular flow charts for $N$ with $\bigcup_{1 \leq k \leq K} U_k = N$.

Then one has $\diam_s(N) < \infty$.
\end{theorem*}
For the proof, one first notes that since there are finitely many charts that cover $N$, it suffices to show that for each $k$ the diameter of $\Sigma \cap U_k$ is uniformly bounded for all Cauchy hypersurfaces $\Sigma$ of $N$. To establish the uniform bound, we show that $\psi_k(\Sigma \cap U_k)$ can be written as a graph over $B_k$ with uniformly bounded slope. Together with the uniform boundedness of the metric components, one obtains that the components of the induced metric on $\Sigma \cap U_k$ are uniformly bounded (independently of the Cauchy hypersurface). This, together with the finiteness of the diameter of $B_k$ then implies the uniform bound on the diameter of $\Sigma \cap U_k$.

Assuming now that $\tilde{\gamma}$ leaves through $\{r=0\}$, we consider a near-Minkowskian neighbourhood $\tilde{U}$ as before. One now has to show that there is a $\mu >0$ such that $I^+\big(\tilde{\gamma}(-\mu), \Mmax\big)$ is contained in $\tilde{U}$. This uses crucially the spacelike nature of the boundary depicted below in a Penrose-style representation.
\begin{center}
\def\svgwidth{9cm}
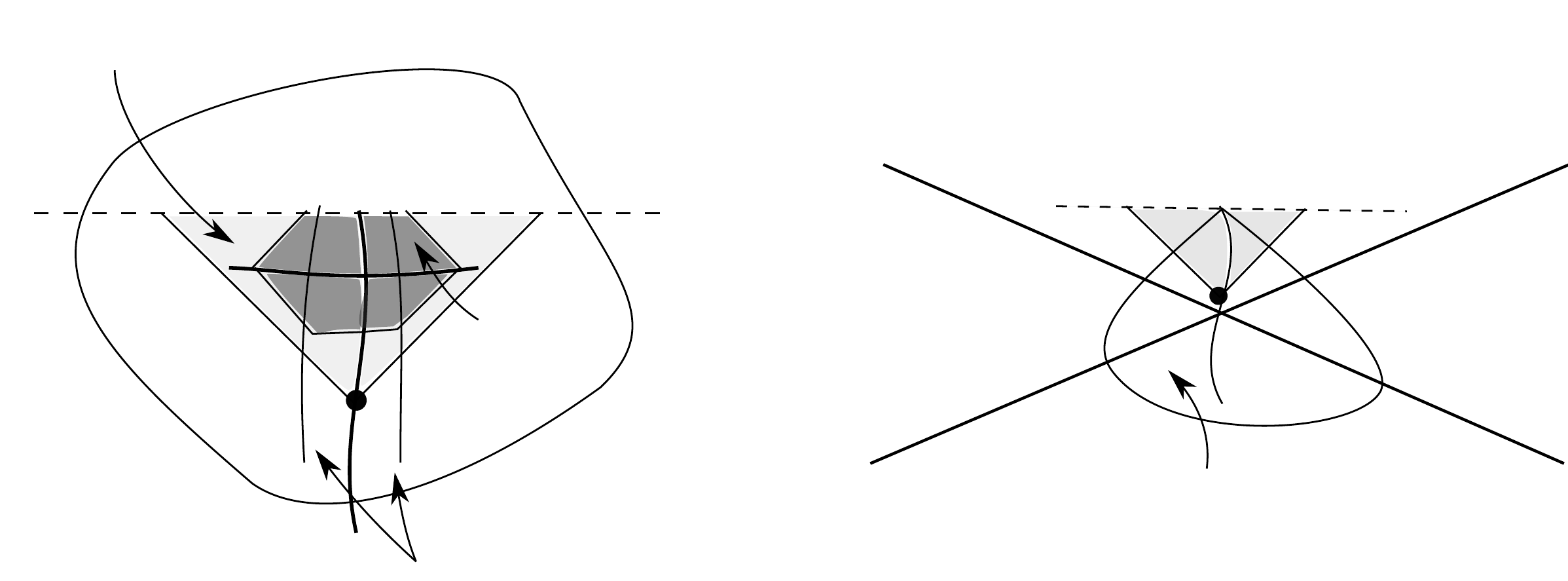
\end{center}
One then considers a spherically symmetric globally hyperbolic region $N$ as depicted above and shows that the spacelike diameter of $N$ is infinite. Here, one exploits that the coefficient of $dt^2$ in the standard form \eqref{SchwarzschildMetricInt} of the Schwarzschild metric diverges for $r \to 0$.  

The contradiction will be obtained from the above theorem. First note, however, that $I^+\big(\tilde{\gamma}(-\mu), \Mmax\big)$ is not spherically symmetric, and hence, $N$ is actually not completely contained in $I^+\big(\tilde{\gamma}(-\mu), \Mmax\big)$. The above representation is, however, accurate if one fixes a point on the sphere and only depicts the corresponding $t,r$-plane. In particular, it holds that an open angular segment of $N$ is completely contained in $I^+\big(\tilde{\gamma}(-\mu), \Mmax\big)$. The near-Minkowskian chart for $\tilde{U}$ can now be used to construct a regular flow chart for $N$ that covers an open angular segment of $N$. Using the spherical symmetry of the Schwarzschild spacetime we can rotate this regular flow chart to obtain a collection of regular flow charts that cover $N$. The above theorem now applies and yields the contradiction.

\subsection{Further applications and open questions}

In the following we collect a few more applications of the techniques developed in this paper and compile some open problems in the realm of $C^0$-extensions of Lorentzian manifolds.

\begin{enumerate}[i)]
\item The timelike diameter can also be used as an obstruction to show that the \emph{de Sitter spacetime} is $C^0$-inextendible. The proof proceeds along the same steps outlined in the first half of the above sketch (and in more detail in Section \ref{SecMink}), where one can exploit, for example, the conformal isometry of de Sitter space with part of the Einstein static universe in order to prove a) the future one-connectedness of de Sitter space and b) the infiniteness of the timelike diameter of the intersection of the past of a future directed timelike curve `leaving' de Sitter space with the future of an arbitrary point along the same curve.

Other straightforward applications of the methods presented in this paper should be to the proof of the $C^0$-inextendibility of the \emph{Schwarzschild-de Sitter spacetime}, \emph{Nariai}, and \emph{Pleba\'{n}ski-Hacyan} - to name a few. 

\item An interesting direction of research is regarding the $C^0$-inextendibility of cosmological spacetimes with a big bang singularity - and in particular the FRW models. A clarification of the obstruction to $C^0$-extensions through the initial singularity would, in particular, shed more light on the structure of the singularity. Here, new techniques are needed. An exception is the Kasner solution with a negative $p_i$, which seems to be amenable to the methods developed in this paper, since the spacelike diameter near the singularity is diverging.

\item Another direction for further research is to leave the class of exact solutions and prove that the spacetimes constructed by Christodoulou in the series of papers \cite{Chris91}, \cite{Chris93}, \cite{Chris99}, and which arise from the generic spherically symmetric collapse of a scalar field, are $C^0$-inextendible \emph{even if one leaves the class of spherically symmetric Lorentzian manifolds}.

\item As already mentioned in the introduction, timelike geodesic completeness implies, in a straightforward way, the $C^2$-inextendibility of the Lorentzian manifold. However, this method of proof does not give $C^0$-inextendibility. An interesting question is whether there are examples of (timelike) geodesically complete Lorentzian manifolds which are $C^0$-extendible. 

\end{enumerate}

\subsection{Outline of the paper}

Section \ref{Defs} collects the notions in causality theory we are using in this paper and presents a few basic, but important results on causality theory for continuous metrics. As a warm-up, we begin by proving the $C^0$-inextendibility of the Minkowski spacetime in Section \ref{SecMink}, before we introduce the Schwarzschild spacetime in Section \ref{SecSchw} and give the first half of the proof of its $C^0$-inextendibility.  Thereafter, we introduce the notion of the spacelike diameter in Section \ref{SecDiam} and give a sufficient criterion for it being finite. This notion is crucial to the second half of the proof of the $C^0$-inextendibility of the maximal analytic Schwarzschild spacetime, which is presented in  Section \ref{SecPf}.

\section{Definitions and aspects of causality theory for Lorentzian manifolds with continuous metrics}
\label{Defs}

In this section, we compile the definitions of concepts in causality theory used in this paper. Moreover, we extend a few standard results in causality theory for a more regular Lorentzian metric to the case of a merely continuous Lorentzian metric.

All manifolds considered in this paper are Hausdorff, second countable, and of dimension $d+1 \geq 2$. Moreover, note that for $M$ to carry a continuous Lorentzian metric, we need to assume that $M$ is at least endowed with a $C^1$ differentiable structure. This, however, implies that one can find a compatible smooth differentiable structure on $M$.\footnote{See for instance \cite{Hirsch12}. In fact, one can even extend any smooth structure (which is compatible with the $C^1$ structure) on a subset of $M$ to a smooth structure (compatible with the $C^1$ structure) on all of $M$.}   Hence, we will assume that all manifolds in this paper are smooth. This is for convenience only and not actually needed anywhere.

\subsection{The timelike future remains open and consequences thereof}

Let $(M,g)$ be a Lorentzian manifold with a continuous metric. Recall that a tangent vector $X \in T_pM$ is called \emph{timelike, null, spacelike} if, and only if,  $g(X,X) < 0$, $g(X,X) = 0$, $g(X,X) >0$, respectively. The set of all timelike vectors in $T_pM$ forms a double cone and, thus, has two connectedness components. One says that $(M,g)$ is \emph{time orientable} if, and only if, one can find a continuous timelike vector field on $M$. A choice of timelike vector field singles out one of the connectedness components and, thus, determines a \emph{time orientation}.

\begin{definition}
Let $(M,g)$ be a Lorentzian manifold with a continuous metric. A piecewise smooth curve $\gamma : I \to M$, where $I \subseteq \R$ is connected, is called a \emph{timelike curve} if, and only if, for all $s \in I$, where $\gamma$ is differentiable, we have $\dot{\gamma}(s)$ is timelike, and at each point of $I$, where the right-sided and left-sided derivative do not coincide, they are still both timelike and lie in the same connectedness component of the timelike double cone in the tangent space. 
\end{definition}

\begin{definition}
Let $(M,g)$ be a time-oriented Lorentzian manifold with a continuous metric.
\begin{enumerate}
\item A timelike curve $\gamma : I \to M$ is called \emph{future (past) directed} if, and only if, its (one-sided) tangent vector is future (past) directed at some point of $I$ (and hence at all points of $I$).
\item For two points $p,q \in M$ we define $p \ll q$ ($p \gg q$) to mean that there exists a future (past) directed timelike curve from $p$ to $q$.
\item For a point $p \in M$, we define the \emph{timelike future $I^+(p, M)$ of $p$ in $M$} by $\{ q \in M \; | \; p \ll q \}$. The \emph{timelike past $I^-(p, M)$ of $p$} is defined analogously.
\end{enumerate}
\end{definition}

\begin{remark}
In the setting of the above definition, where $(M,g)$ is a time-oriented Lorentzian manifold with a continuous metric, consider an open subset $U \subseteq M$. Note that $(U,g|_U)$ is a Lorentzian manifold in its own right, and in particular the future $I^+(p,U)$ of a point $p \in U$ is defined with respect to the causality structure of the manifold $(U,g|_U)$, and \emph{not} with respect to the causality structure of the ambient manifold $M$. Hence, there might be a point $q \in U$ which can be connected to $p$ via a past directed timelike curve lying in $M$, which, however, cannot be connected to $p$ via a past directed timelike curve lying entirely in $U$.
\end{remark}

Although not needed in this paper, let us also remark that we do not make the relation $\ll$ smaller by imposing that a timelike curve should be piecewise \emph{smooth}, since we can always smooth out a piecewise less regular timelike curve and, at the same time, preserve its causal character. However, this smoothing argument does not apply to causal curves! It is an easy exercise to write down a merely continuous Lorentzian metric in $1+1$ dimensions, which does not admit a single smooth null curve. 

Let us also mention at this point that there are results in the causality theory for smooth metrics which do not carry over to the merely continuous case. An instructive example in \cite{ChrusGra12} shows that there are continuous Lorentz metrics for which the light cones are no longer hypersurfaces. In this paper, however, the only result in causality theory for merely continuous metrics we need is that the timelike future and past remains an open set. The proof of this statement is given below.

The following basic, but important, lemma introduces a near-Minkowskian coordinate system adapted to a timelike curve.

\begin{lemma}
\label{NormalForm}
Let $(M,g)$ be a Lorentzian manifold with a continuous metric $g$, and let $\gamma : [-1,0] \to M$ be a timelike curve. After a possible reparametrisation of $\gamma$, we can find for every $\delta >0$ an open neighbourhood $U$ of $\gamma(0)$, an $\varepsilon >0$, and a coordinate chart $\varphi : U \to \ed$ such that
\begin{enumerate}
\item $\varphi\big(\gamma(0)\big) = (0,\ldots, 0)$
\item $(\varphi \circ \gamma) (s) = (s,0, \ldots, 0) $ holds for $s \in (-\varepsilon, 0]$
\item $g_{\mu \nu}(0) = m_{\mu \nu}$
\item $\big| \,g_{\mu \nu}(x) - m_{\mu \nu} \,\big| < \delta $ holds for all $x \in (-\varepsilon , \varepsilon)^{d+1}$ 
\end{enumerate}
are satisfied, where
\begin{equation*}
m_{\mu \nu} = 
\begin{pmatrix}
-1 &    &            & 0      \\
     & 1 &            &         \\
     &    & \ddots &         \\
 0  &    &            & 1
\end{pmatrix}
\end{equation*}
is the Minkowski metric  on $\R^{d+1}$.
\end{lemma}

\begin{proof}
Possibly after a linear change of parameter for $\gamma$, we can assume without loss of generality, that
\begin{equation}
\label{DerivativeNorm}
g\big(\dot{\gamma}(0),\dot{\gamma}(0)\big) = -1
\end{equation}
holds. Since $\gamma$ has only finitely many discontinuities, we can find a neighbourhood $U$ of $\gamma(0)$ and a $\varepsilon >0$ such that $\gamma \big((-\varepsilon, 0]\big) \subseteq U$ and $\gamma\big|_{(-\varepsilon,0]}$ is smooth. Choosing $U$ and $\varepsilon$ smaller if necessary, we can find a coordinate chart $\varphi : U \to  (-\varepsilon , \varepsilon)^{d+1}$ such that the first two points of the lemma are satisfied\footnote{More explicitly, we can choose coordinates $\psi$ which are centred at $\gamma(0)$ and such that $\dot{\gamma}_0(0) \neq 0$. Hence, $\gamma^{-1}_0$ exists in a small enough neighbourhood and we set $F(x_0, \ux) = \big(\gamma^{-1}(x_0), \ux - \underline{\gamma}[\gamma_0^{-1}(x_0)]\big)$, which is a diffeomorphism in a small enough neighbourhood of $0$. It then follows that $(F\circ \psi)$ is a chart satisfying the first two points.}.
After a linear change of coordinates, obtained from the Gram-Schmidt orthonormalisation procedure based at the origin, where we keep $\frac{\partial}{\partial x_0}$ fixed (which is normalised by \eqref{DerivativeNorm}), we can moreover arrange that, in addition, the third point holds.  For given $\delta >0$, we can now choose $U$ and $\varepsilon$ even smaller such that the fourth point of the lemma holds as well. This follows from the continuity of the Lorentz metric $g$.
\end{proof}

\begin{proposition}
\label{FutureOpen}
Let $(M,g)$ be a Lorentzian manifold with a continuous metric $g$. For all $p \in M$ the timelike future $I^+(p,M)$ and the timelike past $I^-(p,M)$ of $p$ are open in $M$.
\end{proposition}

\begin{proof}
We show that $I^+(p,M)$ is open in $M$; the proof of $I^-(p,M)$ being open is analogous. So let $q \in I^+(p,M)$ and let $\gamma : [-1, 0] \to M$ be a future directed timelike curve from $p$ to $q$. By Lemma \ref{NormalForm}, we can find, after a possible reparametrisation of $\gamma$, a chart $\varphi : U \to \ed$ such that 1.\ - 4.\ of Lemma \ref{NormalForm} are satisfied with $\delta = \frac{1}{2}$. It follows that there exists a $\tau >0$ such that 
\begin{equation*}
\frac{\partial}{\partial x_0} + \tau_1 \frac{\partial}{\partial x_1} + \tau_2 \frac{\partial}{\partial x_2} + \ldots + \tau_d \frac{\partial}{\partial x_d}
\end{equation*}
is timelike for $|\tau_i| < \tau$, $1 \leq i \leq d$. Thus, there exists a ball $B_\rho(0)$ such that every point $x \in B_\rho(0)$ can be connected to $(-\frac{\varepsilon}{2}, 0, \ldots, 0)$ by a straight line which is timelike. Concatenating $\gamma|_{[-1, -\frac{\varepsilon}{2}]}$ with the image curve of the straight line under $\varphi$ shows that $\varphi^{-1}\big(B_\rho(0)\big) \subseteq I^+(p,M)$.
\begin{center}
\def\svgwidth{7cm}
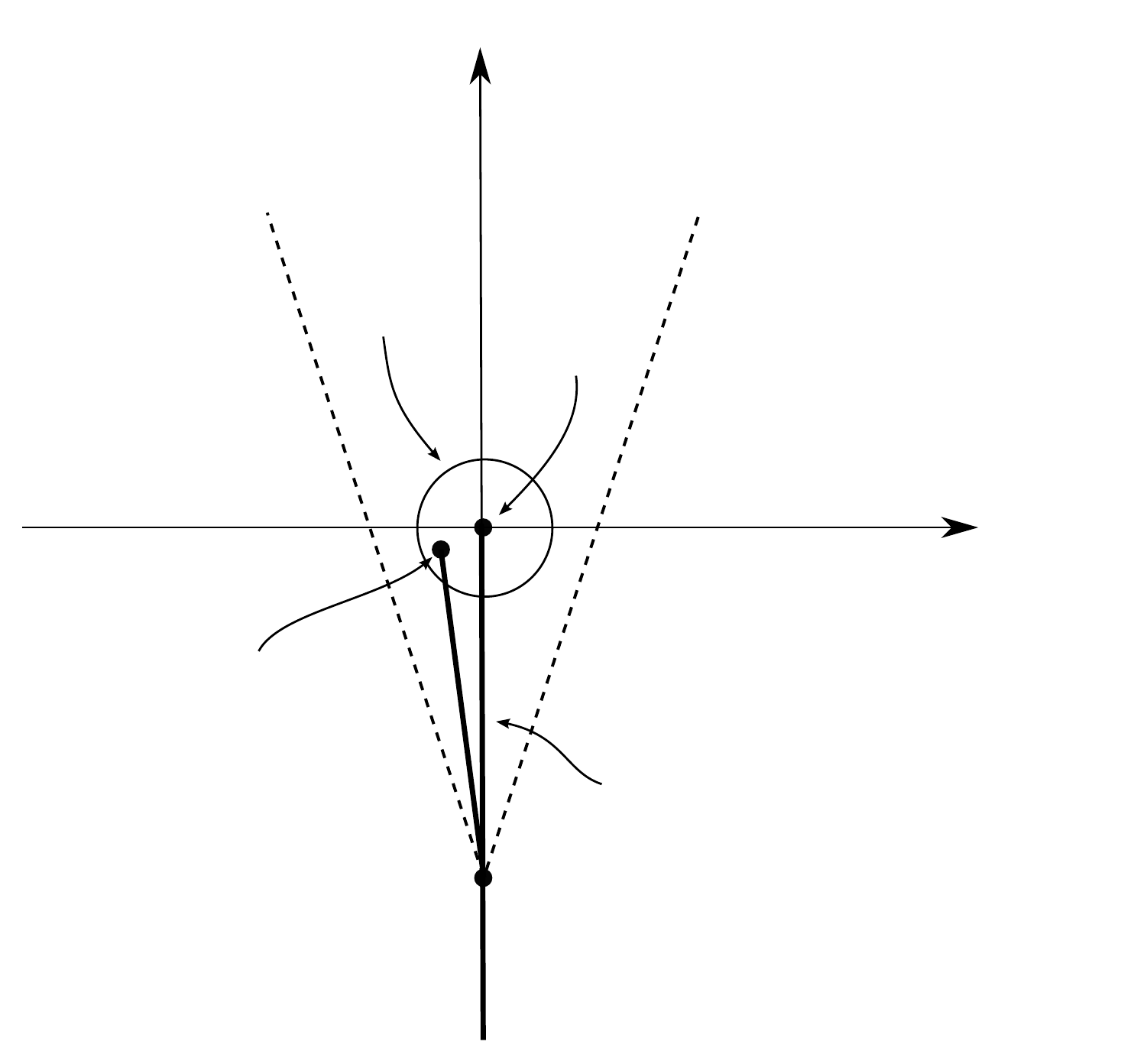
\end{center}
\end{proof}

The next two propositions are easy consequences of Proposition \ref{FutureOpen}.

\begin{proposition}
\label{UnionPasts}
Let $(M,g)$ be a Lorentzian manifold with a continuous metric $g$. Let $\gamma : [0,1] \to M$ be a future directed timelike curve. Then the following holds:
\begin{equation*}
I^-\big(\gamma(1),M\big) = \bigcup_{0 \leq s <1} I^-\big(\gamma(s),M\big) \;.
\end{equation*}
\end{proposition}

\begin{proof}
The inclusion ``$\,\supseteq \,$'' is clear. In order to prove ``$\, \subseteq \,$'', let $q \in I^-\big(\gamma(1),M\big)$. Hence, it follows that $\gamma(1) \in I^+(q,M)$.  By Proposition \ref{FutureOpen}, $I^+(q,M)$ is open and thus contains $\gamma(s)$ for $s$ close enough to $1$. This shows that $q \in I^-\big(\gamma(s),M\big)$ for $s$ close enough to $1$.
\end{proof}

\begin{proposition}
\label{FlowChart}
Let $(M,g)$ be a time oriented Lorentzian manifold with a continuous metric $g$, and let $T$ be a smooth and globally timelike vector field on $M$. 

For every $p \in M$ there exists a chart $\varphi : U \to (-\Delta, \Delta) \times (-E, E)^d$ centred at $p$, where $U$ is an open neighbourhood of $p$ in $M$, and $\Delta, E >0$, such that
\begin{enumerate}[i)]
\item $T$ has the coordinate representation $\frac{\partial}{\partial x_0}$ in the chart $\varphi$
\item every two orbits of $T$ in the chart $\varphi$ can be connected by a future as well as a past directed timelike curve.
\end{enumerate}
We call such a chart $\varphi$, which satisfies i) a \emph{flow chart}, and if it satisfies i) and ii), we call it a \emph{flow chart with timelike connected orbits}.
\end{proposition}

\begin{proof}
Since $T$ is a regular vector field, the first part of the proposition is a standard result - see for example Theorem 17.13 in \cite{LeeSmooth}. In order to prove the second statement, recall from Proposition \ref{FutureOpen} that $I^+\big(0, (-\Delta, \Delta) \times (-E, E)^d\big)$ and $I^-\big(0,  (-\Delta, \Delta) \times (-E, E)^d\big)$ are open. Thus, we can choose $E>0$ smaller such that, firstly, for every $\ux \in  (-E, E)^d$ we can find an $x_0^+ \in [0,\Delta)$ such that $(x_0^+, \ux) \in I^+\big(0, (-\Delta, \Delta) \times (-E, E)^d\big)$, and secondly, we can find an $x_0^- \in (-\Delta ,0]$ such that $(x_0^-, \ux) \in I^-\big(0,  (-\Delta, \Delta) \times (-E, E)^d\big)$. 
\begin{center}
\def\svgwidth{7cm}
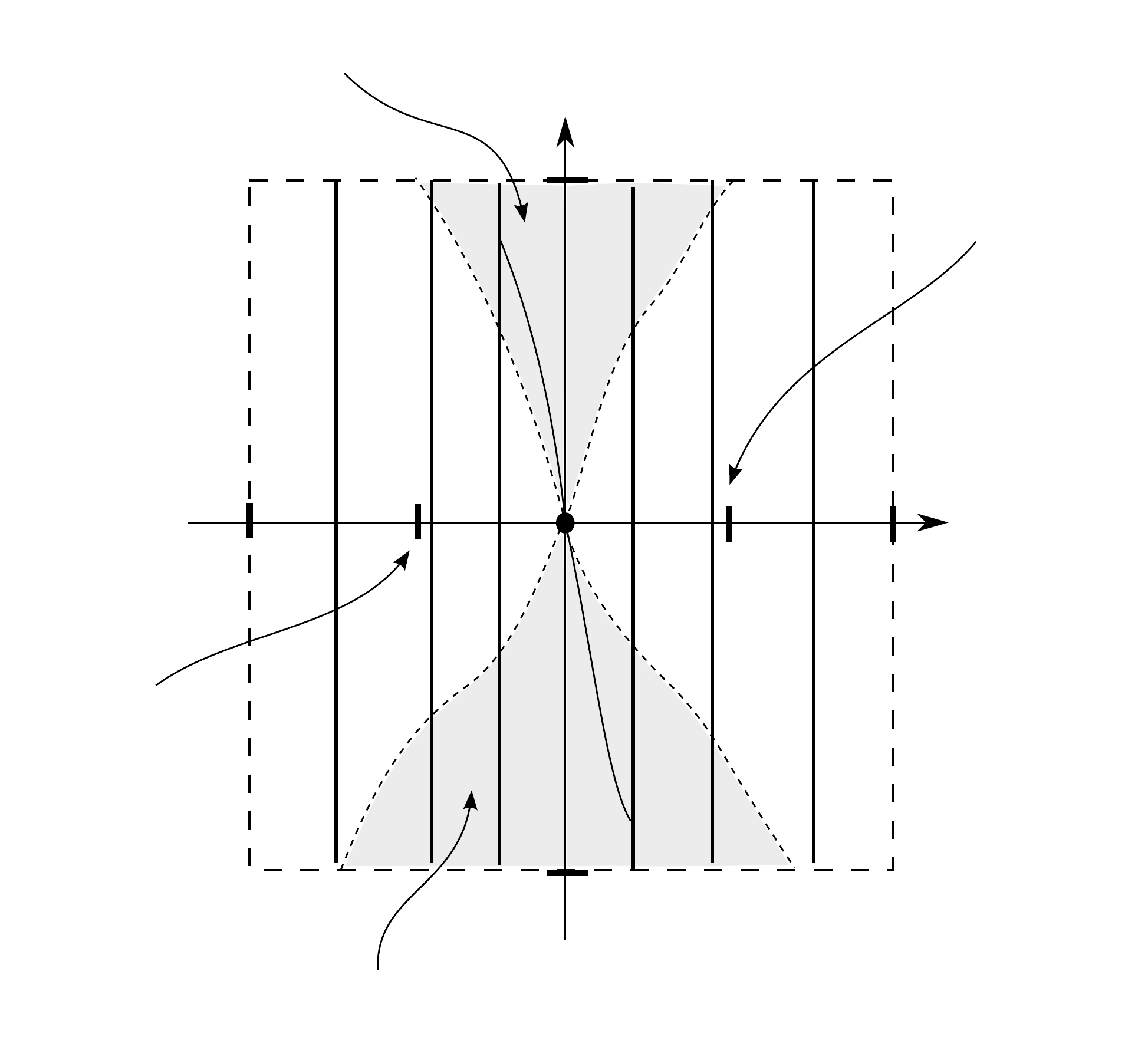
\end{center} 
Hence, any two orbits of $T$ in the chart $\varphi$ can be connected (via $0$) by a future as well as a past directed timelike curve. 
\end{proof}

\subsection{Definitions}

This section lays down more terminology used in this paper.

\begin{definition}
Let $(M,g)$ be a time-oriented Lorentzian manifold with a continuous metric $g$. 
\begin{enumerate}
\item We define the \emph{Lorentzian distance function} $d(\cdot, \cdot) : M \times M \to [0,\infty]$ by
\begin{equation*}
d(p,q) = \begin{cases} 0 &\textnormal{ if } q \notin I^+(p,M) \\ \sup\limits_{\substack{\sigma : [0,1] \to M \textnormal{ future directed} \\ \textnormal{with } \sigma(0) = p \textnormal{ and } \sigma(1) =q}} \Big{\{} \int_0^1 \sqrt{-g\big(\dot{\sigma}(s), \dot{\sigma}(s)\big)} \, ds \Big{\}} \qquad &\textnormal{ if } q \in I^+(p,M) \end{cases}
\end{equation*}
Here, $L(\sigma) := \int_0^1 \sqrt{-g\big(\dot{\sigma}(s), \dot{\sigma}(s)\big)} \, ds$ is also called the \emph{Lorentzian length of the timelike curve $\sigma$}.
\item We define the \emph{timelike diameter of $M$} by $\diam_t(M) := \sup\{d(p,q)\;|\; p,q \in M\}$.
\end{enumerate}
\end{definition}

\begin{definition}
Let $(M,g)$ be a time-oriented Lorentzian manifold with a continuous metric, and let $I \subseteq \R$ be one of the following intervals $(a,b)$, $(a,b]$, $[a,b)$, or $[a,b]$, where $a <b$.  A future directed timelike curve $\gamma : I \to M$ is said to be \emph{future (past) extendible} if, and only if, $\gamma$ can be extended to $I \cup \{b\}$ ($I \cup \{a\}$) as a continuous curve. Otherwise, $\gamma$ is called \emph{future (past) inextendible}. Moreover, we call $\gamma$ inextendible if, and only if, $\gamma$ is future and past inextendible.
\end{definition}

\begin{remark}
Note that the notion of future (past) extendibility of timelike curves only requires the extendibility of the timelike curve as a \emph{continuous} curve. It thus may well happen that a timelike curve is future extendible, but we can not extend it to the future as a \emph{timelike} curve.
\end{remark}

\begin{definition} 
Let $(M,g)$ be a time-oriented Lorentzian manifold with a continuous metric. 
\begin{enumerate}
\item
We call a smooth embedded hypersurface $\Sigma$ of $M$ a \emph{Cauchy hypersurface} in $(M,g)$ if, and only if, $\Sigma$ is met exactly once by every inextendible timelike curve.
\item 
We say that a Lorentzian manifold $(M,g)$ is \emph{globally hyperbolic} if, and only if, there exists a Cauchy hypersurface $\Sigma$ in $(M,g)$.
\end{enumerate}
\end{definition}

\begin{definition}
Let $(M,g)$ be a time-oriented Lorentzian manifold with a continuous metric. 
\begin{enumerate}
\item Two future directed timelike curves $\gamma_i : [0,1] \to M$, $i = 0,1$, with $\gamma_0(0) = \gamma_1(0)$ and $\gamma_0(1) = \gamma_1(1)$ are called \emph{timelike homotopic with fixed endpoints} if, and only if, there exists a continuous map $\Gamma : [0,1] \times [0,1] \to M$ such that $\Gamma(t, \cdot)$ is a future directed timelike curve from $\gamma_0(0)$ to $\gamma_0(1)$ for all $t \in [0,1]$  and, moreover, $\Gamma(0, \cdot) = \gamma_0(\cdot)$ and $\Gamma(1,\cdot) = \gamma_1(\cdot)$. The map $\Gamma$ is also called a \emph{timelike homotopy with fixed endpoints between $\gamma_0$ and $\gamma_1$}.
\item We say that $(M,g)$ is \emph{future one-connected} if, and only if, for all $p, q \in M$, any two future directed timelike curves from $p$ to $q$ are timelike homotopic with fixed endpoints.
\end{enumerate}
\end{definition}

The following concept plays an important role in the proof of Theorem \ref{Inex}.

\begin{definition}
Given two sets $A,B \subseteq M$, we say that $A$ and $B$ are \emph{timelike separated by a set $K \subseteq M$} if, and only if, every timelike curve connecting $A$ and $B$ intersects $K$ - i.e., for any timelike curve $\sigma : [0,1] \to M$ with $\sigma(0) \in A$ and $\sigma(1) \in B$ there exists an $s_0 \in [0,1]$ with $\sigma(s_0) \in K$.
\end{definition}
If $K$ is a closed set, then $M \setminus K$ is again a Lorentzian manifold. In this case, we clearly have that $A$ and $B$ are timelike seperated by $K$ if, and only if, $I^+(A,M\setminus K) \cap B = \emptyset$ and $I^-(A, M \setminus K) \cap B = \emptyset$.

\subsection{Extensions of Lorentzian manifolds}

In the following we will tacitly assume that all manifolds are connected.

\begin{definition}
Let $(M,g)$ be a time-oriented Lorentzian manifold with a smooth metric $g$.
\begin{enumerate}
\item Let $k \in \N$. A \emph{$C^k$-extension of $(M,g)$} is a smooth isometric embedding $\iota : M \hookrightarrow \tilde{M}$ of $M$ into a proper subset\footnote{Since we have tacitly assumed that all manifolds are connected, this implies $\overline{\iota(M)} \setminus \iota(M) \neq \emptyset$.} of a Lorentzian manifold $(\tilde{M}, \tilde{g})$, where $\tilde{M}$ is of the same dimension as $M$, and $\tilde{g}$ is a $C^k$-regular metric. 

By slight abuse of terminology, we sometimes also call $\tilde{M}$ the extension of $M$.
\item The Lorentzian manifold $(M,g)$ is called \emph{$C^k$-extendible} if, and only if, there exists a $C^k$-extension of $(M,g)$. Otherwise, $(M,g)$ is called \emph{$C^k$-inextendible}.
\end{enumerate}
\end{definition}

\begin{remark}
\begin{enumerate}
\item Of course, the question of extendibility is also of physical interest for Lorentzian manifolds $(M,g)$ which do not have a smooth metric $g$. However, in this paper, the Lorentzian manifolds under consideration do have a smooth metric.

\item Note that we did not require in the definition of `extendibility' that $(\tilde{M},\tilde{g})$ is also time-orientable. However, this notion of `extendibility' is \emph{not} more general than if one imposed the condition of time-orientability on $(\tilde{M},\tilde{g})$. To see this, recall that every Lorentzian manifold has a time-orientable double cover\footnote{See for example 17 Lemma in Chapter 7 of \cite{ONeill}.}, to which the isometric embedding can be lifted. However, we will not make use of this fact in this paper.

\item The notion of a Lorentzian manifold being extendible is global by nature. Although in Section 3.1 of \cite{HawkEllis} a definition to capture the idea of a `local extension' was suggested, it was shown by Beem in \cite{Beem80} that it suffers from the shortcoming that the Minkowski spacetime then has to be considered as locally extendible. 
\end{enumerate}
\end{remark}

The next Lemma shows that, given an extension, it is always possible find a timelike curve which leaves the original Lorentzian manifold.

\begin{lemma}
\label{ExtensionLemma}
Let $(M,g)$ be a time-oriented Lorentzian manifold with a smooth metric, let $k \in \N$, and let $\iota : M \hookrightarrow \tilde{M}$ be a $C^k$-extension of $(M,g)$. Then, there exists a timelike curve $\tilde{\gamma} : [0,1] \to \tilde{M}$ such that $\tilde{\gamma}\big([0,1)\big) \subseteq \iota(M)$ and $\tilde{\gamma}(1) \in \tilde{M}\setminus \iota(M)$.
\end{lemma}

\begin{proof}
Since $\iota(M)$ is a proper subset of the connected manifold $\tilde{M}$, its boundary $\partial\big(\iota(M)\big)$ is non-empty. Let $\tilde{p} \in \partial\big(\iota(M)\big)$. We can find a small neighbourhood $\tilde{U}$ of $\tilde{p}$ that is time-oriented. Let $\tilde{q} \in I^-(\tilde{p}, \tilde{U})$. We distinguish two cases:
\begin{enumerate}
\item  $\tilde{q} \in \iota(M)$: There exists a timelike curve $\tilde{\gamma} : [0,1] \to \tilde{U}$ with $\tilde{\gamma}(0) = \tilde{q}$ and $\tilde{\gamma}(1) = \tilde{p}$. Let
\begin{equation*}
s_0 := \sup \{ s \in [0,1] \, | \, \tilde{\gamma}\big([0,s)\big) \subseteq \iota(M) \} \;.
\end{equation*}
Since $\iota(M)$ is open in $\tilde{M}$, it follows that $\tilde{\gamma}(s_0) \in \tilde{M} \setminus \iota(M)$. Reparametrising $\tilde{\gamma}|_{[0,s_0]}$ then gives the timelike curve from the statement of the lemma.
\item $\tilde{q} \in \tilde{M}\setminus \iota(M)$: Since $I^+(\tilde{q},\tilde{U})$ is open by Proposition \ref{FutureOpen} and contains $\tilde{p} \in \partial\big(\iota(M)\big)$, it must also contain a point $\tilde{r} \in \iota(M)$. One now considers a timelike curve $\tilde{\gamma} : [0,1] \to \tilde{U}$ with $\tilde{\gamma}(0) = \tilde{r}$ and $\tilde{\gamma}(1) = \tilde{q}$ and proceeds in analogy to the previous case.
\end{enumerate}
\end{proof}

The importance of this lemma stems from the fact that it allows us to `locate' the extension. By this we mean that we can choose a small neighbourhood $\tilde{U}$ of $\tilde{\gamma}(1)$, which then gives us a region $\iota^{-1}(\tilde{U}) $ in the original manifold $M$ through which one extends.  One can then start using the geometry of this region in $M$ to show that no such extension is possible. In particular, the lemma entails that if no timelike curve can leave $M$, then $M$ must be inextendible.

The proof of the $C^0$-inextendibility of the Schwarzschild spacetime presented in this paper is by contradiction. Assuming that there is a $C^0$-extension, the above lemma implies that there must be a timelike curve leaving the Schwarzschild spacetime. It then follows, that such a timelike curve must either leave through the curvature singularity in the interior, or through timelike or null infinity in the exterior. Having located the possible extensions of the Schwarzschild spacetime, we then show that the timelike diameter being infinite in the exterior is a $C^0$-obstruction to extensions (however, one could also use the spacetime volume in a straightforward way as an obstruction), while in the interior, we show that the spacelike diameter being infinite forbids $C^0$-extensions. 

As an instructive introduction to this scheme, we begin by proving the $C^0$-inextendibility of the Minkowski spacetime in the next section. In this case, we only need to capture the obstruction to $C^0$-extensions coming from the infiniteness of the timelike diameter of the Minkowski spacetime.

\section{The $C^0$-inextendibility of the Minkowski spacetime}
\label{SecMink}

Let $d \in \N_{\geq 1}$. The $(d+1)$-dimensional Minkowski spacetime $(\Mmink, m)$ is given by the smooth manifold $\Mmink = \R^{d+1}$ together with the smooth Lorentzian metric $m = -\,dx_0^2 + \, dx_1^2 + \ldots + \, dx_d^2$, where $(x_0, x_1, \ldots, x_d)$ are the canonical coordinates on $\R^{d+1}$. The time orientation is fixed by stipulating that $\frac{\partial}{\partial x_0}$ is future directed.

A standard result in Lorentzian geometry is the $C^2$-inextendibility of the Minkowski spacetime. This follows from combining the fact, that for $C^2$-extensions, one can always find geodesics (even timelike ones) that leave the original Lorentzian manifold, together with the geodesic completeness of $(\Mmink, m)$.

%The $C^1$ incompleteness can be deduced from the bounded acceleration completeness of the Minkowski spacetime.

Here, we establish the following 

\begin{theorem}
\label{MinkInex}
For every $d \geq 1$, the $(d+1)$-dimensional Minkowski spacetime $(\Mmink, m)$ is $C^0$-inextendible.
\end{theorem}

\begin{proof}
The proof is by contradiction and proceeds in three steps. 
\vspace*{2mm}

\underline{\textbf{Step 1:}}
We assume that there exists a Lorentzian manifold $(\tilde{M}, \tilde{g})$ with a continuous metric $\tilde{g}$ and an isometric embedding $\iota : \Mmink \hookrightarrow \tilde{M}$ such that $\iota(\Mmink)$ is a proper subset of $\tilde{M}$. By Lemma \ref{ExtensionLemma}, we can then find a timelike curve $\tilde{\gamma} : [-1,0] \to \tilde{M}$ such that $\gamma := \iota^{-1} \circ \tilde{\gamma}|_{[-1,0)} : [-1,0) \to \Mmink$ is a timelike curve in $\Mmink$ and $\tilde{\gamma}(0) \in \tilde{M} \setminus \iota(\Mmink)$. Without loss of generality we can assume that $\gamma$ is future directed (otherwise reverse the time orientation), and hence it is future inextendible in $\Mmink$. 
\vspace*{2mm}

\textbf{Step 1.1:}
By Lemma \ref{NormalForm}, after a possible reparametrisation of $\tilde{\gamma}$, there is an $\varepsilon >0$, an open neighbourhood $\tilde{U} \subseteq \tilde{M}$ of $\tilde{\gamma}(0)$, and a chart $\tilde{\varphi} : \tilde{U} \to \ed$ such that
\begin{enumerate}
\item $(\tilde{\varphi} \circ \tilde{\gamma}) (s) = (s,0, \ldots, 0) $ holds for $s \in (-\varepsilon, 0]$
\item $\big| \,\tilde{g}_{\mu \nu}(x) - m_{\mu \nu} \,\big| < \delta $ holds for all $x \in (-\varepsilon , \varepsilon)^{d+1}$\;,
\end{enumerate}
where $\delta >0$ is small and to be fixed in the following.

Let $ 0 < a < 1$ and let $< \cdot, \cdot>_{R^{d+1}}$ denote the Euclidean inner product on $\R^{d+1}$ and $| \cdot |_{\R^{d+1}}$ the associated norm. We introduce the following notation: 
\begin{itemize}
\item $C^+_a := \big{\{} X \in \R^{d+1} \, | \, \frac{<X,e_0>_{\R^{d+1}}}{|X|_{\R^{d+1}}}   > a \big{\}}$
\item $C^-_a := \big{\{} X \in \R^{d+1} \, | \, \frac{<X,e_0>_{\R^{d+1}}}{|X|_{\R^{d+1}}}   < -a \big{\}}$
\item $C^c_a := \big{\{} X \in \R^{d+1} \, | \, -a < \frac{<X,e_0>_{\R^{d+1}}}{|X|_{\R^{d+1}}}   < a \big{\}}$\;.
\end{itemize}
Here, $C^+_a$ is the forward cone of vectors which form an angle of less than $\cos^{-1}(a)$ with the $x_0$-axis, and $C^-_a$ is the corresponding backwards cone. In Minkowski space, the forward and backward cones of timelike vectors correspond to the value $a = \cos(\frac{\pi}{4}) = \frac{1}{\sqrt{2}}$.

Since $\frac{5}{8} < \frac{1}{\sqrt{2}} < \frac{5}{6}$, we can now choose $\delta >0$ such that in the chart $\tilde{\varphi}$ from above all vectors in $C^+_{\nicefrac{5}{6}}$ are future directed timelike, all vectors in $C^-_{\nicefrac{5}{6}}$ are past directed timelike, and all vectors in $C^c_{\nicefrac{5}{8}}$ are spacelike.
\vspace*{2mm}

\textbf{Step 1.2:} We show that for $x \in \ed$ we have the following inclusion relations
\begin{equation}
\label{UpperBoundEstimatesOnPastAndFuture}
\begin{split}
& I^+(x, (-\varepsilon, \varepsilon)^{d+1}) \subseteq \big( x + C^+_{\nicefrac{5}{8}}\big) \cap  \ed \\
& I^-(x, (-\varepsilon, \varepsilon)^{d+1}) \subseteq \big( x + C^-_{\nicefrac{5}{8}}\big) \cap  \ed  \;.
\end{split}
\end{equation}
\vspace*{2mm}

We only prove the first inclusion relation of \eqref{UpperBoundEstimatesOnPastAndFuture}, the second follows by reversing the time orientation. So let $\sigma : [0,L] \to \ed$ be a future directed timelike curve with $\sigma(0) = x$, and we impose that $\sigma$ is parametrised by arc-length with respect to the Euclidean metric on $\R^{d+1}$. Here, $L = L_{\mathrm{Euclidean}}(\sigma) >0$ is the Euclidean length of the curve $\sigma$. Let us first assume that $\sigma$ is smooth. Since $\sigma$ is timelike and future directed, it follows that $\dot{\sigma}(s) \in  C^+_{\nicefrac{5}{8}}$ for all $s \in [0,L]$.

We then compute 
\begin{equation}
\label{RemainsInCone}
\frac{<\sigma(L) - x, e_0>_{\R^{d+1}}}{|\sigma(L) - x|_{\R^{d+1}}} = \frac{\int_0^L<\dot{\sigma}(s'), e_0>_{\R^{d+1}}\,ds'}{|\sigma(L) - x|_{\R^{d+1}}} > \frac{5}{8} \cdot \frac{L}{|\sigma(L) - x|_{\R^{d+1}}} \geq \frac{5}{8} \;,
\end{equation}
where we have used that $L \geq |\sigma(L) - x|_{\R^{d+1}}$. It now follows that $\sigma(L)  \in  x + C^+_{\nicefrac{5}{8}}$. If $\sigma$ is only piecewise smooth, one splits the integral in \eqref{RemainsInCone} into a sum of integrals over the smooth segments of $\sigma$. This proves \eqref{UpperBoundEstimatesOnPastAndFuture}.
\vspace*{2mm}

It now follows from \eqref{UpperBoundEstimatesOnPastAndFuture} that we can choose $s_0 \in (0, \varepsilon)$ such that the closure of 
\begin{equation*}
I^-\big(0, \ed\big) \cap I^+\big((-s_0, 0, \ldots, 0), \ed\big)
\end{equation*}
in $\ed$ is compact.
\vspace*{2mm}

\textbf{Step 1.3:} We show that the timelike diameter of  $I^-\big(0, \ed\big) \cap I^+\big((-s_0, 0, \ldots, 0), \ed\big)$ is bounded. 
\vspace*{2mm}

Clearly, we have 
\begin{equation*}
\diam_t\Big(I^-\big(0, \ed\big) \cap I^+\big((-s_0, 0, \ldots, 0), \ed\big)\Big) = d_{\ed}\big((-s_0, 0, \ldots, 0),0\big) \;,
\end{equation*}
where $d_{\ed}$ is the Lorentzian distance function in $\ed$.
So let $\sigma :[-s_0,0] \to \ed$ be future directed with $\sigma(-s_0) = (-s_0, 0, \ldots, 0)$ and $\sigma(0) = 0$. To simplify notation, we assume again that $\sigma$ is smooth. The general case is not more difficult. 

Since $\dot{\sigma}(s) \in C^+_{\nicefrac{5}{8}}$ for all $s \in [0,1]$, we have $dx_0\big(\dot{\sigma}(s)\big) >0$. It follows that we can reparametrise $\sigma$ so that we can assume without loss of generality that  $\sigma : [-s_0, 0] \to \ed$ is given by
\begin{equation*}
\sigma(s) = \big(s, \underline{\sigma}(s)\big)\;.
\end{equation*}
From $\dot{\sigma}(s) \in C^+_{\nicefrac{5}{8}}$ for all $s \in [-s_0, 0]$, it follows that
\begin{equation*}
\frac{5}{8} < \frac{<\dot{\sigma}(s), e_0>_{\R^{d+1}}}{|\dot{\sigma}(s)|_{\R^{d+1}}} = \frac{1}{\sqrt{1+ |\dot{\underline{\sigma}}(s)|_{\R^d}}} \;.
\end{equation*}
Hence, we obtain $|\dot{\underline{\sigma}}(s)|_{\R^d} < \frac{\sqrt{39}}{5} $ for all $s \in [-s_0,0]$.
It now follows that
\begin{equation*}
\int\limits_{-s_0}^0 \sqrt{-\tilde{g}\big(\dot{\sigma}(s), \dot{\sigma}(s)\big)} \,ds = \int\limits_{-s_0}^0 \sqrt{ - \Big[ \tilde{g}_{00} + 2\sum_{i=1}^d \tilde{g}_{0i} \dot{\underline{\sigma}}_i(s) + \sum_{i,j =1}^d \tilde{g}_{ij} \dot{\underline{\sigma}}_i(s) \dot{\underline{\sigma}}_j(s) \Big]}\, ds
\end{equation*}
is bounded by a constant $C_{td} >0$ which is independent of $\sigma$, since every term in the integral can be bounded by a constant independently of $\sigma$.
\vspace*{2mm}

\underline{\textbf{Step 2:}} We show that for every $C>0$ there exists an $s_1 \in (0, s_0)$ such that, in $\Mmink$, $\gamma|_{[-s_0, -s_1]}$ is timelike homotopic with fixed endpoints to a timelike curve of length greater than $C$.
\vspace*{2mm}

The proof of Step 2 requires two ingredients.
\vspace*{2mm}

\textbf{Step 2.1:} Any future directed timelike curve $\lambda : [0,1] \to \Mmink$  is timelike homotopic with fixed endpoints to the unique Lorentzian length maximising timelike geodesic from $\lambda(0)$ to $\lambda(1)$.  
\vspace*{2mm}

Recall that for $p, q \in \Mmink$ with $p \ll q $, the unique Lorentzian length maximising timelike geodesic from $p$ to $q$ is given by the straight line connecting $p$ with $q$. Also recall that the concatenation $\lambda_0 * \lambda_1 : [0,1] \to \Mmink$ of two curves $\lambda_0 : [0,1] \to \Mmink$ and $\lambda_1 : [0,1] \to \Mmink$ with $\lambda_0(1) = \lambda_1(0)$ is given by
\begin{equation*}
\lambda_0 * \lambda_1 (s) = \begin{cases} \lambda_0(2s) &\textnormal{ for } 0 \leq s \leq \frac{1}{2} \\ \lambda_1(2s -1) \; &\textnormal{ for } \frac{1}{2} \leq s \leq 1 \;.\end{cases} 
\end{equation*}

Without loss of generality we can assume that $\lambda(0) =0$. For $X \in T_0\Mmink \approx \Mmink$, we denote with $\sigma_X : [0,1] \to \Mmink$ the geodesic given by $\sigma_X (s) = s \cdot X$. A timelike homotopy $\Gamma : [0,1] \times [0,1] \to \Mmink$ with fixed endpoints between $\lambda$ and the Lorentzian length maximising geodesic $\sigma_{\lambda(1)}$ is given by
\begin{equation*}
\Gamma(t,s) = \big( \sigma_{\lambda(t)} * \lambda|_{[t,1]}\big) (s) \;,
\end{equation*}
where $\lambda|_{[t,1]}$ is understood to be reparametrised to the interval $[0,1]$ (for example by a linear rescaling).
\begin{center}
\def\svgwidth{7.5cm}
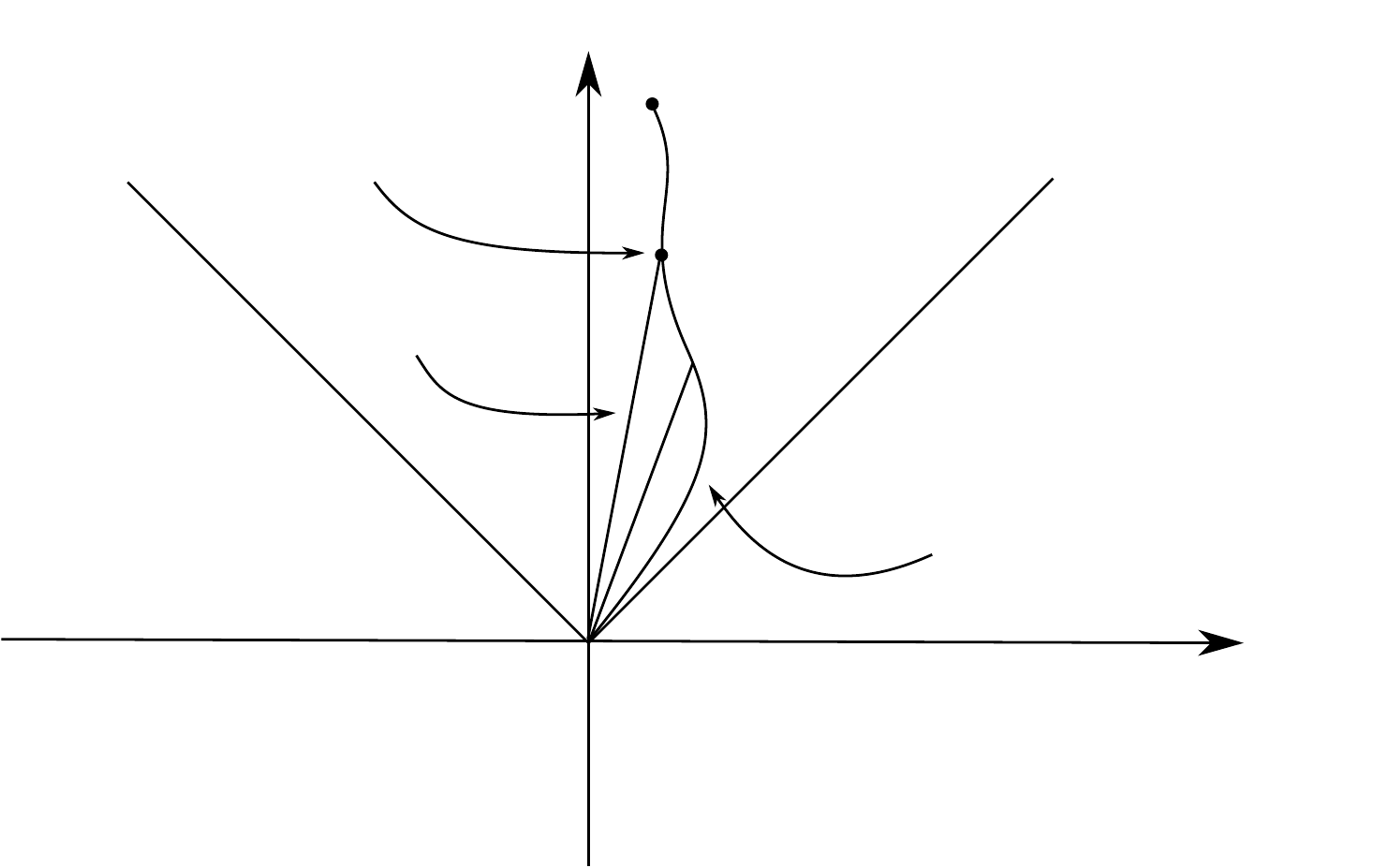
\end{center}
\vspace*{2mm}

\textbf{Step 2.2:} Let $\lambda : [0, \infty) \to \Mmink$ be a future directed and future inextendible timelike curve. Then $d\big(\lambda(0), \lambda(s)\big) \to \infty$ for $s \to \infty$.
\vspace*{2mm}

After a translation of the coordinate system we can again assume that $\lambda(0) = 0$, so that $d\big(\lambda(0), \lambda(s)\big)$ is given by $\sqrt{- m\big(\lambda(s), \lambda(s)\big)}$. Moreover, after an additional Lorentz transformation, we can choose the coordinates on $\Mmink$ such that $\gamma(1) = (\Delta, 0, \ldots, 0)$ for some $\Delta >0$. The future inextendibility of $\lambda = (\lambda_0, \lambda_1, \ldots, \lambda_d)$ implies that
\begin{equation}
\label{FutureInex}
\lambda_0(s) \to \infty \; \textnormal{ for } s \to \infty \;.
\end{equation}
Moreover, since $\lambda$ is a timelike curve, we obtain
\begin{equation}
\label{Timelike}
0 \geq m\big(\lambda(s) - \lambda(1), \lambda(s) - \lambda(1)\big) = m\big(\lambda(s), \lambda(s)\big) - 2 m\big(\lambda(s), \lambda(1)\big) + m\big(\lambda(1), \lambda(1)\big) \;.
\end{equation}
From \eqref{Timelike}, together with \eqref{FutureInex}, it now follows that
\begin{equation*}
-m\big(\lambda(s), \lambda(s)\big) \geq -2 m \big(\lambda(s), \lambda(1)\big) + m\big(\lambda(1), \lambda(1)\big) \to \infty \textnormal{ for } s \to \infty \;,
\end{equation*} 
which proves Step 2.2.
\vspace*{2mm}

We now finish the proof of Step 2. Let $C>0$ be given. By Step 2.2, we can find an $s_1 \in (0, s_0)$ such that $d_{\Mmink}\big(\gamma(-s_0), \gamma(-s_1)\big) > C$, where $d_{\Mmink}$ is the Lorentzian distance function in $(\Mmink, m)$. By Step 2.1, we can find a timelike homotopy $\Gamma : [0,1] \times [0,1] \to \Mmink$ with fixed endpoints between $\gamma|_{[-s_0,-s_1]}$ and the Lorentzian distance maximising geodesic from $\gamma(-s_0)$ to $\gamma(-s_1)$ of Lorentzian length $d_{\Mmink}\big(\gamma(-s_0), \gamma(-s_1)\big) > C$. 
\vspace*{2mm}

\underline{\textbf{Step 3:}} We show that the result of Step 2 contradicts the result of Step 1.3.
\vspace*{2mm}

Let us first point out that, a priori, the bound
\begin{equation}
\label{BoundDiam}
\diam_t\Big(I^-\big(0, \ed\big) \cap I^+\big((-s_0, 0, \ldots, 0), \ed\big)\Big) = d_{\ed}\big((-s_0, 0, \ldots, 0), 0\big) \leq C_{td}
\end{equation} 
from Step 1.3 does \emph{not} rule out the existence of timelike curves \emph{in $\Mmink$} from $\gamma(-s_0)$ to $\gamma(-s_1)$ of Lorentzian length bigger than $C_{td}$, where $s_1 \in (0, s_0)$. However, as we will show now, the bound \eqref{BoundDiam} rules out the existence of timelike curves in $\Mmink$ from $\gamma(-s_0)$ to $\gamma(-s_1)$ of Lorentzian length bigger than $C_{td}$, which, moreover, are timelike homotopic to $\gamma|_{[-s_0, -s_1]}$ with fixed endpoints.

By Step 2, there is an $s_1 \in (0,s_0)$ and a timelike homotopy $\Gamma : [0,1] \times [0,1] \to \Mmink$ with fixed endpoints between $\gamma|_{[-s_0,-s_1]}$ and a timelike curve from $\gamma(-s_0)$ to $\gamma(-s_1)$ of Lorentzian length greater than $C_{td}$. Then, $\iota \circ \Gamma  : [0,1] \times [0,1] \to \tilde{M}$ is a timelike homotopy with fixed endpoints between $\tilde{\gamma}|_{[-s_0,-s_1]}$ and a timelike curve from $\tilde{\gamma}(-s_0)$ to $\tilde{\gamma}(-s_1)$ of Lorentzian length greater than $C_{td}$. 

We claim that $\iota \circ \Gamma$ actually maps into $\tilde{U}$. This is seen as follows: Let $I \subset [0,1]$ denote the set of all $t \in [0,1]$ such that $\iota \circ \Gamma (t, \cdot)$ is a timelike curve in $\tilde{U}$. Clearly, $I$ is non-empty, since $\iota \circ \Gamma (0, \cdot) = \tilde{\gamma}|_{[-s_0, -s_1]}(\cdot)$ (modulo parametrisation). Moreover, the openness of $\tilde{U}$ implies the openness of $I$, and since we have chosen $s_0 \in (0, \varepsilon)$ such that the closure of $I^-\big(\tilde{\gamma}(0), \tilde{U}\big) \cap I^+\big(\tilde{\gamma}(-s_0), \tilde{U}\big)$ in $\tilde{M}$ is contained in $\tilde{U}$, it also follows that $I$ is closed in $[0,1]$. This yields $I = [0,1]$ and hence proves the claim.

Thus, we have shown that $\tilde{\varphi} \circ \iota \circ \Gamma (1, \cdot)$ is, firstly, well-defined, and secondly, is a timelike curve in $\ed$ from $(-s_0, 0, \ldots, 0)$ to $(-s_1, 0, \ldots, 0)$ of length greater than $C_{td}$, which contradicts the bound \eqref{BoundDiam} from Step 1.3. This concludes the proof of Theorem \ref{MinkInex}.
\end{proof}

\section{The $C^0$-inextendibility of the Schwarzschild spacetime}
\label{SecSchw}

\subsection{The Schwarzschild spacetime}

Let $d \in \N_{\geq 3}$, $m>0$, and $r_+ := (2m)^{\frac{1}{d-2}}$. Moreover, we define $D(r) := 1 - \frac{2m}{r^{d-2}} = 1 - \frac{r_+^{d-2}}{r^{d-2}}$.

\subsubsection*{The Schwarzschild exterior}

Consider the smooth manifold $\Mext := \R \times \big((2m)^{\frac{1}{d-2}}, \infty\big) \times \mathbb{S}^{d-1}$, where $m>0$ is a parameter. We denote with $t$ and $r$ the canonical coordinate functions on $\R$ and $\big((2m)^{\frac{1}{d-2}}, \infty\big)$, respectively. A smooth Lorentzian metric $\gext$ is given on $\Mext$  by 
\begin{equation}
\label{SchwarzschildMetricExt}
\gext = -\big(1 - \frac{2m}{r^{d-2}}\big) \,dt^2 + \big(1 - \frac{2m}{r^{d-2}}\big)^{-1}\,dr^2 + r^2\, \mathring{\gamma}_{d-1} \;,
\end{equation}
where $\mathring{\gamma}_{d-1}$ is the standard metric on the unit $(d-1)$-sphere. 
The Lorentzian manifold $(\Mext,\gext)$ is called the \emph{exterior of a $d+1$ dimensional Schwarzschild black hole with mass $m$}. 
It was introduced in 1916 in \cite{Schw16} by Schwarzschild for $d=3$, where he also showed that it is a solution to the vaccum Einstein equations $\mathrm{Ric}(g) =0$. In 1963, Tangherlini generalised Schwarzschild's metric to higher dimensions, cf.\ \cite{Tang63}.

We define a time-orientation on $(\Mext, \gext)$ by stipulating that $\frac{\partial}{\partial t}$ is future directed.

\subsubsection*{The Schwarzschild interior}

Consider the smooth manifold $\Mint := \R \times \big(0,(2m)^{\frac{1}{d-2}}\big) \times \mathbb{S}^{d-1}$, where $m>0$ is a parameter. We denote with $t$ and $r$ the canonical coordinate functions on $\R$ and $\big(0,(2m)^{\frac{1}{d-2}}\big)$, respectively. A smooth Lorentzian metric $\gint$ is given on $\Mint$  by 
\begin{equation}
\label{SchwarzschildMetricInt}
\gint = -\big(1 - \frac{2m}{r^{d-2}}\big) \,dt^2 + \big(1 - \frac{2m}{r^{d-2}}\big)^{-1}\,dr^2 + r^2\, \mathring{\gamma}_{d-1} \;.
\end{equation}
The Lorentzian manifold $(\Mint,\gint)$ is called the \emph{interior of a $d+1$ dimensional Schwarzschild black hole with mass $m$}.  We define a time-orientation on $(\Mint,\gint)$ by stipulating that $-\frac{\partial}{\partial r}$ is future directed.

\subsubsection*{The maximal analytic extension of the Schwarzschild spacetime}
\label{MaxAna}

We first discuss in some more detail the $3+1$-dimensional case. Consider the smooth manifold $\Mmax := H \times \mathbb{S}^2$, where $H := \big{\{} (u,v) \in \R^2 \, | \, uv <1\big{\}}$. We define the four regions
\begin{align*}
I &:= \{u < 0 \} \cap \{ v > 0\} \\
II &:= \{u > 0 \} \cap \{ v > 0\} \\
III &:= \{u < 0 \} \cap \{ v < 0\} \\
IV &:= \{u > 0 \} \cap \{ v < 0\} \;.
\end{align*}
In the following we will define an analytic function $r : H \to (0,\infty)$.

We begin by defining $r^* : (2m, \infty) \to (-\infty, \infty)$ by
\begin{equation}
\label{RStar}
r^*(r) = r + 2m\log\big( \frac{r}{2m} -1\big) \;.
\end{equation}
Note that $r*(r)$ satisfies $\frac{dr^*}{dr} = \frac{1}{D(r)}$. Consider now $F: (2m, \infty) \to (0, \infty)$, given by
\begin{equation}
\label{DefF}
F(r) := e^{\frac{r^*(r)}{2m}} = \big(\frac{r}{2m} -1 \big) e^{\frac{r}{2m}} \;.
\end{equation}
We observe
\begin{enumerate}
\item $F$ is analytic and thus extends by analytic continuation to $\tilde{F} : (0,\infty) \to (-1, \infty)$
\item Also by analytic continuation we obtain from \eqref{DefF} that $\tilde{F}'(r) = \frac{r}{4m^2} e^{\frac{r}{2m}} >0$
\item We have $\tilde{F}(r) \to -1$ for $r \to 0$, and $\tilde{F}(r) \to \infty$ for $r \to \infty$.
\end{enumerate}
Hence, $\tilde{F}$ is bijective and has an analytic inverse $\tilde{F}^{-1} : (-1, \infty) \to (0, \infty)$. We now define $r: H \to (0, \infty)$ by
\begin{equation*}
r(u,v) := \tilde{F}^{-1}(-uv)\;,
\end{equation*}
which is analytic.\footnote{Note that in region $I$ this implies $-uv = e^{\frac{r^*}{2m}}$.}

Finally, the Lorentzian metric on $\Mmax$ is given by
\begin{equation*}
\gmax = -\frac{16m^3}{r} e^{-\frac{r}{2m}} \big(du \otimes dv + dv \otimes du \big) + r^2 \, \mathring{\gamma}_2 \;.
\end{equation*}
We fix the time orientation by demanding that $\partial_v + \partial_u$ is future directed. A Penrose diagram of $(\Mmax, \gmax)$, i.e., $H$, is depicted below.
\begin{center}
\def\svgwidth{4cm}
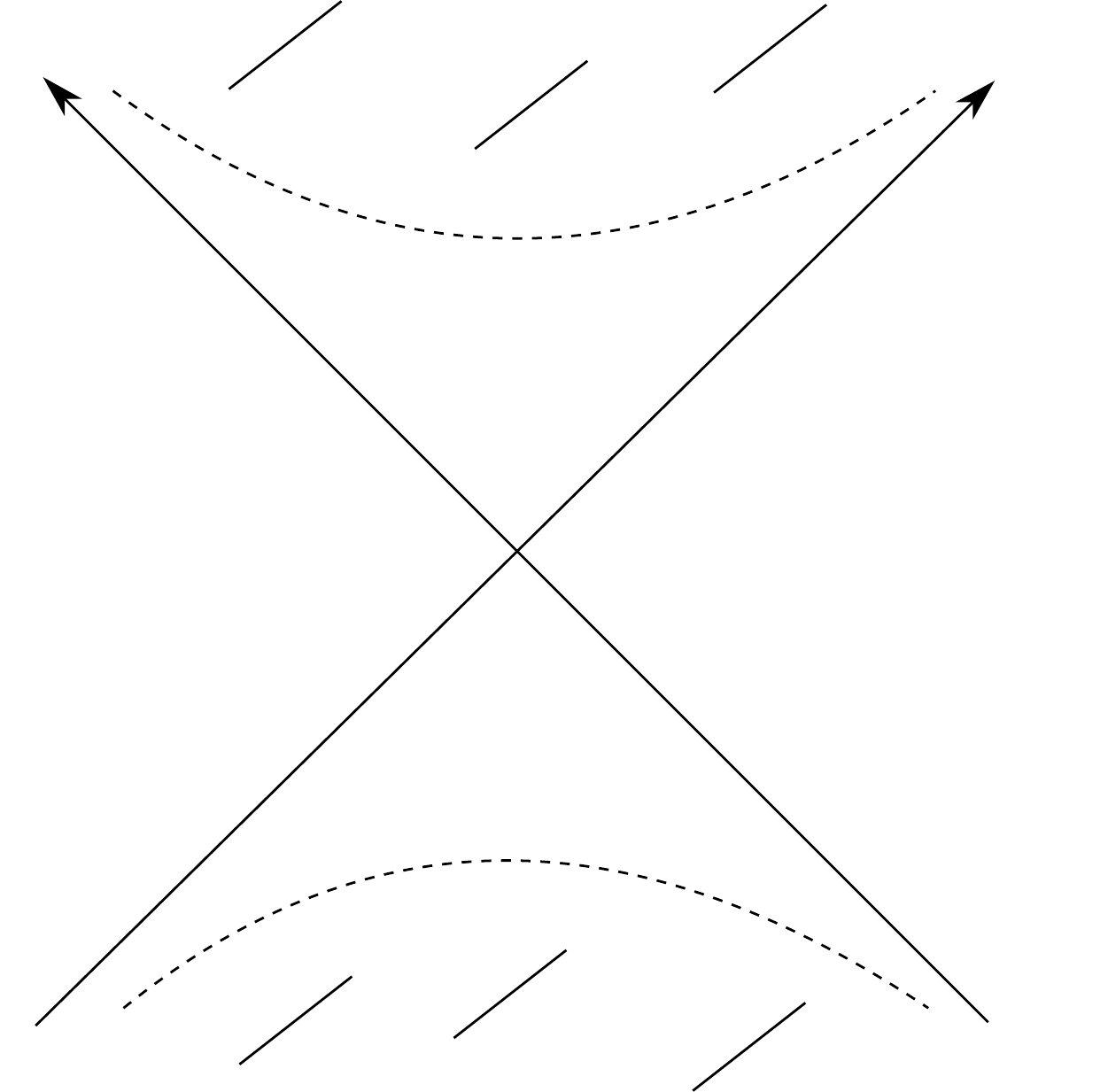
\end{center}
Note that the hypersurfaces $\{u=0\}$ and $\{v=0\}$, which separate the regions $I - IV$, are null hypersurfaces. Moreover, we consider the functions
\begin{equation}
\label{TFunction}
\begin{split}
t_I : I \to \R, \qquad t_I(u,v) &:= 2m \log \big(-\frac{v}{u}\big) \\
t_{II} : II \to \R, \qquad t_{II}(u,v) &:= 2m \log \big(\frac{v}{u}\big) \\
t_{III} : III \to \R, \qquad t_{III}(u,v) &:= 2m \log \big(\frac{v}{u}\big) \\
t_{IV }: IV \to \R, \qquad t_{IV}(u,v) &:= 2m \log \big(-\frac{v}{u}\big) \,.
\end{split}
\end{equation}
It can now be checked that the maps $(u,v, \omega) \mapsto \big(t_{A}(u,v), r(u,v), \omega\big)$, where $\omega \in \mathbb{S}^2$, are time orientation preserving isometries between region $A$ and $\Mext$ for $A \in \{I, IV\}$, and between region $A$ and $\Mint$ for $A \in \{II, III\}$.

The coordinates we have chosen to represent the maximal analytic extension of the Schwarzschild spacetime are due to Kruskal, cf.\ the paper \cite{Krus60} from 1960. Already in 1950, Synge had described this maximal extension in his paper \cite{Synge50}, using, however, a less concise choice of coordinates. In fact, he extended the `metric components' even beyond $\{r=0\}$, however, at $\{r=0\}$ they don't define a Lorentz metric.
\vspace*{3mm}

The form of the maximal analytic Schwarzschild spacetime in dimensions $d >3$ is slightly more complicated due to a slightly more complicated form of the so-called `tortoise coordinate' $r^*(r)$, but otherwise proceeds analogously:

The only zero on the positive real axis of the rational function $D(r) = 1 - \frac{r_+^{d-2}}{r^{d-2}}$ is the simple zero at $r = r_+$. Moreover, note that $D'(r_+) = \frac{d-2}{r_+}$. Hence, we can express $\frac{1}{D(r)}$ as
\begin{equation*}
\frac{1}{D(r)} = \frac{\nicefrac{r_+}{d-2}}{r-r_+} + w(r) \;,
\end{equation*}
where $w(r)$ is an analytic function on $[0,\infty)$. We fix a $c > r_+$ and define an analytic function $r^* : (r_+, \infty) \to (-\infty, \infty)$ by
\begin{equation}
\label{RStar2}
r^*(r) = \int_c^r \frac{1}{D(r')} \, dr' = \int_c^r \frac{\nicefrac{r_+}{d-2}}{r' - r_+} \, dr' + \int_c^r w(r') \, dr' = \frac{r_+}{d-2}\log (r-r_+) + W(r) \;,
\end{equation}
where $W(r)$ is an analytic function on $[0, \infty)$. Note that $r^*(r) \to \infty$ for $r \to \infty$ since $\frac{1}{D(r)} > \frac{1}{2}$ for $r > r_0$, $r_0$ large enough. We proceed by defining an analytic function $F : (r_+ , \infty) \to (0,\infty)$ by
\begin{equation*}
F(r) = e^{\frac{d-2}{r_+} r^*(r)} = (r- r_+) e^{\frac{d-2}{r_+} W(r)} \;.
\end{equation*}
Clearly, $F$ extends analytically to a $\tilde{F} : (0, \infty) \to \big(-r_+ e^{\frac{d-2}{r_+} W(0)}, \infty\big)$, and we compute
\begin{equation*}
\frac{d\tilde{F}}{dr}(r) = \frac{d-2}{r_+} \frac{dr^*}{dr} \tilde{F}(r) = \frac{d-2}{r_+} \frac{1}{D(r)} \tilde{F}(r) = \frac{d-2}{r_+} \frac{r^{d-2}}{r^{d-2} - r_+^{d-2}}(r-r_+) e^{\frac{d-2}{r_+}W(r)} >0 \;,
\end{equation*}
since $r^{d-2} - r_+^{d-2}$ has a simple zero at $r = r_+$.
Hence, $\tilde{F}$ is bijective and has an analytic inverse $\tilde{F}^{-1} : \big(-r_+ e^{\frac{d-2}{r_+} W(0)}, \infty\big) \to (0, \infty)$. 

We now set $H:= \big{\{}(u,v) \in \R^2 \, | \, uv < r_+ e^{\frac{d-2}{r_+} W(0)} \big{\}}$, and define $\Mmax := H \times \Sd$. Moreover, define $r : H \to (0, \infty)$ by $r(u,v) := \tilde{F}^{-1}(-uv)$. Finally, the Lorentzian metric on $\Mmax$ is given by
\begin{equation*}
\gmax = - \frac{2r_+}{(d-2) \tilde{F}'(r)} \big( du \otimes dv + dv \otimes du\big) + r^2 \, \mathring{\gamma}_{d-2} \;.
\end{equation*}
The time orientation is fixed by stipulating that $\partial_u + \partial_v$ is future pointing. We conclude by remarking that $H$ can again be partitioned into four regions $I - IV$, separated by the null hypersurfaces $\{u = 0\}$ and $\{v=0\}$, such that $I$ and $IV$ are isometric to $(\Mext, \gext)$, and $II$ and $III$ are isometric to $(\Mint, \gint)$. The $t$-component of the isometries  is again being given by \eqref{TFunction}, where $2m$ needs to be replaced by $\frac{r_+}{d-2}$. Also recall that we have $r<r_+$ in regions $II$ and $III$, $r > r_+$ in regions $I$ and $IV$, and $r=r_+$ one $\{u=0\}$ and $\{v = 0\}$.

\subsubsection*{The Eddington-Finkelstein $v^*$ coordinate in the exterior}

Using \eqref{RStar}, \eqref{RStar2}, respectively, we define the Eddington-Finkelstein coordinate $v^* : \Mext \to \R$ by
\begin{equation*}
v^*(t,r, \omega) := t + r^*(r)\;.
\end{equation*}
Note that under the above identification (i.e., the one induced by \eqref{TFunction}) of region $I$ with $\Mext$, the Eddington-Finkelstein coordinate $v^*$ is also given by 
\begin{equation}
\label{VRel}
v^* = \frac{2r_+}{d-2} \log v\;. 
\end{equation}
In $(v^*, r)$ coordinates for $\Mext$, the metric \eqref{SchwarzschildMetricExt} takes the form
\begin{equation}
\label{MetricEd}
\gext = -D(r) \, (dv^*)^2 + dv^* \otimes dr + dr \otimes dv^* + r^2 \, \mathring{\gamma}_{d-2} \;.
\end{equation}

\subsection{The main theorem}

\begin{theorem}
\label{MainThm}
The maximal analytic extension $(\Mmax, \gmax)$ of the Schwarzschild spacetime is $C^0$-inextendible.
\end{theorem}

This theorem can be easily deduced from the following two theorems:

\begin{theorem}
\label{Exterior}
There does not exist a $C^0$-extension $\iota : \Mext \hookrightarrow \tilde{M}$  and a timelike curve $\gamma : [-1, 0) \to \Mext$ with $(v^* \circ \gamma)(s) \to \infty$ for $s \nearrow 0$ such that $\iota \circ \gamma : [-1,0) \to \tilde{M}$ can be extended as a timelike curve to $[-1,0]$.
\end{theorem}

\begin{theorem}
\label{Inex}
There does not exist a $C^0$-extension $\iota : \Mint \hookrightarrow \tilde{M}$  and a timelike curve $\gamma : [-1, 0) \to \Mint$ with $(r \circ \gamma)(s) \to 0$ for $s \nearrow 0$ such that $\iota \circ \gamma : [-1,0) \to \tilde{M}$ can be extended as a timelike curve to $[-1,0]$.
\end{theorem}

\begin{proof}[Proof of Theorem \ref{MainThm} from Theorems \ref{Exterior} and \ref{Inex}:]
The proof is by contradiction, so we assume that there is a $C^0$-extension $(\tilde{M}, \tilde{g})$  of $(\Mmax, \gmax)$, where  $\iota : \Mmax \hookrightarrow \tilde{M}$ is the isometric embedding. By Lemma \ref{ExtensionLemma}, there exists a timelike curve $\tilde{\gamma} : [-1,0] \to \tilde{M}$ such that $\gamma := \iota^{-1} \circ \tilde{\gamma}|_{[-1, 0)} \to \Mmax$ is a timelike curve in $\Mmax$ and $\tilde{\gamma}(0) \in \tilde{M} \setminus \iota(\Mmax)$. Without loss of generality we can assume that $\gamma$ is future directed (otherwise reverse the time orientation of $\Mmax$), and thus it is future inextendible in $\Mmax$.

Since $\partial_u$ and $\partial_v$ are future directed, we obtain from $\gmax(\dot{\gamma}, \partial_u) <0$ and $\gmax(\dot{\gamma}, \partial_v)<0$ that $\dot{\gamma}_u >0$ and $\dot{\gamma}_v >0$. Since $\gamma$ is also future inextendible in $\Mmax$, we obtain that one of the following cases must hold:
\begin{enumerate}[(i)]
\item $\gamma_u(s) \nearrow u_0 \leq 0$ and $\gamma_v(s) \nearrow \infty$ for $s \nearrow 0$
\item $\gamma_u(s) \nearrow \infty$ and $\gamma_v(s) \nearrow v_0  \leq 0$ for $s \nearrow 0$
\item $\gamma_u(s) >0$ and $\gamma_v(s) >0$ for $s$ close enough to $0$ and $\gamma_u(s) \cdot \gamma_v(s) \nearrow r_+ e^{\frac{d-2}{r_+}W(0)}$ (or $\nearrow 1$ for $d=3$) for $s \nearrow 0$.
\end{enumerate} 
We distinguish the following cases:
\begin{itemize}
\item
In the cases (i) and (ii),  $\gamma$ is eventually contained in either region $I$ or in region $IV$. Since both regions are isometric, let us assume without loss of generality, that $\gamma$ is contained in region $I$, which we identify with $\Mext$ as explained above. It then follows from \eqref{VRel}, that $(v^* \circ \gamma) (s) \to \infty$ for $s \nearrow 0$. Moreover, the extension $\tilde{M}$ clearly furnishes a $C^0$-extension of $\Mext$ in which $\gamma$ can be extended to the future as a timelike curve. This, however, is a contradiction to Theorem \ref{Exterior}.
\item
In the case (iii), $\gamma$ is eventually contained in region $II$, which we identify with $\Mint$. It follows that $(r \circ \gamma)(s) \to 0$ for $s \nearrow 0$. Moreover, $\tilde{M}$ then also furnishes a $C^0$-extension of $\Mint$ in which $\gamma$ can be extended to the future as a timelike curve. This is in contradiction to Theorem \ref{Inex}.
\end{itemize}
\end{proof}
It thus remains to prove the Theorems \ref{Exterior} and \ref{Inex}.

\subsection{Proof of Theorem \ref{Exterior}}

The proof of Theorem \ref{Exterior} proceeds in analogy to the proof of Theorem \ref{MinkInex}. However, we need to prove the analogous statement to Step 2 for the case of the Schwarzschild exterior. Step 2.1 will be replaced by Proposition \ref{ExtFutCon}, and Step 2.2 by Proposition \ref{LengthToInfty}.

\begin{proposition}
\label{ExtFutCon}
The exterior of the Schwarzschild spacetime $(\Mext, \gext)$ is future one-connected.
\end{proposition}

Note that the proof of Step 2.1 for the Minkowskian case shows, in particular, that the Minkowski spacetime is future one-connected. Moreover, the proof given there can be transferred whenever the exponential map induces a global chart for the Lorentzian manifold. This, however, is clearly not the case for the Schwarzschild exterior.

\begin{proof}
The proof proceeds in two steps.
\vspace*{2mm}

\underline{\textbf{Step 1:}} Reduction to a Riemannian problem.
\vspace*{2mm}

We first note that whether a Lorentzian manifold is future one-connected or not depends only on the conformal class of the metric.
Thus, we can show instead that $\Mext$ endowed with the Lorentzian metric
\begin{equation*}
h_\mathrm{ext} := \big(1 - \frac{2m}{r^{d-2}}\big)^{-1} \gext = -\,dt^2 + \frac{1}{\big(1 - \frac{2m}{r^{d-2}}\big)^2} \,dr^2 + \frac{r^2}{1 - \frac{2m}{r^{d-2}}} \, \mathring{\gamma}_{d-1} =: -\,dt^2 + \overline{h}_{\mathrm{ext}}
\end{equation*}
is future one-connected. 
Note here, that $(\Mext, \hext)$ is a product of the Riemannian manifold $\big(\overline{M}_{\mathrm{ext}} = \big((2m)^{\frac{1}{d-2}}, \infty\big) \times \mathbb{S}^{d-1}, \overline{h}_{\mathrm{ext}}\big)$ and $(\R, -\,dt^2)$.  The advantage of working with such a product Lorentzian manifold is that the causality relations are determined purely by the geometry of the Riemannian factor. 

We first note that any future directed timelike curve $\gamma$ can be reparametrised by the $t$-coordinate, i.e., by slight abuse of notation, such that $\gamma : [t_0, t_1] \to \Mext$ is given by
\begin{equation*}
\gamma(t) = \big(t, \overline{\gamma}(t)\big)\;,
\end{equation*}
where $t_0, t_1 \in \R$ with $t_0 < t_1$. 

Since 
\begin{equation*}
0 > \hext\big(\dot{\gamma}(t), \dot{\gamma}(t)\big) = -1 + \overline{h}_{\mathrm{ext}}\big(\dot{\overline{\gamma}}(t), \dot{\overline{\gamma}}(t)\big)
\end{equation*}
holds for all $t \in [t_0, t_1]$, it follows that $||\dot{\overline{\gamma}}(t)||_{\overline{M}_{\mathrm{ext}}} <1$.

Vice versa, given a curve $\overline{\gamma} : [t_0, t_1] \to \overline{M}_{\mathrm{ext}}$ which satisfies $||\dot{\overline{\gamma}}(t)||_{\overline{M}_{\mathrm{ext}}} <1$  for all $t \in [t_0, t_1]$, then $\gamma(t) = \big(t, \overline{\gamma}(t)\big)$ is a future directed timelike curve\footnote{Paraphrased as a geometric statement, this shows that for $t_0 < t_1$ and $\overline{p}_0, \overline{p}_1 \in \oMe$ the causal relation $(t_0, \overline{p}_0) \ll (t_1, \overline{p}_1)$ holds if, and only if,  $d_{\oMe}(\overline{p}_0, \overline{p}_1) < |t_1 - t_0|$}.

Similarly, let $\overline{\Gamma} : [0,1] \times [t_0, t_1] \to \oMe$ be a homotopy with fixed endpoints which satisfies in addition
\begin{enumerate}[(i)] 
\item $\overline{\Gamma}(u, \cdot)$ is a piecewise smooth curve for all $u \in [0,1]$ 
\item $||\partial_t \overline{\Gamma}(u,t) ||_{\oMe} < 1$ holds for all $(u,t) \in [0,1] \times [t_0, t_1]$.
\end{enumerate} 
Then $\Gamma : [0,1] \times [t_0, t_1] \to \Mext$, given by $\Gamma (u,t) = \big(t, \overline{\Gamma}(u,t)\big)$ is a timelike homotopy with fixed endpoints. Thus, we say that a homotopy $\overline{\Gamma}$ with fixed endpoints has the \emph{timelike lifting property} if, and only if, it satisfies (i) and (ii).

To prove that $(\Mext, \hext)$ is future one-connected, we need to show that if $\gamma_i : [t_0, t_1] \to \Mext$, $\gamma_i(t) = \big(t, \overline{\gamma}_i(t)\big)$, where $i=0,1$, are two future directed timelike curves with $\gamma_0(0) = \gamma_1(0)$ and $\gamma_0(1) = \gamma_1(1)$, then there exists a timelike homotopy $\Gamma : [t_0,t_1] \times [t_0,t_1] \to \Mext$ with fixed endpoints between $\gamma_0$ and $\gamma_1$. As shown above, the timelike curves $\gamma_0$ and $\gamma_1$ project down to piecewise smooth curves $\overline{\gamma}_0, \overline{\gamma}_1 : [t_0, t_1] \to \overline{M}_{\mathrm{ext}}$ with $||\dot{\overline{\gamma}}_i(t)||_{\overline{M}_{\mathrm{ext}}} < 1$ for all $t \in [t_0, t_1]$, for $i = 0,1$. 
By the above argument, it suffices to show that there exists a homotopy $\overline{\Gamma}$ with fixed endpoints between $\overline{\gamma}_0$ and $\overline{\gamma}_1$ which has the timelike lifting property. Moreover, we can assume that $t_0 = 0$ and $t_1 = T$.

By concatenation of homotopies, this follows from Step 2:
\vspace*{2mm}

\underline{\textbf{Step 2:}} Let $\overline{\gamma} : [0, T] \to \overline{M}_{\mathrm{ext}}$ be a piecewise smooth curve with $||\dot{\overline{\gamma}}(t)||_{\overline{M}_{\mathrm{ext}}} < 1$ for all $t \in [0, T]$. We write $\overline{\gamma}(t) = \big(\ovg_r(t), \ovg_\omega(t)\big)$, where $\ovg_r$ is the projection of $\ovg$ on $\big((2m)^{\frac{1}{d-2}}, \infty\big)$, while $\ovg_\omega$ is the projection on $\mathbb{S}^{d-1}$. 

If $\ovg_\omega (0) \neq - \ovg_\omega(T)$, then there exists a continuous homotopy $\overline{\Gamma} : [0, T] \times [0, T] \to \overline{M}_{\mathrm{ext}}$ with fixed endpoints between $\ovg$ and  the \emph{unique} shortest curve from $\ovg(0)$ to $\ovg(T)$ which, moreover, has the timelike lifting property.
 
If $\ovg_\omega(0) = -\ovg_\omega(T)$, then there exists a continuous homotopy $\overline{\Gamma} : [0, T] \times [0, T] \to \overline{M}_{\mathrm{ext}}$ with fixed endpoints between $\ovg$ and the unique shortest curve $\overline{\sigma}$ from $\ovg(0)$ to $\ovg(T)$, such that the projection $\overline{\sigma}_\omega$ traces out a fixed geodesic arc (half a great circle) on $\mathbb{S}^{d-1}$ connecting $\ovg_\omega(0)$ with $\ovg_\omega(T)$ which, moreover, has the timelike lifting property.
\vspace*{2mm}

We prove the statement of Step 2 in several steps. The idea behind steps 2.1 - 2.4 is to first transform $\ovg_\omega$ to a shortest geodesic arc while leaving $\ovg_r$ unchanged. The homotopy is constructed using the exponential map on the sphere $\Sd$ in a way similar to the argument already encountered in Step 2.1 of the proof of Theorem \ref{MinkInex}. This however requires a small perturbation of the curve such that the image of the curve is disjoint from the point antipodal to the base point of the exponential map. This perturbation is constructed in the steps 2.1 and 2.2. Having straightened out the angular part of $\ovg$, the problem is reduced to shortening the curve even further in a two-dimensional submanifold of $\oMe$, which has negative Gauss curvature. The Cartan-Hadamard Theorem then allows us to shorten the curve using again the exponential map.
\vspace*{2mm}

\textbf{Step 2.1:} Since $\ovg : [0,T] \to \oMe$ is piecewise smooth with $||\dot{\ovg}||_{\oMe} < 1$, there exists a $\delta >0$ such that $\ovg$ is smooth on $[0,\delta]$, and, moreover, there exists an $\varepsilon >0$ such that $||\dot{\ovg}(t)||_{\oMe} < 1 - \varepsilon$ for all  $t \in [0,\delta]$. We define $\lambda : [0, \varepsilon\delta] \times [0, T] \to [0,T]$ by
\begin{equation*}
\lambda(u,t) := \begin{cases} 0 &\textnormal{ for } 0 \leq t \leq u \\
(t-u) \frac{\delta}{\delta - u} \quad &\textnormal{ for } u \leq t \leq \delta \\
t &\textnormal{ for }  \delta \leq t \leq T \;,  \end{cases}
\end{equation*}
and set
\begin{equation*}
\overline{\Gamma}_1 (u,t) := \ovg\big( \lambda(u,t)\big) \;.
\end{equation*}

Note that $\frac{\delta}{\delta - u} \leq \frac{1}{1-\varepsilon}$ for all $u \in [0, \varepsilon \delta]$, and thus we have $|| \partial_t \overline{\Gamma}_1||_{\oMe} < 1$. It now follows that $\overline{\Gamma}_1 : [0, \varepsilon \delta] \times [0, T] \to \oMe$ is a homotopy with fixed endpoints that has the timelike lifting property. We set $\overline{\Gamma}_1(\varepsilon \delta, \cdot) =: \ovg^{(1)}(\cdot)$.
\vspace*{2mm}

\textbf{Step 2.2:} Applying Sard's Theorem to each smooth component of $\ovg_\omega$, we infer that the image $\mathrm{Im}(\ovg_\omega)$ of $\ovg_\omega$ has measure zero in $\Sd$. Let $\rho_0 = \frac{\varepsilon \delta D^{\nicefrac{1}{2}}\big(\ovg_r(0)\big)}{2 \ovg_r(0)}$. It follows that 
\begin{equation}
\label{ChoiceOmega}
\parbox{0.8\textwidth}{
there exists an $\omega_0 \in B_{\rho_0}\big(\ovg_\omega(0)\big)$, such that $\mathrm{Im}(\ovg_\omega) \subseteq \Sd \setminus \{-\omega_0\}$, and such that the closed geodesic arc from $\omega_0$ to $\ovg_\omega(0)$ intersects $- \ovg_\omega(T)$ at most in $\ovg_\omega(0)$.}
\end{equation}
We consider the exponential map $\exp_{\ovg_\omega(0)} : T_{\ovg_\omega(0)}\Sd \supseteq B_\pi(0) \to \Sd$ with base point $\ovg_\omega(0)$. Recall, that it is a diffeomorphism on $B_\pi(0)$. We now define $\overline{\Gamma}_2 : [0,1] \times [0,T] \to \oMe$ by
\begin{equation*}
\overline{\Gamma}_2(u,t) := \begin{cases} \Big( \ovg^{(1)}_r(0), \exp_{\ovg_\omega(0)}\big[\frac{t}{\nicefrac{\varepsilon \delta}{2}} \cdot u \cdot \exp^{-1}_{\ovg_\omega(0)}(\omega_0)\big] \Big)  &\textnormal{ for } 0 \leq t \leq \frac{\varepsilon \delta}{2} \\
\Big( \ovg^{(1)}_r(0), \exp_{\ovg_\omega(0)}\big[\frac{(\varepsilon \delta - t)}{\nicefrac{\varepsilon \delta}{2}} \cdot u \cdot \exp^{-1}_{\ovg_\omega(0)}(\omega_0)\big] \Big) \quad &\textnormal{ for }  \frac{\varepsilon \delta}{2} \leq t \leq \varepsilon \delta \\
\ovg^{(1)}(t) &\textnormal{ for } \varepsilon \delta \leq t \leq T \;. 
\end{cases}
\end{equation*}
We compute for $t \in [0, \frac{\varepsilon \delta}{2}]$ 
\begin{equation*}
\begin{split}
||\partial_t \overline{\Gamma}_2(u,t) ||^2_{\oMe} &= \big[\ovg_r(0)\big]^2 \cdot \frac{1}{D\big( \ovg_r(0)\big)} \cdot \frac{u^2}{\big(\nicefrac{\varepsilon \delta}{2}\big)^2} \cdot || \exp^{-1}_{\ovg_\omega}(0)(\omega_0) ||_{\Sd}^2 \\
&=  \big[\ovg_r(0)\big]^2 \cdot \frac{1}{D\big( \ovg_r(0)\big)} \cdot \frac{u^2}{\big(\nicefrac{\varepsilon \delta}{2}\big)^2} \cdot   d^2_{\oMe}\big( \omega_0, \ovg_\omega(0)\big) \\
&< u^2 \\
&\leq 1 \;,
\end{split}
\end{equation*}
where we have used \eqref{ChoiceOmega}. We set $\ovg^{(2)}(\cdot) := \overline{\Gamma}_2(1, \cdot)$. Hence, $\overline{\Gamma}_2$ is a homotopy with fixed endpoints between $\ovg^{(1)}$ and $\ovg^{(2)}$ that has the timelike lifting property. Note that $\ovg^{(2)}(\frac{\varepsilon \delta}{2}) = \omega_0$.
\vspace*{2mm}

\textbf{Step 2.3:} We now use the exponential map $\exp_{\omega_0} : T_{\omega_0} \Sd \supseteq B_\pi(0) \to \Sd$ based at $\omega_0$ to define $\overline{\Gamma}_3 : [ \frac{\varepsilon \delta }{2} , T] \times [0, T] \to \oMe$ by
\begin{equation*}
\overline{\Gamma}_3(u,t) := \begin{cases} \ovg^{(2)}(t) &\textnormal{ for } 0 \leq t \leq \frac{\varepsilon \delta}{2} \\
\Big(\ovg_r^{(2)}(t), \exp_{\omega_0} \big[f(t,u) \exp^{-1}_{\omega_0}\big(\ovg^{(2)}_\omega(u)\big)\big] \Big) \quad &\textnormal{ for } \frac{\varepsilon \delta}{2} \leq t \leq u \\
\ovg^{(2)}(t) &\textnormal{ for } u \leq t \leq T \;,
\end{cases}
\end{equation*}
where
\begin{equation*}
f(t,u) := \frac{\int_{\nicefrac{\varepsilon \delta}{2}}^t || \dot{\ovg}^{(2)}_\omega(t')||_{\Sd} \, dt'}{\int_{\nicefrac{\varepsilon \delta}{2}}^u || \dot{\ovg}^{(2)}_\omega(t') ||_{\Sd} \, dt' } \;.
\end{equation*}
We compute
\begin{equation*}
\partial_t f(t,u) = \frac{|| \dot{\ovg}^{(2)}_\omega(t)||_{\Sd}}{\int_{\nicefrac{\varepsilon \delta}{2}}^u || \dot{\ovg}^{(2)}_\omega(t') ||_{\Sd} \, dt' } \leq \frac{|| \dot{\ovg}^{(2)}_\omega(t)||_{\Sd}}{d_{\Sd}\big(\omega_0, \ovg^{(2)}_\omega(u)\big)} = \frac{|| \dot{\ovg}^{(2)}_\omega(t)||_{\Sd}}{||\exp^{-1}_{\omega_0}\big( \ovg^{(2)}_\omega(u)\big)||_{\Sd}}\,
\end{equation*}
and, thus, for $t \in [\frac{\varepsilon \delta}{2}, u]$ we obtain
\begin{equation*}
\begin{split}
||\partial_t\overline{\Gamma}_3(u,t) ||^2_{\oMe} &= D^{-2}\big(\ovg^{(2)}_r(t)\big) \cdot  \big[ \dot{\ovg}_r^{(2)}(t)\big]^2 + \big[\ovg_r^{(2)}\big]^2 \cdot D^{-1}\big(\ovg_r^{(2)}(t)\big) \cdot |\partial_t f(t,u)|^2 \cdot ||\exp^{-1}_{\omega_0}\big( \ovg^{(2)}_\omega(u)\big)||^2_{\Sd} \\
&\leq ||\dot{\ovg}^{(2)}(t)||_{\oMe}^2 \\
&< 1 \;.
\end{split}
\end{equation*}
Setting $\ovg^{(3)}(\cdot) := \overline{\Gamma}_3(T,\cdot)$,  it follows that $\overline{\Gamma}_3$ is a homotopy with fixed endpoints between $\ovg^{(2)}$ and $\ovg^{(3)}$ which has the timelike lifting property.
\vspace*{2mm}

\textbf{Step 2.4:} We now consider the exponential map $\exp_{\ovg_\omega(T)} : T_{\ovg_\omega(T)} \Sd \supseteq B_\pi (0) \to \Sd$ based at $\ovg_\omega(T)$. By the choice of $\omega_0$, \eqref{ChoiceOmega}, we have $\mathrm{Im}(\ovg_\omega^{(3)}) \setminus \big{\{} \ovg_\omega^{(3)}(0)\big{\}} \subseteq \Sd \setminus \big{\{} - \ovg_\omega^{(3)}(T)\big{\}}$. Hence, $\overline{\Gamma}_4 : [0,T) \times [0,T] \to \oMe$, defined by
\begin{equation*}
\overline{\Gamma}_4(u,t) := \begin{cases} \ovg^{(3)}(t) &\textnormal{ for } 0 \leq t \leq T-u \\
\Big( \ovg^{(3)}_r (t) , \exp_{\ovg_\omega(T)} \big[ f(t,u) \exp^{-1}_{\ovg_\omega(T)} \big( \ovg_\omega^{(3)} (T-u)\big)\big] \Big) \quad &\textnormal{ for } T- u \leq t \leq T \;, \end{cases}
\end{equation*}
where
\begin{equation*}
f(t,u) = \frac{ \int_t^T || \dot{\ovg}^{(3)}_\omega(t') ||_{\Sd} \, dt'}{ \int_{T-u}^T || \dot{\ovg}_\omega^{(3)}(t')||_{\Sd} \, dt'}\;,
\end{equation*}
is well-defined. If $\ovg_\omega(0) \neq - \ovg_\omega(T)$, then one can extend $\overline{\Gamma}_4$ to $[0,T] \times [0,T]$ by the above definition. In the case $\ovg_\omega(0) = - \ovg_\omega(T)$, we introduce spherical normal coordinates $(\rho, \overline{\theta})$ for $\Sd$ at $\ovg_\omega(T)$, and express $\ovg_\omega^{(3)}(t)$ with respect to these coordinates by $\big(\rho(t), \overline{\theta}(t)\big)$. It now follows that $\overline{\Gamma}_4(u,\cdot)$ converges for $u \nearrow T$ to a piecewise smooth curve such that its projection onto the sphere $\Sd$ traces out the geodesic arc parametrised, in the spherical normal coordinates, by $\rho \mapsto \big(\rho, \lim_{t \searrow 0} \overline{\theta}(t)\big)$, $\rho \in [0,\pi)$. In this case this limit curve furnishes the extension of $\overline{\Gamma}_4$. Moreover, after a homotopy of rotations, we can assume that the projection on the sphere $\Sd$ of this limit curve lies in a fixed two-dimensional plane in $\R^d \supseteq \Sd$ through $0, \ovg_\omega(0)$, and $\ovg_\omega(T)$.

Finally, we set $\ovg^{(4)}(\cdot):= \overline{\Gamma}_4(T, \cdot)$, and as in Step 2.3 one computes that $\overline{\Gamma}_4$ is a homotopy with fixed endpoints between $\ovg^{(3)}$ and $\ovg^{(4)}$ that has the timelike lifting property.
\vspace*{2mm}

\textbf{Step 2.5:} We now introduce standard coordinates on $\Sd$ such that the $\mathbb{S}^1$, in which $\ovg_\omega^{(4)}$ is mapping, is parametrised by $\varphi \in (0, 2\pi)$. We now consider the submanifold $F:= \big((2m)^{\frac{1}{d-2}}, \infty\big) \times \mathbb{S}^1 \subseteq \oMe$ with the induced Riemannian metric
\begin{equation*}
\overline{h}_F := \frac{1}{\big[D(r)\big]^2} \, dr^2 + \frac{r^2}{D(r)} \, d\varphi^2 \;.
\end{equation*}
The Gauss curvature of $(F, \overline{h}_F)$ is computed to be
\begin{equation*}
K_F =  \frac{1}{2} D(r) D''(r) - \frac{1}{4} \big[ D'(r)\big]^2 \;.
\end{equation*}
Moreover, we have $D'(r) = 2m (d-2) r^{1-d} >0$, and $D''(r) =-2m (2-d)(1-d)r^{-d} <0$, so the Gauss curvature of $F$ is negative. Moreover, it is easy to see that $(F, \overline{h}_F)$ is complete: let $\sigma : \big( (2m)^{\frac{1}{d-2}} , a\big) \to F$, $\sigma(r) = \big( r, \varphi(r)\big)$, be a $C^1$ curve, where $a >  (2m)^{\frac{1}{d-2}}$. We compute
\begin{equation*}
L(\sigma) = \int^a_{(2m)^{\frac{1}{d-2}}} \sqrt{ \overline{h}_F \big( \dot{\sigma}(r), \dot{\sigma}(r)\big)} \, dr \geq \int^a_{(2m)^{\frac{1}{d-2}}}  \frac{1}{D(r)} \, dr = \int^a_{(2m)^{\frac{1}{d-2}}}  \frac{r^{d-2}}{r^{d-2} - 2m} \, dr = \infty \;.
\end{equation*}

We now consider the universal cover $\pi_F : \tilde{F} = \big((2m)^{\frac{1}{d-2}}, \infty\big) \times \R \to F$  of $F$ and lift $\ovg^{(4)}$ to a curve $\ovg^{(4)}_{\mathrm{lift}} : [0,T] \to \tilde{F}$. Clearly, $\tilde{F}$ with the induced Riemannian metric $\overline{h}_{\tilde{F}}$ is complete, has negative Gauss curvature, and is simply connected.
Hence, the Cartan-Hadamard Theorem (see for instance Theorem 11.5 in \cite{LeeRiem}) states, in particular, that $\exp_{\ovg_{\mathrm{lift}}^{(4)}(0)} : T_{\ovg_{\mathrm{lift}}^{(4)}(0)} \tilde{F} \to \tilde{F}$ is a diffeomorphism. Define $\overline{\Gamma}_{5, \mathrm{lift}} : [0,T] \times [0,T] \to \tilde{F}$ by
\begin{equation*}
\overline{\Gamma}_{5,\mathrm{lift}} (u,t) := \begin{cases} \exp_{\ovg_{\mathrm{lift}}^{(4)}(0)} \Big( \frac{t}{u} \exp^{-1}_{\ovg_{\mathrm{lift}}^{(4)}(0)}\big[\ovg_{\mathrm{lift}}^{(4)} (u) \big]\Big) \quad &\textnormal{ for } 0 \leq t \leq u \\
\ovg_{\mathrm{lift}}^{(4)}(t) &\textnormal{ for } u \leq t \leq T \;.
\end{cases}
\end{equation*}
Note that for $0 \leq t \leq u$, we have
\begin{equation*}
|| \partial_t \overline{\Gamma}_{5, \mathrm{lift}} (u,t) ||_{\tilde{F}} = \frac{1}{u} || \exp^{-1}_{\ovg_{\mathrm{lift}}^{(4)}(0)} \big( \ovg^{(4)}_{\mathrm{lift}}(u)\big)||_{\tilde{F}} = \frac{1}{u} d_{\tilde{F}} \big( \ovg^{(4)}_{\mathrm{lift}}(0), \ovg^{(4)}_{\mathrm{lift}}(u)\big) < 1 \;,
\end{equation*}
where we have used for the second equality that $\exp_{\ovg_{\mathrm{lift}}^{(4)}(0)}$ is a diffeomorphism, while for the inequality that $||\dot{\ovg}^{(4)} (t) ||_F < 1$ together with $\pi_F$ being a local isometry. We now define $\overline{\Gamma}_5 : [0,T] \times [0,T] \to F$ by $\overline{\Gamma}_5(u,t) := \pi_F \circ \overline{\Gamma}_5(u,t)$, and it is clear that $\overline{\Gamma}_5$ is a homotopy with fixed endpoints between $\ovg^{(4)}$ and the unique shortest curve from $\ovg(0)$ to $\ovg(T)$ (in the case of $\ovg_\omega(0) = - \ovg_\omega(T)$, with the additional requirement as in the statement of Step 2) that has the timelike lifting property. Concatenating the homotopies $\overline{\Gamma}_1$ up to $\overline{\Gamma}_5$ finishes the proof of Step 2, and thus proves Proposition \ref{ExtFutCon}.
\end{proof}

\begin{proposition}
\label{LengthToInfty}
Let $\gamma : [0,\infty) \to \Mext$ be a future directed timelike curve with $(v^* \circ \gamma)(s) \to \infty$ for $s \to \infty$. It then follows that $d_{\Mext}\big(\gamma(0), \gamma(s)\big) \to \infty$ for $s \to \infty$.
\end{proposition}

\begin{proof}
We distinguish the following two cases:
\vspace*{2mm}

\textbf{Case 1:} $(r \circ \gamma)(s) \to \infty$ for $s \to \infty$.
\vspace*{2mm}

Let $R> r_+$ be sufficiently large such that $D(r) > \frac{1}{2}$ for all $r > R$. By assumption there exists an $s_0 > 0$ such that for all $s \geq s_0$ we have $\gamma_r(s) >R$. Let us define $M_{\mathrm{ext}, R} := \Mext \cap \{r>R\}$. Clearly, it suffices to prove $d_{\big(M_{\mathrm{ext},R}, \gext\big)}\big(\gamma(s_0), \gamma(s)\big) \to \infty$ for $s \to \infty$. 

Considering again the conformal metric
\begin{equation*}
h_{\mathrm{ext}} := \frac{1}{D(r)} \gext = -\, dt^2 + \frac{1}{\big[D(r)\big]^2} \, dr^2 +\frac{r^2}{D(r)} \, \mathring{\gamma}_{d-2} = -\, dt^2 + \overline{h}_{\mathrm{ext}}
\end{equation*}
on $M_{\mathrm{ext}, R} =: \R \times \overline{M}_{\mathrm{ext}, R}$, we note that $\gamma : [s_0, \infty) \to M_{\mathrm{ext}, R} $ is also timelike in $\big(M_{\mathrm{ext}, R} , h_{\mathrm{ext}}\big)$. Moreover, for $s_1 > s_0$, we have
\begin{equation*}
L_{(M_{\mathrm{ext}, R} , \gext)}\big(\gamma|_{[s_0, s_1]}\big) > \frac{1}{\sqrt{2}} L_{(M_{\mathrm{ext}, R} , h_{\mathrm{ext}})} \big(\gamma|_{[s_0, s_1]}\big) \;.
\end{equation*}
It thus suffices to show $d_{\big(M_{\mathrm{ext},R}, h_{\mathrm{ext}}\big)}\big(\gamma(s_0), \gamma(s)\big) \to \infty$ for $s \to \infty$.

Without loss of generality we can assume that $\gamma : [s_0, \infty) \to M_{\mathrm{ext},R}$ is parametrised by the $t$-coordinate, i.e.,
\begin{equation*}
\gamma(t) = \big(t, \overline{\gamma}(t)\big) \;.
\end{equation*}
We obtain
\begin{equation}
\label{TDist}
d_{\big(M_{\mathrm{ext},R}, h_{\mathrm{ext}}\big)}\big(\gamma(s_0), \gamma(s)\big) = \sqrt{ (s-s_0)^2 - d_{\overline{M}_{\mathrm{ext}, R}}^2\big(\ovg(s_0), \ovg(s)\big) } \;.
\end{equation}
Since $\gamma$ is timelike, we have $\overline{h}_{\mathrm{ext}}\big(\dot{\ovg}(s), \dot{\ovg}(s)\big) < 1$, and, hence, there exists an $\varepsilon >0$ such that $d_{\overline{M}_{\mathrm{ext}, R}}\big(\ovg(s_0), \ovg(s_0 + 1)\big) = 1 - \varepsilon$. It follows that
\begin{equation*}
\begin{split}
d_{\overline{M}_{\mathrm{ext}, R}}\big(\ovg(s_0), \ovg(s)\big) &\leq d_{\overline{M}_{\mathrm{ext}, R}}\big(\ovg(s_0), \ovg(s_0 + 1) \big) + d_{\overline{M}_{\mathrm{ext}, R}}\big(\ovg(s_0 + 1), \ovg(s)\big) \\
&\leq 1- \varepsilon + (s-s_0) -1 \\
&= (s-s_0) -\varepsilon \;.
\end{split}
\end{equation*}
We now obtain from \eqref{TDist}
\begin{equation*}
d_{\big(M_{\mathrm{ext},R}, h_{\mathrm{ext}}\big)}\big(\gamma(s_0), \gamma(s)\big) \geq \sqrt{(s-s_0)^2 - (s-s_0 - \varepsilon)^2} = \sqrt{2(s-s_0)\varepsilon - \varepsilon^2} \to \infty 
\end{equation*}
for $s \to \infty$.
\vspace*{2mm}

\textbf{Case 2:} $(r \circ \gamma)(s) \not\to \infty$ for $s \to \infty$.
\vspace*{2mm}

We choose the Eddington-Finkelstein coordinates $(v^*, r, \omega)$ to work with. Without loss of generality we can assume that $\gamma : [0, \infty) \to \Mext$ is parametrised by $v^*$, i.e.,
\begin{equation*}
\gamma(v^*) = \big(v^*, \gamma_r(v^*), \gamma_\omega(v^*)\big) \;.
\end{equation*}
By assumption there exists an $R>r_+$ such that for all $n \in \mathbb{N}$ there exists an $s_n > n$ such that $\gamma_r(s_n) < R$. We show that $d_{\Mext}\big(\gamma(0), \gamma(s_n) \big) \to \infty$ for $n \to \infty$.

The following two statements can be easily proved - the first using geodesic arcs on $\Sd$, the second using radial null geodesics:
\begin{equation}
\label{Sphere}
\parbox{0.8\textwidth}{There exists a $v^*_1 >0$ such that for all $\omega_f \in \Sd$, there exists a timelike curve $\sigma_1 : [0, v^*_1] \to \Mext$ with $\sigma_1(0) = \gamma(0)$ and $\sigma_1(v^*_1) = \big(v^*_1, \gamma_r(0), \omega_f\big)$.}
\end{equation}

\begin{equation}
\label{RadialGeod}
\parbox{0.8\textwidth}{There exists a $\Delta v^* >0$ such that for all $s_{n_0} >0$, for all $\omega_f \in \Sd$, and for all $r_f \in (r_+, R)$, there exists a timelike curve $\sigma_3 :[s_{n_0} - \Delta v^*, s_{n_0}] \to \Mext$ with $\sigma_3(s_{n_0} - \Delta v^*) = (s_{n_0} - \Delta v^*, \gamma_r(0), \omega_f)$ and $\sigma_3(s_{n_0}) = (s_{n_0}, r_f, \omega_f)$.}
\end{equation}

Given $C>0$, we then choose $n_0 \in \mathbb{N}$ so that $D\big(\gamma_r(0)\big) [n_0 - \Delta v^* - v_1]^2 > C^2$. Let $\omega_f := \gamma_\omega(s_{n_0})$ and $r_f := \gamma_r(s_{n_0})$. Moreover, define $\sigma_2 : [v_1, s_{n_0} - \Delta v^*] \to \Mext$ by
\begin{equation*}
\sigma_2(v^*) = \big(v^*, \gamma_r(0), \omega_f\big) \;.
\end{equation*}
Note that by \eqref{MetricEd} we have $L(\sigma_2) > C$.
It then follows that $\sigma_1 * \sigma_2 * \sigma_3$ is a future directed timelike curve from $\gamma(0)$ to $\gamma(s_{n_0})$ of Lorentzian length greater than $C$.
\end{proof}

\begin{proof}[Proof of Theorem \ref{Exterior}:]
The proof is analogous to the proof of Theorem \ref{MinkInex}. One first repeats literally Step 1.1 - Step 1.3. The analogous statement to Step 2 is obtained from Proposition \ref{ExtFutCon} (which replaces Step 2.1) and Proposition \ref{LengthToInfty} (which replaces Step 2.2). Step 3 is then virtually identical again.
\end{proof}

\section{The spacelike diameter in Lorentzian geometry}
\label{SecDiam}

Theorem \ref{DiamFinite} of this section will be used in the proof of Theorem \ref{Inex}. The results of this section might, however, be also of independent interest. We begin by introducing the notion of the \emph{spacelike diameter}, a geometric quantity at the level of $C^0$-regular Lorentzian metrics.

\begin{definition}
Let $(N,g)$ be a connected  and globally hyperbolic Lorentzian manifold with a $C^0$-regular metric $g$. The \emph{spacelike diameter}  $\diam_s(N)$ of $N$ is defined by 
\begin{equation*}
\diam_s(N) := \sup_{\substack{\Sigma \textnormal{ Cauchy} \\ \textnormal{hypersurface of } N}} \diam\, \Sigma = \sup_{\substack{\Sigma \textnormal{ Cauchy} \\  \textnormal{hypersurface of } N}} \sup_{p,q \in \Sigma} \inf_{\substack{\gamma : [0,1] \to \Sigma \\ \textnormal{ piecewise smooth curve} \\  \textnormal{with } \gamma(0) = p \textnormal{ and } \gamma(1) = q}} L(\gamma)\;.
\end{equation*}
Here, $L(\gamma) = \int_0^1 \sqrt{g\big(\dot{\gamma}(s), \dot{\gamma}(s)\big)} \, ds $ is the length of the curve $\gamma$. Note that this is well-defined since the tangent space of a Cauchy hypersurface $\Sigma$ does not contain timelike vectors.
\end{definition}

Let us remark, that an attempt to capture the notion of a spacelike diameter in direct analogy to the Riemannian case, i.e., by considering shortest spacelike curves, clearly fails, since one can connect any two points by a spacelike curve of arbitrarily short length by using a nearly null zig-zag path. 

\begin{definition}
\label{DefRegFlowChart}
Let $(N,g)$ be a globally hyperbolic Lorentzian manifold with a $C^0$-regular metric $g$. A chart $\psi : U \to D$ for $N$, where $U \subseteq N$ and $D \subseteq \R^{d+1}$, is called a \emph{regular flow chart for $N$} if, and only if
\begin{enumerate}
\item There exist constants $C, c$ such that the metric components in this chart satisfy the uniform bounds $|g_{\mu \nu} | \leq C < \infty$ and $g_{00} \leq c < 0$.
\item The domain $D$ is of the form $D = \bigcup_{ \ux \in B} I_{\ux} \times \{\ux\} \subseteq \R \times B$, where $B \subseteq \R^d$ and $I_{\ux} \subseteq \R$ is an open and connected interval. Moreover, the \emph{coordinate diameter} $\diam_e(B)$ of $B$ is finite, where $\diam_e(B)$ is the diameter of $B \subseteq \R^d$ with respect to the standard Euclidean metric $e$.
\item The timelike curves $I_{\ux} \ni s \mapsto \psi^{-1}(s, \ux)$ are inextendible in $N$.
\end{enumerate}
\end{definition}

The following theorem is the main theorem of this section - it gives a sufficient criterion for the spacelike diameter to be finite.

\begin{theorem}
\label{DiamFinite}
Let $(N,g)$ be a connected and globally hyperbolic Lorentzian manifold with a $C^0$-regular metric $g$ and let $\psi_k : U_k \to D_k$, $k = 1, \ldots, K$, be a finite collection of regular flow charts for $N$ with $\bigcup_{1 \leq k \leq K} U_k = N$.

Then one has $\diam_s(N) < \infty$.
\end{theorem}

\begin{proof}
The proof is divided into two steps.
\vspace*{2mm}

\underline{\textbf{Step 1:}} For all Cauchy hypersurfaces $\Sigma$ of $N$ we have $\diam (\Sigma) \leq \sum_{k =1}^K \diam(\Sigma \cap U_k)$.
\vspace*{2mm}

Let $\Sigma$ be an arbitrary Cauchy hypersurface of $N$ and set $V_k := U_k \cap \Sigma$.
First note that without loss of generality we can assume that the $V_k$ are ordered such that for all $1 \leq m < K$ we have
\begin{equation}
\label{NiceOrder}
\Big(\bigcup_{1 \leq k \leq m} V_k \Big) \cap V_{m+1} \neq \emptyset \;.
\end{equation}
The proof of this is an easy induction: starting with the open cover $\{ V_1, \bigcup_{2 \leq k \leq K} V_k \}$ of $\Sigma$, the connectedness of $\Sigma$ implies that $V_1 \cap \bigcup_{2 \leq k \leq K} V_k \neq \emptyset$. Hence, there is a $k_0 \in \{2, \ldots, K\}$ with  $V_1 \cap V_{k_0} \neq \emptyset$. After relabelling we can assume that $k_0 = 2$. We then  consider the open cover $\{(V_1 \cup V_2, \bigcup_{3 \leq k \leq K}V_k\}$ and proceed analogously. A finite number of iterations proves the claim.

Now assuming \eqref{NiceOrder}, we claim
\begin{equation}
\label{IndCover}
\diam\Big(\bigcup_{1 \leq k \leq m} V_k \Big) \leq \sum_{k=1}^m \diam(V_k) \qquad \textnormal{ holds for all } 1 \leq m \leq K \;.
\end{equation}
The proof is an induction in $m$. For $m=1$ there is nothing to show. Now assume \eqref{IndCover} holds for $m = m_0 < K$. Let $p,q \in \bigcup_{1 \leq k \leq m_0 + 1} V_k$. We need to show that the distance between $p$ and $q$ in $\bigcup_{1 \leq k \leq m_0 + 1} V_k$ is bounded by the right hand side of \eqref{IndCover}. If $p, q \in \bigcup_{1 \leq k \leq m_0} V_k$, this follows by induction hypothesis. If $p,q \in V_{m_0 + 1}$, this is trivial. It remains the case $p \in \bigcup_{1 \leq k \leq m_0} V_k$ and $q \in V_{m_0 +1}$ (or  the other way around). By \eqref{NiceOrder} there is a $r \in \Big(\bigcup_{1 \leq k \leq m} V_k \Big) \cap V_{m+1} $. The distance from $p$ to $r$ is by the induction hypothesis bounded by $\sum_{k=1}^{m_0} \diam(V_k) $, while the distance from $r$ to $q$ is bounded by $\diam(V_{m_0 + 1})$. The triangle inequality then concludes the proof.
\vspace*{2mm}

\underline{\textbf{Step 2:}} We show that for each $1 \leq k \leq K$ there exists a constant $0 < C_k < \infty$ such that $\diam(\Sigma \cap U_k) \leq C_k$ holds for all Cauchy hypersurfaces $\Sigma$ of $N$.
\vspace*{2mm}

Let $\Sigma$ be a Cauchy hypersurface of $N$ and let $\psi_k : U_k \to D_k = \bigcup_{\ux \in B_k} I_{\ux} \times \{\ux\}$ be a regular flow chart for $N$. Again, we set $V_k := \Sigma \cap U_k$, which is a smooth submanifold of $N$. Since the timelike curves $I_{\ux} \ni s \mapsto \psi_k^{-1}(s ,\ux)$ are inextendible in $N$ for each $\ux \in B_k$, they intersect $\Sigma$ - and hence $V_k$ - exactly once. This defines a function $f : B_k \to \R$ with the property that $\psi_k^{-1}\big(f(\ux), \ux\big) \in \Sigma$.  Moreover, we define $\omega : B_k \to \psi_k(V_k)$ by $\omega(\ux) = \big( f(\ux), \ux\big)$.
\vspace*{2mm}

\textbf{Step 2.1:} We show that $\omega^{-1}$ is a global chart for $\psi_k(V_k)$.
\vspace*{2mm}

It suffices to show that $f : B_k \to \R$ is smooth.
Let $\ux_0 \in B_k$. Since $\psi_k(V_k)$ is a smooth submanifold of $D_k$, there exists a smooth submersion $g : W \to \R$, where $W \subseteq D_k$ is an open neighbourhood of $\big(f(\ux_0), \ux_0\big)$, such that $\psi_k(V_k) \cap W = \{ g=0\}$. Since $\Sigma$ is a Cauchy hypersurface, the timelike vector field $\partial_0$ can be nowhere tangent to $\psi_k(V_k)$.\footnote{\label{TangentSpace}Note that the tangent plane of a smooth Cauchy hypersurface cannot contain timelike vectors, since otherwise we could find a small timelike curve lying in the Cauchy hypersurface - which is clearly a contradiction to the definition of a Cauchy hypersurface.} It thus follows that $\partial_0 g|_{\big(f(\ux_0), \ux_0\big)} \neq 0$. By the implicit function theorem, there is now a smooth function $h : X \to \R$, where $X \subseteq B_k$ is an open neighbourhood of $\ux_0$, such that $g\big(h(\ux),\ux\big) =0$. Thus, we must have $f|_X = h$ - and hence, $f$ is smooth.
\vspace*{2mm}

\textbf{Step 2.2:} We show that there exists a $0< C_\mathrm{slope} < \infty$ such that $|\partial_i f(\ux)| \leq C_\mathrm{slope} $ holds for all $\ux \in B_k$  and for all $i = 1, \ldots, d$.
\vspace*{2mm}

Since the tangent space of a Cauchy hypersurface does not contain timelike vectors, we obtain the inequality
\begin{equation}
\label{Slope}
0 \leq g\Big((\partial_i f) \partial_0 + \partial_i, (\partial_i f) \partial_0 + \partial_i\Big) = (\partial_i f)^2 g_{00} + 2 (\partial_i f) g_{0i} + g_{ii} 
\end{equation}
for all $i = 1, \ldots, d$.
Equality in \eqref{Slope} holds for  
\begin{equation*}
(\partial_i f)_\pm =  \frac{-\,g_{0i} \mp \sqrt{(g_{0i})^2 - g_{ii} g_{00}}}{g_{00}} \;.
\end{equation*}
Recall that $g_{00} \leq c  < 0$ and $|g_{\mu \nu}| \leq C$. Hence, we obtain the uniform bound
\begin{equation*}
\max\Big{\{}\big{|}(\partial_i f)_+\big{|} \, , \,\big{|} (\partial_i f)_-\big{|} \Big{\}} \leq C_\mathrm{slope}  \;,
\end{equation*}
where $0< C_\mathrm{slope}  < \infty$ is a constant depending on $C$ and $c$ (but not on $f$).
Moreover, together with $g_{00} <0$, the inequality \eqref{Slope} implies
\begin{equation*}
(\partial_i f)_- \leq (\partial_i f) \leq (\partial_i f)_+ 
\end{equation*}
and thus $|\partial_i f| \leq C_\mathrm{slope} $ for all $i = 1, \ldots, d$.
\vspace*{2mm}

The ambient metric $g$ on $D_k$ induces a metric $\overline{g}$ on $\psi_k(V_k)$.
Note that this metric is not necessarily positive definite - it might be degenerate. Its components with respect to the chart $\omega^{-1} $ are denoted by $\overline{g}_{ij}$, where $i,j \in \{1, \ldots, d\}$.
\vspace*{2mm}

\textbf{Step 2.3:} We show that there exists a constant $0 < C_{\overline{g}} < \infty$ such that for all  $\ux \in B_k$ and for all vectors $Z \in \R^d$ we have $\overline{g}_{ij}(\ux) Z^i Z^j \leq C_{\overline{g}} \cdot e_{ij}Z^i Z^j$, where $e$ is the Euclidean metric on $B_k \subseteq \R^d$.
\vspace*{2mm}

We compute
\begin{equation*}
\overline{g}_{ij} = \big(\omega^*g\big)_{ij} = g_{\mu \nu} \frac{\partial \omega^\mu}{\partial x_i} \frac{\partial \omega^\nu}{\partial x_j} = g_{00} \frac{\partial f}{\partial x_i} \frac{\partial f}{\partial x_j} + g_{0j}\frac{\partial f}{\partial x_i} + g_{i0}\frac{\partial f}{\partial x_j} + g_{ij} \;.
\end{equation*}
It now follows from the uniform bound $|g_{\mu \nu}| \leq C$ and Step 2.2 that there exists a constant $0 < C' < \infty$ such that $|\overline{g}_{ij}| \leq C'$ holds for all  $i,j = 1, \ldots, d$. This concludes Step 2.3.
\vspace*{2mm}

We can now finish Step 2. Let $\omega(\ux), \omega(\underline{y}) \in \psi_k(V_k)$, with $\ux, \underline{y} \in B_k$, be given. There exists a curve $\gamma : [0,1] \to B_k$ with $\gamma(0) = \ux$ and $\gamma(1) = \underline{y}$ of coordinate-length 
\begin{equation*}
L_e(\gamma) = \int_0^1 \sqrt{e\big(\dot{\gamma}(s), \dot{\gamma}(s)\big)} \, ds \leq \diam_e(B_k) + 1 \;.
\end{equation*}
It then follows from Step 2.3 that 
\begin{equation*}
\begin{split}
L_{\overline{g}}(\gamma) &= \int_0^1 \sqrt{\overline{g}\big( \dot{\gamma}(s), \dot{\gamma}(s)\big)} \, ds \\
&\leq \sqrt{C_{\overline{g}}} \int_0^1 \sqrt{e\big( \dot{\gamma}(s), \dot{\gamma}(s)\big)} \, ds \\
&\leq \sqrt{C_{\overline{g}}} \cdot (\diam_e(B_k) + 1) \;.
\end{split}
\end{equation*}
Note that this bound does not depend on the points $\ux, \underline{y} \in B_k$. 
This concludes Step 2 with $C_k = \sqrt{C_{\overline{g}}} \cdot \big(\diam_e(B_k) + 1\big)$ and hence the proof of Theorem \ref{DiamFinite}.
\end{proof}

The next theorem is not needed for the proof of Theorem \ref{Inex}. However, its proof illustrates how an embedding of a globally hyperbolic Lorentzian manifold $N$ into a larger Lorentzian manifold $M$, such that $N$ is precompact in $M$, allows us to construct regular flow charts for $N$. Later, in the proof of Theorem \ref{Inex}, we morally encounter the more complex situation of finitely many embeddings of a globally hyperbolic Lorentzian manifold $N$ into larger Lorentzian manifolds. The image of $N$ is not precompact in any of the larger manifolds, but taking all the embeddings into account, we can recover some sort of compactness and construct regular flow charts. This, however, is done `by hand' in the proof of Theorem \ref{Inex}. Thus, the reader only interested in the proof of Theorem \ref{Inex} can skip directly to Section \ref{SecPf}.

\begin{theorem}
Let $(M,g)$ be a time oriented Lorentzian manifold with a $C^0$-regular metric and $N \subseteq M$ an open and globally hyperbolic subset. Moreover assume that $N$ is precompact in $M$ and that $\psi : \R^d \supseteq B_2(0) \hookrightarrow M$ is a smooth embedding of $B_2(0)$ such that $\psi|_{B_1(0)} : B_1(0) \hookrightarrow N \subseteq M$ is a Cauchy hypersurface in $(N,g)$.

Then one has $\diam_s(N) < \infty$.
\end{theorem}

Let us remark that we have assumed the existence of a global chart, and hence a trivial topology, of the Cauchy hypersurface for $N$ merely for simplicity of exposition. The reader is invited to write down more general formulations of this theorem.

\begin{proof}
The proof is divided into three steps.
\vspace*{2mm}

\underline{\textbf{Step 1:}} The set-up.
\vspace*{2mm}

Since $(M,g)$ is time-oriented, we can find a smooth globally timelike and future directed vector field $T$ on $M$. We denote the flow of $T$ by $\Phi_{(\cdot)} (\cdot) : \R \times M \subseteq \D \to M$, where $\D$ denotes the maximal domain of the flow. 

Let us denote the maximal time interval of existence of the integral curve of $T$ in $M$, starting at $p \in M$, by $I_{p,M}$. It follows that $\D = \bigcup_{p \in M} I_{p,M} \times \{p\}$. Moreover let us define the maximal time interval of existence of the integral curve of $T$ in $N$, starting at $q \in N$, by $I_{q,N}$.

Setting $\D_{\psi, 2} := \bigcup_{\underline{x} \in B_2(0)} I_{\psi(\underline{x}),M} \times \{\underline{x}\} \subseteq \R \times B_2(0)$, we define $\chi_2 : \D_{\psi, 2} \to M$ by 
\begin{equation*}
\chi_2(s,\underline{x}) := \Phi_s\big(\psi(\underline{x})\big) \;.
\end{equation*}
Clearly, $\chi_2$ is smooth. Define now $\D_{\psi, 1} := \bigcup_{\underline{x} \in B_1(0)} I_{\psi(\underline{x}), N} \times \{\underline{x}\} \subseteq \R\times B_1(0) \subseteq \D_{\psi, 2}$. We claim that
\begin{equation*}
\chi_1 := \chi_2|_{\D_{\psi,1}} : \D_{\psi, 1} \to N
\end{equation*}
is a diffeomorphism. This is seen as follows:

\begin{itemize}
\item \textbf{Surjectivity:}  Let $p \in N$. Since $N$ is globally hyperbolic, the maximal integral curve of $T$, starting at $p$, must intersect the Cauchy hypersurface $ \Sigma := \psi\big(B_1(0)\big)$ after time $-s_0 \in \R$ at $\psi(\underline{x}_0) \in \Sigma$ (say). We thus have $p = \chi_1(s_0,\underline{x}_0)$.

\item \textbf{Injectivity:} Assume $\chi_1(s_0,\underline{x}_0) = \chi_1(s_1, \underline{x}_1)$ with $s_0, s_1 \in \R$ and $\underline{x}_0, \underline{x}_1 \in B_1(0)$. It follows that $\psi(\ux_0)$ and $\psi(\ux_1)$ lie on the same integral curve of $T$. Since every integral curve of $T$ intersects $\Sigma$ only once, it follows that $\ux_0 = \ux_1$. Moreover, if $s_0 \neq s_1$, then there is a closed integral curve of $T$ in $N$, which contradicts the global hyperbolicity of $N$. Thus, we have $s_0 = s_1$.

\item \textbf{$\chi_1$ is a local diffeomorphism:} Let $(s_0, \ux_0) \in \D_{\psi, 1}$ be given. We have
\begin{equation*}
(\chi_1)_* |_{(s_0,\ux_0)} \big(\frac{\partial}{\partial s}\big) = T_{\Phi_{s_0}\big(\psi(\ux_0)\big)} = (\Phi_{s_0})_*|_{\psi(\ux_0)} (T)
\end{equation*}
and
\begin{equation*}
(\chi_1)_*|_{(s_0,\ux_0)}\big(\frac{\partial}{\partial x_i} \big) = (\Phi_{s_0})_*|_{\psi(\ux_0)} \big(\psi_*(\frac{\partial}{\partial x_i}) \big) \quad \textnormal{ for } i= 1, \ldots, d \:.
\end{equation*}
Since $\Phi_{s_0}$ is a diffeomorphism and $\big(T, \psi_*(\frac{\partial}{\partial x_1}), \ldots, \psi_*(\frac{\partial}{\partial x_d})\big)$ are linearly independent at $\psi(\ux_0)$, it follows that $(\chi_1)_*|_{(s_0, \ux_0)}$ is surjective.
\end{itemize}
Hence,  we have shown that $\chi_1$ is a diffeomorphism - and thus we obtain a global coordinate system $\chi_1^{-1} : N \to \D_{\psi, 1} \subseteq \R \times B_1(0)$ for $N$.

Pulling back the metric $g$ via $\chi_2$, we obtain a continuous $2$-covariant tensor field $\chi_2^*g$ on $\D_{\psi,2}$. The components $\chi_2^*g(\partial_\mu, \partial_\nu)$ are continuous functions on $\D_{\psi,2}$ and agree on $\D_{\psi,1}$ with the metric components $g_{\mu \nu}$ of the metric $g$ in the chart $\chi_1^{-1}$.
\vspace*{2mm}

\underline{\textbf{Step 2:}} We show that $\overline{\D_{\psi,1}} \subseteq \R^{d+1}$ is compact and that $\overline{\D_{\psi,1}}  \subseteq \D_{\psi, 2}$. Thus, there exists a constant $C_g >0$ such that  $|\chi_2^*g(\partial_\mu, \partial_\nu) (x)| \leq C_g$ holds for all $x \in \overline{\D_{\psi, 1}}$. In particular, the metric components $g_{\mu \nu}$ in the chart $\chi_1^{-1}$ are uniformly bounded.
\vspace*{2mm}

We define the \emph{future boundary function} $h_f : B_1(0) \to [0, \infty] $ by
\begin{equation*}
h_f(\ux) := \sup I_{\psi(\ux),N}
\end{equation*}
and the \emph{past boundary function} $h_p : B_1(0) \to [-\infty, 0] $ by
\begin{equation*}
h_p(\ux) := \inf I_{\psi(\ux),N}\;.
\end{equation*}

\textbf{Step 2.1} The future boundary function $h_f$ actually maps into $[0,\infty)$ and the past boundary function $h_p$ maps into $(-\infty,0]$.
\vspace*{2mm}

So let $\ux \in B_1(0)$. We show that $h_f(\ux) < \infty$ - the analogous statement for the past boundary function follows by reversing the time orientation. 

The proof is by contradiction. Hence, let us assume that the integral curve $s \mapsto \gamma(s) := \Phi_s\big(\psi(\ux)\big)$ is defined for $s \in [0,\infty)$ and is contained in $N$. We define the sequence $\{s_n\}_{n \in \N} \subseteq [0,\infty)$ by $s_n := n$. Since $\overline{N} \subseteq M$ is compact, it follows that, after possibly taking a subsequence, there exists a $p \in \overline{N}$ such that $\gamma(s_n) \to p$ for $n \to \infty$. We distinguish the following three cases:

\begin{enumerate}
\item \textbf{$p \in N$:} By Proposition \ref{FlowChart} there exists a flow chart $\varphi : U \to (-\Delta, \Delta) \times (-E,E)^d$ centred at $p$, with timelike connected orbits. 
For $n$ large enough, we have $\gamma(s_n), \gamma(s_{n+1}) \in U$ - and moreover we can without loss of generality assume that they lie on different orbits of $T$ when restricted to $U$. We now connect the orbit $\gamma(s_{n+1})$ is lying on by a future directed timelike curve with the orbit $\gamma(s_n)$ is lying on. 
\begin{center}
\def\svgwidth{6cm}
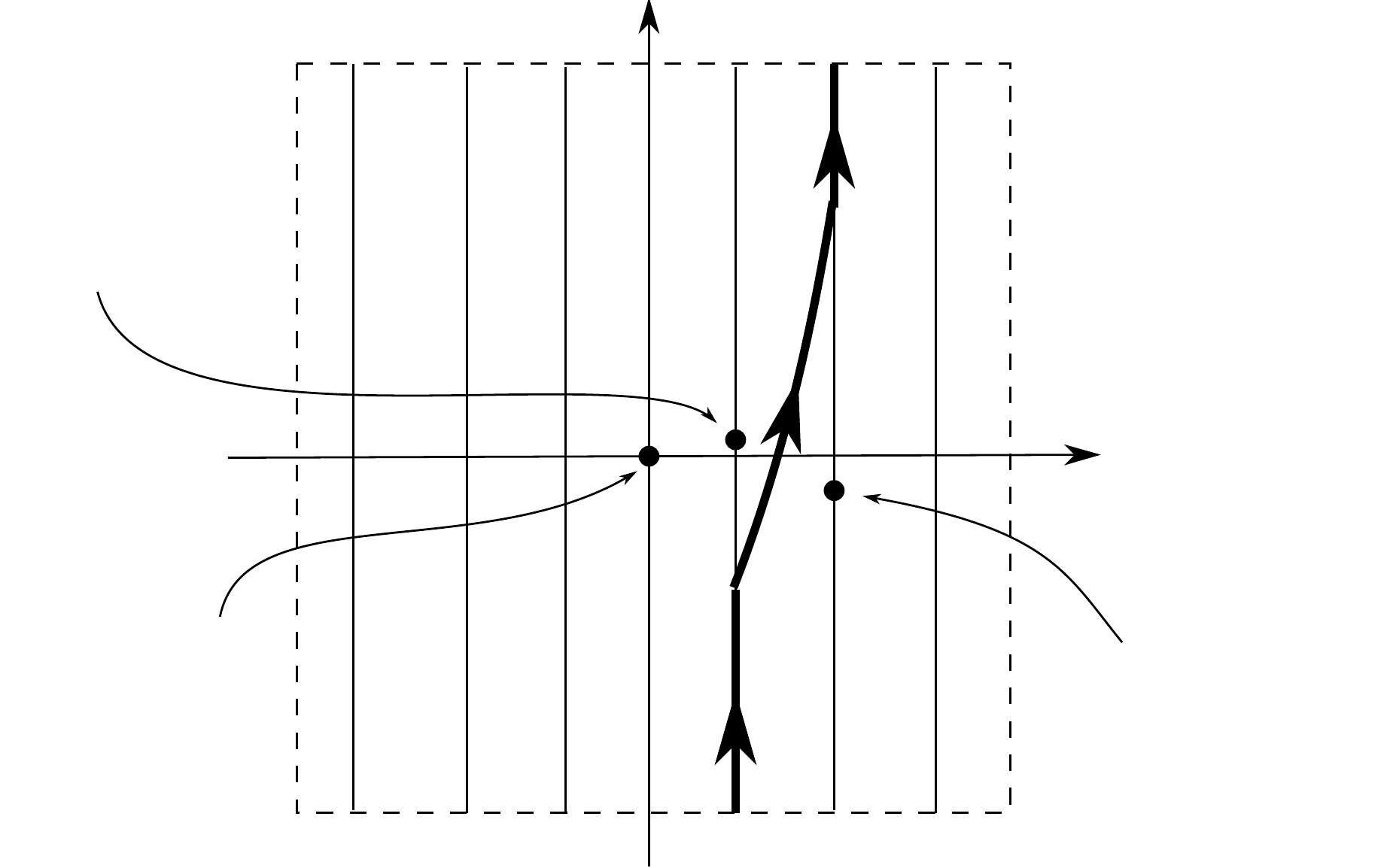
\end{center}
This gives rise to a closed timelike curve in $N$, which contradicts the global hyperblicity of $N$.

\item \textbf{$p \in \partial N \setminus \overline{\Sigma}$:}
Let $\varphi : M \supseteq U \to (-\Delta, \Delta) \times (-E,E)^d$ be a flow chart centred at $p$, with timelike connected orbits, such that $U$ is disjoint from $\Sigma$. For $n$ big enough, we have $\gamma(s_n) \in U$. Clearly, the orbit $\Gamma \subseteq U$ of $T$ restricted to $U$, on which $\gamma(s_n)$ lies, lies in $I^+(\Sigma, N)$. Thus, every past-inextendible timelike curve in $N$, which starts at a point of the orbit $\Gamma$ in $U$, has to intersect $\Sigma$. However, by the proof of Proposition \ref{FlowChart}, we can find a past directed timelike curve from a point of $\Gamma$ to $p$. 
\begin{center}
\def\svgwidth{6cm}
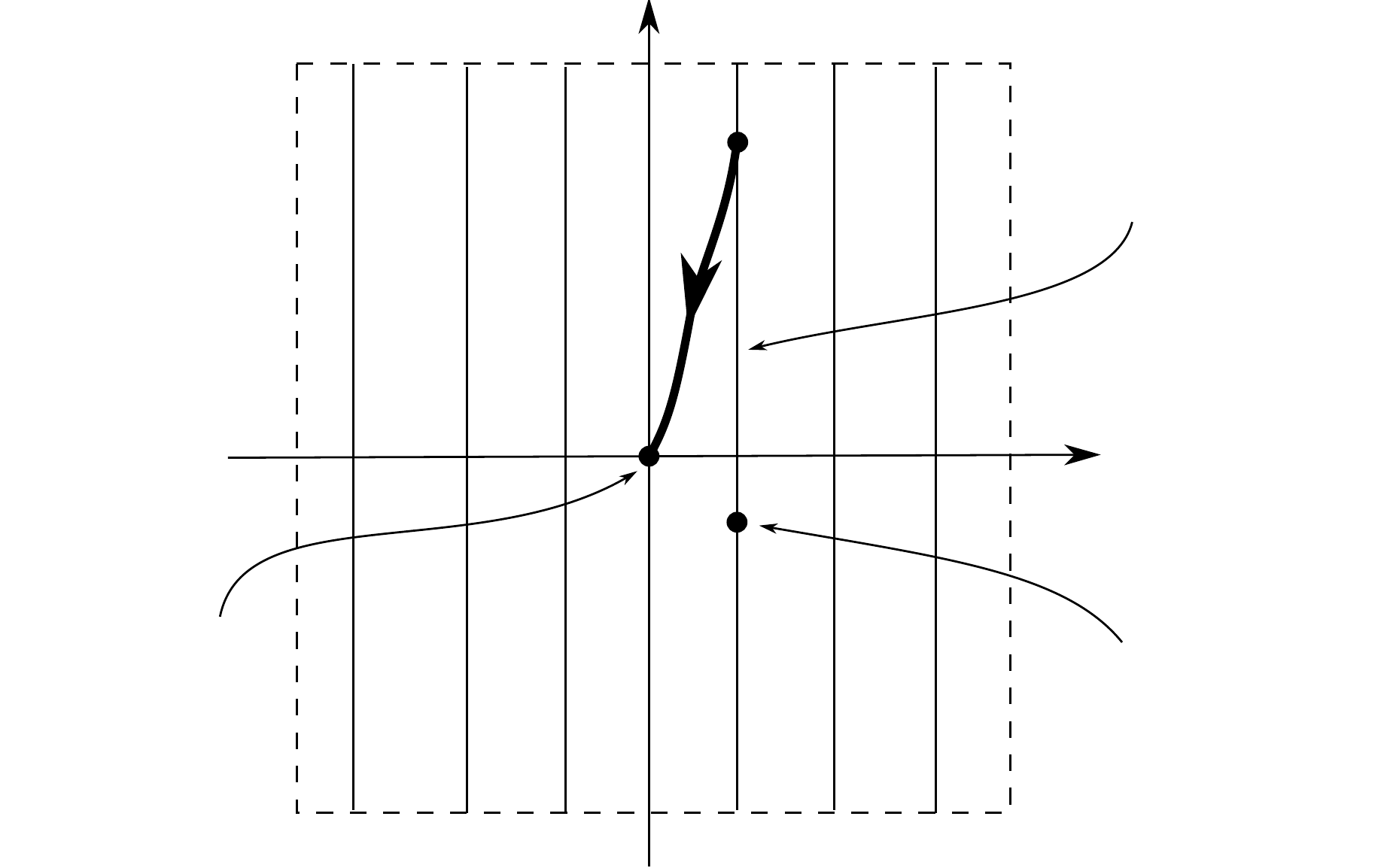
\end{center}
Since the Cauchy hypersurface $\Sigma$ is disjoint from $U$, this gives rise to an inextendible timelike curve in $N$ which does not intersect $\Sigma$ - a contradiction to the global hyperbolicity of $N$.

\item \textbf{$p \in \overline{\Sigma} \setminus \Sigma$:} There is an $\ux_0 \in \partial B_1(0) \subseteq B_2(0)$ with $\psi(\ux_0) = p$.  As already pointed out in Footnote \ref{TangentSpace}, the tangent space of a smooth Cauchy hypersurface cannot contain timelike vectors. Hence, also the tangent space $T_{\psi(\ux_0)}\psi\big(B_2(0)\big)$ does not contain timelike vectors. Also using that $\psi\big(B_2(0)\big)$ is an embedded submanifold, we can thus find flow coordinates $\varphi : U \to (-\Delta, \Delta) \times (-E,E)^d$, centred at $p$, such that 
\begin{equation}
\label{SliceCoord}
\varphi\big(\psi(B_2(0) \cap U\big) = \{0\} \times (-E,E)^d\;.
\end{equation}
It now follows from the continuity of the flow $\Phi$ that $\gamma(s_n + \frac{\Delta}{2}) = \Phi_{\frac{\Delta}{2}}\big(\gamma(s_n)\big) \to \varphi^{-1}\big((\frac{\Delta}{2}, 0 , \ldots, 0)\big)$ for $n \to \infty$. However, $\varphi^{-1}\big((\frac{\Delta}{2}, 0 , \ldots, 0)\big) \in \overline{N}$ cannot lie in $\overline{\Sigma}$ by \eqref{SliceCoord} - thus it must either lie in $\partial N \setminus \overline{\Sigma}$ or in $N$, which brings us back to the contradiction derived in the previous two cases.
\end{enumerate}
This finishes the proof of $h_f(\ux) < \infty$.
\vspace*{2mm}

\textbf{Step 2.2} We show that for $\ux \in B_1(0)$ we have $\overline{I_{\psi(\ux),N}} \subseteq I_{\psi(\ux), M}$, where the closure is in $\R$.
\vspace*{2mm}

This follows since $\overline{N} \subseteq M$ is compact and an integral curve of $T$ cannot break down without leaving every compact subset of $M$.
\vspace*{2mm}

\textbf{Step 2.3} We show that $h_f : B_1(0) \to [0,\infty)$ is continuous. Moreover, $h_f$ can be continuously extended to $\overline{B_1(0)}$ by defining $h_f(\ux) := 0$ for $\ux \in \partial B_1(0)$. Similarly, after the same extension, $h_p : \overline{B_1(0)} \to (-\infty,0]$ is continuous.
\vspace*{2mm}

Without loss of generality, we only show the statement for the future boundary function $h_f$. We begin with the continuity in $B_1(0)$.

The proof is by contradiction - we assume that $h_f$ is not continuous at $\ux_\infty \in B_1(0)$. Thus, there exists a $\delta >0$ and a sequence $\{\ux_n\}_{n \in \N} \subseteq B_1(0)$ with $\ux_n \to \ux_\infty$ for $n \to \infty$ such that $|h_f(\ux_n) - h_f(\ux_\infty)| > \delta$ holds for all $n \in \N$. 

By Step 2.2, we have that $\big(h_f(\ux_\infty), \ux_\infty\big) \in \D_{\psi,2}$. In the same way we have shown that $\chi_1$ is a local diffeomorphism, it follows that $\chi_2|_{\D_{\psi,2} \cap \big(\R \times B_1(0)\big)}$ is a local diffeomorphism. Thus, there exists a small neighbourhood $U \subseteq B_1(0)$ of $\ux_\infty$ and a $\Delta >0$  such that, if we set $W :=\big(h_f(\ux_\infty) - \Delta, h_f(\ux_\infty) + \Delta) \times U$, we have that $\chi_2|_W : W \to \chi_2(W)$ is a diffeomorphism. After possibly making $\Delta$ smaller, we can assume without loss of generality that $0 < \Delta < \delta$ and that $h_f(\ux_\infty) - \Delta >0$.
We now distinguish two cases:
\begin{enumerate}
\item There exists a subsequence $\{\ux_{n_k}\}_{k \in \N}$ with $n_k \to \infty$ for $k \to \infty$ such that $h_f(\ux_{n_k}) > h_f(\ux_\infty) + \delta$ holds for all $k \in \N$. By Proposition \ref{FutureOpen} it follows that for $k_0 \in \N$ large enough we have 
\begin{equation*}
\big( h_f(\ux_\infty) + \frac{\Delta}{2}, \ux_{n_{k_0}}\big) \in I^+\Big(\big(h_f(\ux_\infty), \ux_\infty\big), W\Big) \;.
\end{equation*}
We can thus find a past directed timelike curve $\sigma : [0,1] \to W$ with $\sigma(0) =\big( h_f(\ux_\infty) + \frac{\Delta}{2}, \ux_{n_{k_0}})\big)$ and $\sigma(1) =(h_f(\ux_\infty), \ux_\infty)$. 
\begin{center}
\def\svgwidth{8cm}
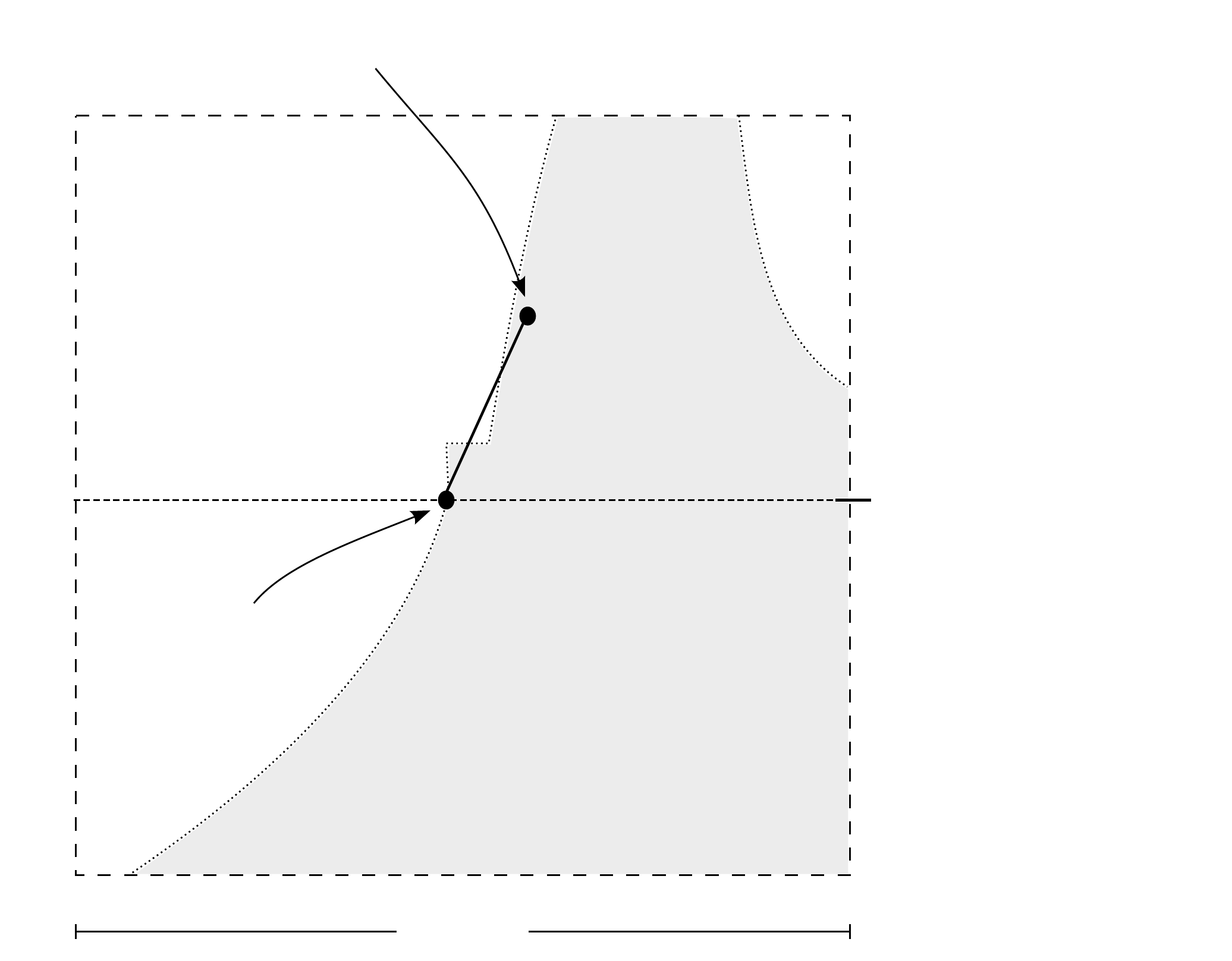
\end{center}
Let $s_0 := \sup\{s \in [0,1] \, | \, \sigma(s') \in \D_{\psi,1} \textnormal{ for all } 0 \leq s' < s\}$. It follows that $\chi_2 \circ \sigma : [0,s_0) \to N$ is a past-inextendible timelike curve in $N$ with $(\chi_2 \circ \sigma)(0) \in I^+(\Sigma, N)$. However, this curve does not intersect $\Sigma$, since $\chi_1^{-1}(\Sigma) = \{0\} \times B_1(0)$, but $\sigma_0 (s) >0$ for all $s \in [0,s_0)$, where $ \sigma = (\sigma_0, \sigma_1, \ldots, \sigma_d)$. This contradicts the global hyperbolicity of $N$.

\item There exists a subsequence $\{\ux_{n_k}\}_{k \in \N}$ with $n_k \to \infty$ for $k \to \infty$ such that $h_f(\ux_{n_k}) < h_f(\ux_\infty) - \delta$ holds for all $k \in \N$. Again, by Proposition \ref{FutureOpen} we can choose $k_0 \in \N$ large enough such that
\begin{equation*}
\big(h_f(\ux_\infty) - \frac{\Delta}{2}, \ux_{n_{k_0}}\big) \in I^-\Big(\big(h_f(\ux_\infty) - \frac{\Delta}{10}, \ux_\infty\big),W\Big) \;.
\end{equation*}
We can thus find a past directed timelike curve $\sigma : [0,1] \to W$ with $\sigma(0) =\big( h_f(\ux_\infty) - \frac{\Delta}{10}, \ux_\infty\big)$ and $\sigma(1) =(h_f(\ux_\infty) - \frac{\Delta}{2}, \ux_{n_{k_0}})$. 

Let $s_0 := \sup\{s \in [0,1] \, | \, \sigma(s') \in \D_{\psi,1} \textnormal{ for all } 0 \leq s' < s\}$. It follows that $\chi_2 \circ \sigma : [0,s_0) \to N$ is a past-inextendible timelike curve in $N$ with $(\chi_2 \circ \sigma)(0) \in I^+(\Sigma, N)$. The same argument as in the previous case shows that this curve does not intersect the Cauchy hypersurface $\Sigma$ - which contradicts again the global hyperbolicity of $N$.
\end{enumerate}
This shows that the future boundary function $h_f$ is continuous in $B_1(0)$.

We proceed by showing that for all sequences $\{\ux_n\}_{n \in \N} \subseteq B_1(0)$ with $\ux_n \to \ux_\infty \in \partial B_1(0)$ for $n \to \infty$ we have $h_f(\ux_n) \to 0$ for $n \to \infty$, which then implies the continuity of $h_f$ on $\overline{B_1(0)}$. 

This proof is also by contradiction - thus assume there is a $\delta >0$  and a sequence $\{\ux_n\}_{n \in \N} \subseteq B_1(0)$ with $\ux_n \to \ux_\infty \in \partial B_1(0)$ for $n \to \infty$ such that $h_f(\ux_n) > \delta$ holds for all $n \in \N$.

As argued in the third case of the proof of Step 2.1, there exists a $\Delta >0$ and a neighbourhood $U \subseteq B_2(0)$ of $\ux_\infty$ such that $\chi_2|_W : W \to \chi_2(W)$ is a diffeomorphism, where $W:= (-\Delta, \Delta) \times U \subseteq \D_{\psi,2}$. Without loss of generality we can again assume that $0 < \Delta < \delta$. By Proposition \ref{FutureOpen} there exists a large $n_0 \in \N$ such that
\begin{equation*}
(\frac{\Delta}{2} , \ux_{n_0}) \in  I^+\big((0,\ux_\infty), W\big) \;.
\end{equation*}
We can thus find a past directed timelike curve $\sigma : [0,1] \to W$ with $\sigma(0) =( \frac{\Delta}{2}, \ux_{n_0})$ and $\sigma(1) =(0, \ux_\infty)$. Moreover, without loss of generality, we can assume that $\sigma(s) \notin \{0\} \times B_1(0)$ for all $s \in [0,1]$, since, by Proposition \ref{FutureOpen}, we can also connect the points $(\frac{\Delta}{3}, \ux_\infty)$ and $(\frac{\Delta}{2}, \ux_{n_0})$ (for $n_0 \in \N$ large enough) by a timelike curve contained in an open set which is disjoint from $\{0\} \times B_1(0)$.

\begin{center}
\def\svgwidth{6cm}
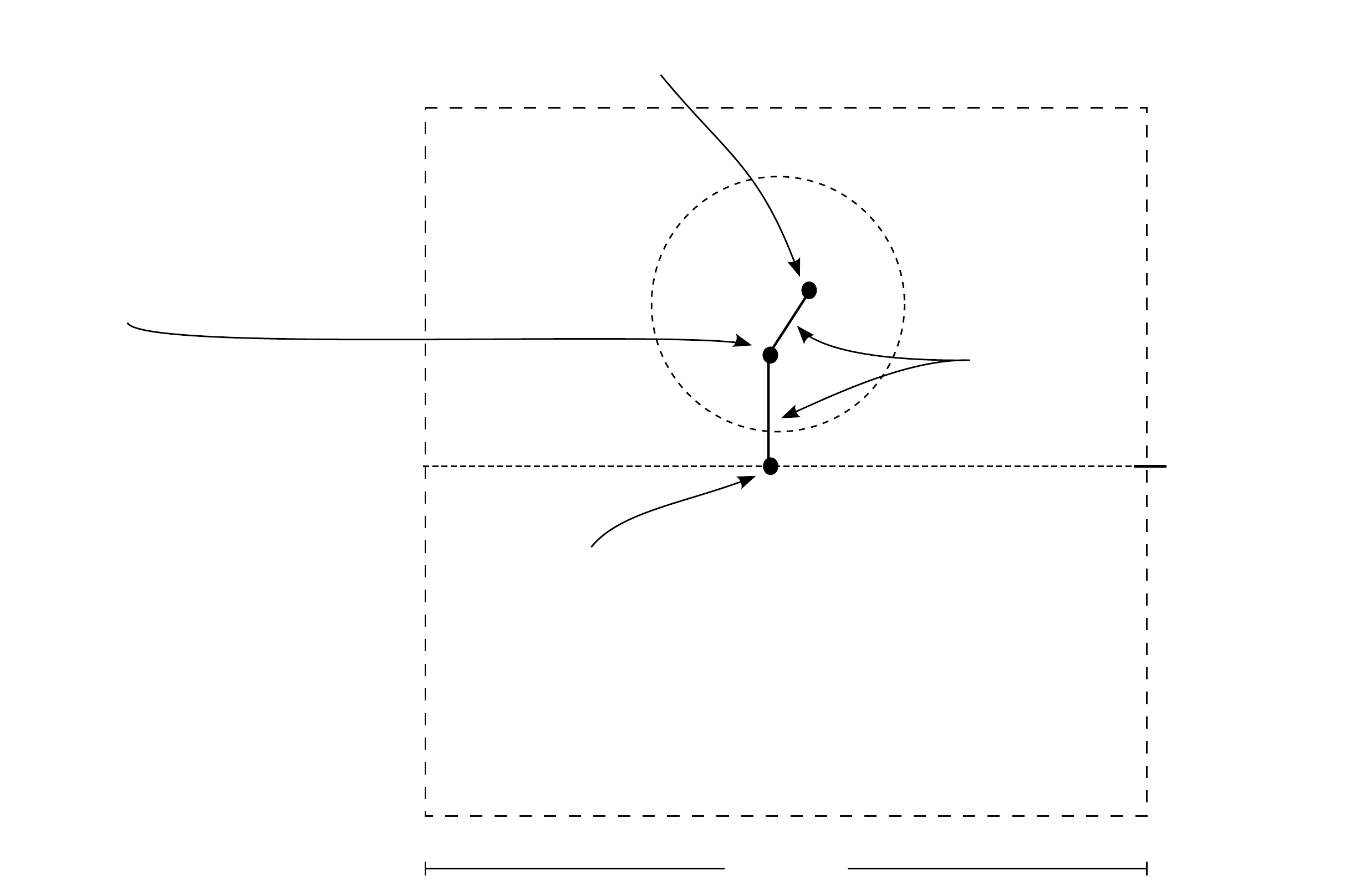
\end{center}

As before, this now gives rise to a past-inextendible timelike curve in $N$, starting in $I^+(\Sigma, N)$, which does not intersect the Cauchy hypersurface $\Sigma$ - again a contradiction to the global hyperbolicity of $N$.
This concludes Step 2.3.
\vspace*{2mm}

We now continue with Step 2. Clearly, we have
\begin{equation*}
\D_{\psi,1} = \{(x_0, \ux) \in \R \times B_1(0) \, | \, h_p(\ux) < x_0 < h_f(\ux)\} \;.
\end{equation*}
The continuity of $h_f : \overline{B_1(0)} \to [0, \infty)$ and of $h_p : \overline{B_1(0)} \to (-\infty, 0]$ implies that the closure of $\D_{\psi_1}$ in $\R^{d+1}$ is given by
\begin{equation*}
\overline{\D_{\psi,1}} = \{(x_0, \ux) \in \R \times \overline{B_1(0)} \, | \, h_p(\ux) \leq x_0 \leq h_f(\ux)\} \;.
\end{equation*}
Moreover, it follows from the continuity of $h_f$ and $h_p$ on $\overline{B_1(0)}$ that there is a constant $0< C < \infty$ such that $h_f \leq C$ and $h_p \geq -C$. Hence, we have $\overline{\D_{\psi,1}} \subseteq [-C,C] \times \overline{B_1(0)}$ and thus $\overline{\D_{\psi,1}}$ is compact.

Finally, it follows from 
\begin{equation*}
\overline{\D_{\psi,1}} = \Big(\bigcup_{\ux \in B_1(0)} \overline{I_{\psi(\ux),N}} \times \{\ux\}\Big) \cup \big(\{0\} \times \partial B_1(0)\big)
\end{equation*}
and Step 2.2, that  $\overline{\D_{\psi,1}} \subseteq \D_{\psi,2}$.
This concludes Step 2.
\vspace*{2mm}

\underline{\textbf{Step 3:}} We appeal to Theorem \ref{DiamFinite}.
\vspace*{2mm}

Since the vector field $T$ is timelike on $\overline{N}$, there exists a constant $c$ such that the metric component $g_{00}$ in the chart $\chi_1^{-1}$ satisfies $g_{00} = g(T,T) \leq c <0$. It is now easy to see that $\chi_1^{-1} : N \to \D_{\psi,1}$ is a global regular flow chart for $N$. The theorem now follows from Theorem \ref{DiamFinite}.
\end{proof}

\section{Proof of Theorem \ref{Inex}}
\label{SecPf}

Before we start with the proof of Theorem \ref{Inex}, let us make some preliminary observations.
First note that if $\sigma : (-s_0, 0) \to \Mint$ is a future directed timelike curve, where $s_0 >0$, we have  $0 > \gint\big(\dot{\sigma}, -\frac{\partial}{\partial r}\big) = -(1 -\frac{2m}{r^{d-2}})^{-1} \dot{\sigma}_r $. It follows that 
\begin{equation}
\label{RVel}
\dot{\sigma}_r <0
\end{equation}
and hence we can parametrise $\sigma$ by the $r$-coordinate. We will often use in the arguments in the proof of Theorem \ref{Inex} (without spelling it out explicitly) that the coordinate value of $r$ along future directed timelike curves can only decrease, while along past directed timelike curves it can only increase. In particular, let us remark that this observation, together with the bound \eqref{BoundTVel} proved below, shows that the surfaces of constant $r$ are Cauchy hypersurfaces of $\Mint$.

Furthermore, the following lemma and proposition are needed in the proof of Theorem \ref{Inex}.

\begin{lemma}
\label{BoundsOnReach}
Let $0<r_0< (2m)^{\frac{1}{d-2}}$. For every $\varepsilon >0$ we can find $0<\tilde{r}_0 < r_0 $ such that for any future directed timelike curve $\sigma : (-r_0,0) \to \Mint$,
\begin{equation*}
\sigma(s) = \big(\sigma_t(s), -s, \sigma_{\omega}(s)\big) \;,
\end{equation*}
where $\sigma_{\omega}$ is the canonical projection of $\sigma$ on the sphere $\Sd$, the following holds:
\begin{equation*}
d_{\Sd}\big(\sigma_{\omega}(s), \sigma_{\omega}(s') \big) < \varepsilon     
\end{equation*}
and
\begin{equation*}
|\sigma_t(s) - \sigma_t(s')| < \varepsilon
\end{equation*}
for all  $-\tilde{r}_0 \leq s,s' <0$.
\end{lemma}
Note that this lemma implies, in particular, that $\sigma_t(s)$ and $\sigma_\omega(s)$ converge for $s \nearrow 0$.

\begin{proof}
Let $\sigma : (-r_0,0) \to \Mint$ be a timelike curve, parametrised as above. We obtain for all $s \in (-r_0,0)$
\begin{equation*}
0 > \gint(\dot{\sigma}(s),\dot{\sigma}(s)) = -\big(1- \frac{2m}{(-s)^{d-2}}\big) (\dot{\sigma}_t(s))^2 + \big(1- \frac{2m}{(-s)^{d-2}}\big)^{-1} + s^2 \, \mathring{\gamma}_{d-1}\big(\dot{\sigma}_{\omega}(s), \dot{\sigma}_{\omega}(s)\big) 
\end{equation*}
and hence
\begin{equation*}
\frac{(-s)^{d-2}}{2m - (-s)^{d-2}}  >  \underbrace{\frac{2m - (-s)^{d-2}}{(-s)^{d-2}} (\dot{\sigma}_t(s))^2}_{\geq 0}  + \underbrace{ s^2 \, \mathring{\gamma}_{d-1}\big(\dot{\sigma}_{\omega}(s), \dot{\sigma}_{\omega}(s)\big)}_{\geq 0} \;.
\end{equation*}
It follows that
\begin{equation}
\label{BoundTVel}
|\dot{\sigma}_t(s)| < \frac{(-s)^{d-2}}{2m - (-s)^{d-2}} 
\end{equation}
and
\begin{equation}
\label{BoundAngVel}
||\dot{\sigma}_\omega(s)||_{\Sd} < \frac{(-s)^{\nicefrac{d}{2} - 2}}{\big[2m - (-s)^{d-2}\big]^{\nicefrac{1}{2}}}
\end{equation}
holds for all $s \in (-r_0,0)$. Since $d \geq 3$, it follows that both upper bounds are integrable on $(-r_0,0)$. The lemma now follows from integration.
\end{proof}

\begin{proposition}
\label{IntFutCon}
The interior of the Schwarzschild spacetime $(\Mext, \gext)$ is future one-connected.
\end{proposition}

\begin{proof}
Let $(t_0, r_0, \omega_0)$, $(t_1, r_1, \omega_1) \in \Mint$ with $(t_0, r_0, \omega_0) \ll (t_1, r_1, \omega_1)$, and let $\gamma_i : [r_1, r_0] \to \Mint$, 
\begin{equation*}
\gamma_i(r) = \big( (\gamma_i)_t(r) , r, (\gamma_i)_\omega(r) \big)\;,
\end{equation*}
be past directed timelike curves with $\gamma_i (r_1) = (t_1, r_1, \omega_1)$ and $\gamma_i(r_0) = (t_0, r_0, \omega_0)$, where $i \in  \{1,2\}$. We need to show that there exists a timelike homotopy with fixed endpoints between $\gamma_1$ and $\gamma_2$.

We give the detailed proof under the assumption
\begin{equation}
\label{Asump}
\mathrm{Im}\big((\gamma_i)_\omega\big) \subseteq \Sd \setminus \{ - \omega_1\}\;.
\end{equation} 
Note that \eqref{Asump} is implied by Lemma \ref{BoundsOnReach} if we choose $r_0$ sufficiently close to $0$. Moreover, in this paper we will only need the statement that $\Mint \cap \{ 0 < r < \tilde{r}\}$ is future one-connected for \emph{some} $\tilde{r} >0$, and hence one can choose $\tilde{r} >0$ small enough such that \eqref{Asump} is satisfied. However, the general case can be proven by a perturbation argument analogous to the one presented in the proof of Proposition \ref{ExtFutCon}.

Consider the exponential map $\exp_{\omega_1} : T_{\omega_1} \Sd \supseteq B_\pi (0) \to \Sd$ based at $\omega_1$, and define $\Gamma_i : [r_1, r_0] \times [r_1, r_0] \to \Mint$, $i \in \{1,2\}$, by
\begin{equation*}
\Gamma_i(u,r) := \begin{cases} \Big( (\gamma_i)_t(r) , r , \exp_{\omega_1} \big[ f_i(u,r) \exp_{\omega_1}^{-1} \big( (\gamma_i)_\omega(u)\big) \big] \Big) \quad &\textnormal{ for } r_1 \leq r \leq u \\
\gamma_i(r) &\textnormal{ for } u \leq r \leq r_0\;,\end{cases}
\end{equation*}
where
\begin{equation*}
f_i(u,r) = \frac{\int_{r_1}^r || (\dot{\gamma}_i)_\omega(r') ||_{\Sd} \, dr' }{ \int_{r_1}^u || (\dot{\gamma}_i)_\omega(r')||_{\Sd} \, dr'} \;.
\end{equation*}
We compute for $r_1 \leq r \leq u$
\begin{equation*}
\begin{split}
\gint(\partial_r \Gamma_i, \partial_r \Gamma_i)(r,u) &= - D(r)\big[(\dot{\gamma}_i)_t(r)\big]^2 + \frac{1}{D(r)} + r^2 | \partial_r f_i(u,r)|^2 \cdot ||\exp_{\omega_1}^{-1}\big((\gamma_i)_\omega(u)\big)||^2_{\Sd} \\
&=- D(r)\big[(\dot{\gamma}_i)_t(r)\big]^2 + \frac{1}{D(r)} + r^2 ||(\dot{\gamma}_i)_\omega(r)||^2_{\Sd} \cdot \underbrace{\frac{d^2_{\Sd}\big((\gamma_i)_\omega(u), (\gamma_i)_\omega(r_1)\big)}{\Big[\int_{r_1}^u || (\dot{\gamma}_i)_\omega(r') ||_{\Sd}\, dr' \Big]^2}}_{\leq 1}  \\
&\leq \gint\big(\dot{\gamma}_i(r), \dot{\gamma}_i(r)\big) \\
&< 0 \;.
\end{split}
\end{equation*}
Hence, $\Gamma_i$ is a timelike homotopy with fixed endpoints between $\gamma_i$ and a timelike curve $\sigma_i$ whose projection on $\Sd$ lies on the geodesic arc connecting $\omega_0$ with $\omega_1$. We now introduce coordinates on $\Sd$ such that the geodesic $\mathbb{S}^1$ through $\omega_0$ and $\omega_1$ is parametrised by $\varphi \in (0,2\pi)$. It follows that $\sigma_i$ maps into the submanifold $N:= \R \times \big(0, (2m)^{\frac{1}{d-2}}\big) \times \mathbb{S}^1$ of $\Mint$, which carries the induced Lorentzian metric $g_N = - D(r) \, dt^2 + \frac{1}{D(r)} \, dr^2 + r^2 \, d\varphi^2 $. In coordinates, $\sigma_i : [r_1, r_0] \to N \subseteq \Mint$, $i \in \{1,2\}$, is given by 
\begin{equation*}
\sigma_i(r) = \big((\sigma_i)_t (r), r , (\sigma_i)_\varphi(r)\big) \;.
\end{equation*}
We now define $\Gamma : [0,1] \times [r_1, r_0] \to N$ by
\begin{equation*}
\Gamma(u,r) := \Big( (1-u)(\sigma_1)_t (r) + u (\sigma_2)_t(r), r, (1-u) (\sigma_1)_\varphi(r) + u (\sigma_2)_\varphi(r)\Big) \;.
\end{equation*}
Using the convexity of $x \mapsto x^2$, we compute
\begin{equation*}
\begin{split}
g_N\big( \partial_r \Gamma(u,r), \partial_r \Gamma(u,r)\big) &= \underbrace{-D(r)}_{\geq 0}\big[(1-u)(\dot{\sigma}_1)_t(r) + u(\dot{\sigma}_2)_t(r)\big]^2 + \frac{1}{D(r)} + r^2\big[(1-u)(\dot{\sigma}_1)_\varphi(r) + u (\dot{\sigma}_2)_\varphi(r)\big]^2 \\
&\leq - D(r)\Big( (1-u)\big[(\dot{\sigma}_1)_t(r)\big]^2 + u \big[(\dot{\sigma}_2)_t(r)\big]^2\Big) + \frac{1}{D(r)} + r^2 \Big((1-u)\big[(\dot{\sigma}_1)_\varphi(r)\big]^2 + u \big[(\dot{\sigma}_2)_\varphi(r)\big]^2\Big) \\
&=(1-u)\cdot g_N\big(\dot{\sigma}_1(r), \dot{\sigma}_1(r)\big) + u \cdot g_N\big(\dot{\sigma}_2(r), \dot{\sigma}_2(r)\big) \\
&<0 \;,
\end{split}
\end{equation*}
and, hence, $\Gamma$ is a timelike homotopy with fixed endpoints between $\sigma_1$ and $\sigma_2$. This concludes the proof.
\end{proof}

We are now well-prepared to start with the proof of Theorem \ref{Inex}.

\begin{proof}[Proof of Theorem \ref{Inex}:]
The proof is by contradiction. So assume that there exists a $C^0$-extension $\iota : \Mint \hookrightarrow \tilde{M}$ and a timelike curve $\tilde{\gamma} : [-1,0] \to \tilde{M}$ such that $\gamma := \iota^{-1} \circ \tilde{\gamma}|_{[-1, 0)} : [-1, 0) \to \Mint$ is a timelike curve in $\Mint$ with $(r \circ \gamma)(s) \to 0$ for $s \nearrow 0$.
The proof is divided into three main steps.
\vspace*{2mm}

\underline{\textbf{Step 1:}} We construct a neighbourhood $\tilde{U} \subseteq \tilde{M}$ of $\tilde{\gamma}(0)$ together with a chart $\tilde{\psi} : \tilde{U} \to (-\varepsilon, \varepsilon) \times B^d_{\rho + \delta}(0)$, where $\rho, \delta >0$,  that have the following properties:
\begin{enumerate}
\item The metric components in this chart satisfy the following uniform bounds: $|\tilde{g}_{\mu \nu}| \leq C < \infty$ and $\tilde{g}_{00} \leq c < 0$, where $C$ and $c$ are constants.
\item There exists a $\mu >0$ such that $(\tilde{\psi} \circ \iota) \Big( I^+\big( \gamma(-\mu), \Mint\big)\Big) \subseteq (-\varepsilon, \varepsilon) \times B^d_\rho(0)$
\item $\tilde{\psi}^{-1}\Big(( - \varepsilon, - \frac{19}{20} \varepsilon] \times B^d_{\rho + \delta}(0)\Big) \subseteq \iota\Big(I^-\big(\gamma(-\mu), \Mint\big)\Big)$ 
\item For all $\underline{x} \in B^d_{\rho + \delta}(0)$ we have $\sup \big{\{} s_0 \in (-\varepsilon, \varepsilon) \; \big| \; \tilde{\psi}^{-1} (s,\underline{x}) \in \iota(\Mint) \;\; \forall  s \in (-\varepsilon, s_0)\big{\}} < \varepsilon$.
\end{enumerate}
\vspace*{2mm}

\textbf{Step 1.1}
Consider a hypersurface $\{r = r_0\}$ with $0< r\big(\gamma(-1)\big) < r_0 < r_+$ and choose $t_0, t_1 \in \R$ such that every past-inextendible timelike curve in $\Mint$ that goes through a point $q \in I^+\Big( I^-\big(\gamma((-1,0)),\Mint\big) \cap \{ r < r_0\}\Big)$ intersects $\{r=r_0\} \cap \{t_0 < t < t_1\} =:S_\mathrm{aux}$.
Note that $\tilde{\gamma}(0) \notin \overline{\iota({S_\mathrm{aux}})}$. By Lemma \ref{NormalForm}, and after a possible reparametrisation of $\tilde{\gamma}$, there is an $\varepsilon >0$, an open neighbourhood $\tilde{O} \subseteq \tilde{M} \setminus \overline{\iota(S_\mathrm{aux})}$ of $\tilde{\gamma}(0)$, and a chart $\tilde{\varphi} : \tilde{O} \to \ed$ such that
\begin{enumerate}
\item $(\tilde{\varphi} \circ \tilde{\gamma}) (s) = (s,0, \ldots, 0) $ holds for $s \in (-\varepsilon, 0]$
\item $\big| \,\tilde{g}_{\mu \nu}(x) - m_{\mu \nu} \,\big| < \delta $ holds for all $x \in (-\varepsilon , \varepsilon)^{d+1}$\;,
\end{enumerate}
where $\delta >0$ is chosen such that all vectors in $C^+_{\nicefrac{5}{6}}$ are future directed timelike, all vectors in $C^-_{\nicefrac{5}{6}}$ are past directed timelike, and all vectors in $C^c_{\nicefrac{5}{8}}$ are spacelike, where we use the notation introduced in Step 1.1 of the proof of Theorem \ref{MinkInex}. In particular, this fixes a time orientation on $\tilde{O}$ with respect to which $\tilde{\gamma}$ is future directed. Moreover, our choice of $\delta >0$ yields\footnote{Cf.\ Step 1.2 in the proof of Theorem \ref{MinkInex}.} the following inclusion relations for the timelike past and future of a point $x \in (-\varepsilon, \varepsilon)^{d+1}$:
\begin{equation}
\label{EstimatesOnPastAndFuture}
\begin{split}
&\big( x + C^+_{\nicefrac{5}{6}}\big) \cap  \ed \subseteq I^+(x, (-\varepsilon, \varepsilon)^{d+1}) \subseteq \big( x + C^+_{\nicefrac{5}{8}}\big) \cap  \ed \\
&\big( x + C^-_{\nicefrac{5}{6}}\big) \cap  \ed \subseteq I^-(x, (-\varepsilon, \varepsilon)^{d+1}) \subseteq \big( x + C^-_{\nicefrac{5}{8}}\big) \cap  \ed  \;.
\end{split}
\end{equation}
Furthermore, since we have chosen the neighbourhood $\tilde{O}$ to be disjoint from $S_\mathrm{aux}$, the following holds:
\begin{equation}
\label{PastPreserving}
\begin{split}
&\textnormal{Let } \iota(q) \in \tilde{O} \cap \iota(\Mint) \textnormal{ be in the same connectedness component of } \tilde{O} \cap \iota(\Mint) \\ &\textnormal{as }  \tilde{\gamma}\big((-\varepsilon, 0)\big) \textnormal{ and assume that } q \in I^+ \Big( I^-\big(\gamma((-1,0)), \Mint\big) \cap \{r < r_0\}, \Mint\Big) \;. \\ &\textnormal{Then } I^-\big(\iota(q),\tilde{O}\big) \subseteq \iota\big(I^-(q,\Mint) \big) \cap  \tilde{O} \;.
\end{split}
\end{equation}
\begin{center}
\def\svgwidth{8cm}
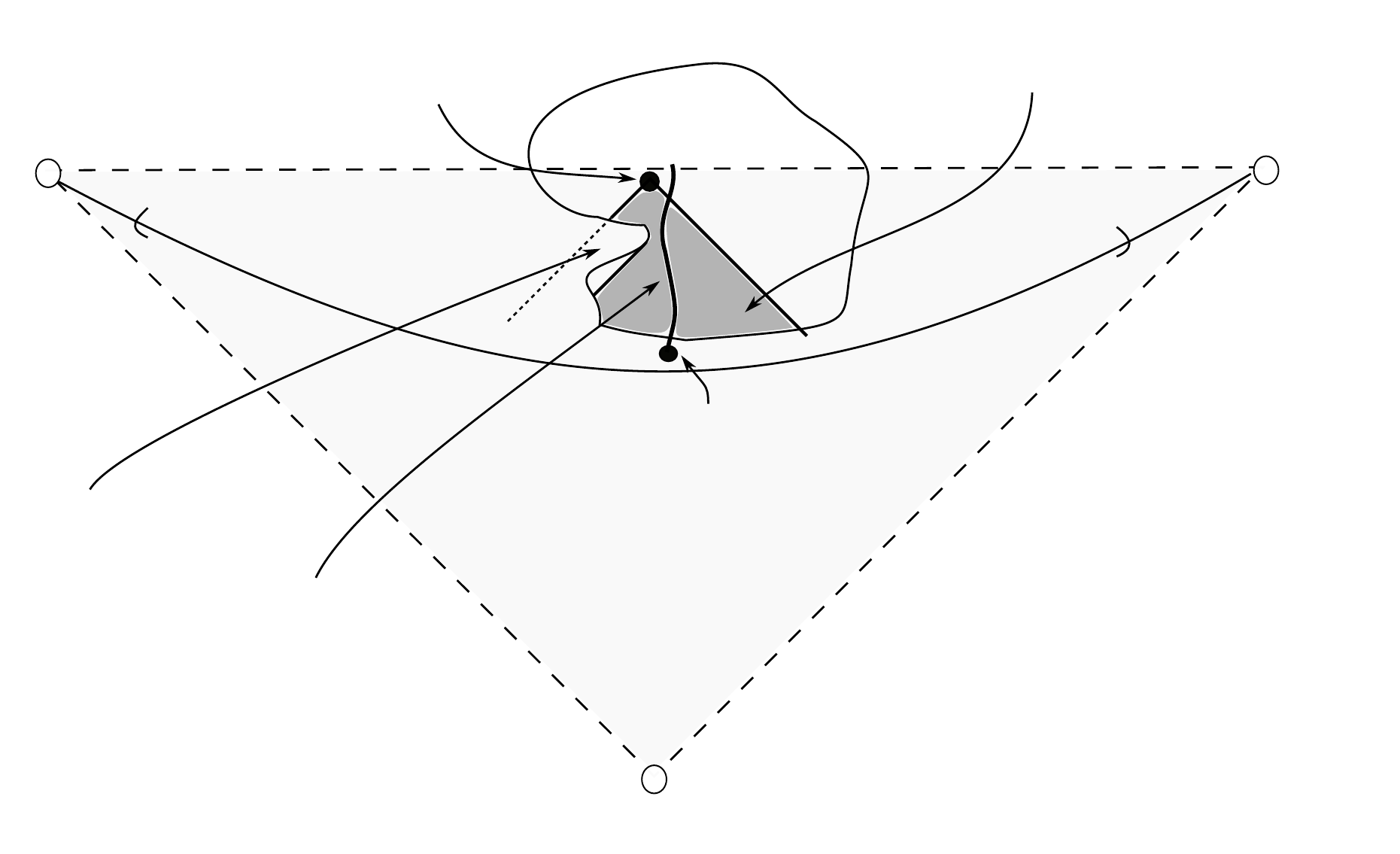
\end{center}
To prove \eqref{PastPreserving}, let $q \in I^+ \Big( I^-\big(\gamma((-1,0)), \Mint\big) \cap \{r < r_0\}, \Mint\Big)$ be such that $\iota(q) \in \tilde{O} \cap \iota(\Mint)$ is in the same connectedness component of $\tilde{O} \cap \iota(\Mint)$ as $\tilde{\gamma}\big((-\varepsilon, 0)\big)$, and let $\sigma : [0,1] \to \tilde{O}$ be a past directed timelike curve with $\sigma(0) = \iota(q)$. Initially, this curve is contained in $\iota(\Mint)$ and since $\iota(q) \in \tilde{O} \cap \iota(\Mint)$ is in the same connectedness component of $\tilde{O} \cap \iota(\Mint)$ as $\tilde{\gamma}\big((-\varepsilon, 0)\big)$, the curve is also past directed with respect to the time orientation of $\Mint$. Thus, it can only leave $\iota(\Mint)$ if its $r$ value tends to $r_+$. However, since this curve is contained in $\tilde{O}$, the $r$-value of the curve never exceeds $r_0 < r_+$. Thus, the curve must be contained completely in $\iota(\Mint)$.
\vspace*{2mm}

\textbf{Step 1.2}
Let us define $x^+ :=(\frac{3}{4} \varepsilon, 0, \ldots, 0)$ and $x^- :=(-\frac{3}{4} \varepsilon, 0, \ldots, 0)$. Note that the closure of $\big(x^+ + C^-_{\nicefrac{5}{6}}\big) \cap \big(x^- + C^+_{\nicefrac{5}{6}}\big)$ in $\ed$ is compact.
Now choose $y^- := (y^-_0, 0, \ldots, 0)$ with $ - \frac{1}{5}\varepsilon < y^-_0 < 0$ so that the closure of $C^-_{\nicefrac{5}{8}} \cap \big(y^- + C^+_{\nicefrac{5}{8}} \big)$ in $(-\varepsilon, \varepsilon)^{d+1}$ is contained in  $x^+ + C^-_{\nicefrac{6}{7}} \cap x^- + C^+_{\nicefrac{6}{7}}$.
\begin{center}
\def\svgwidth{8.5cm}
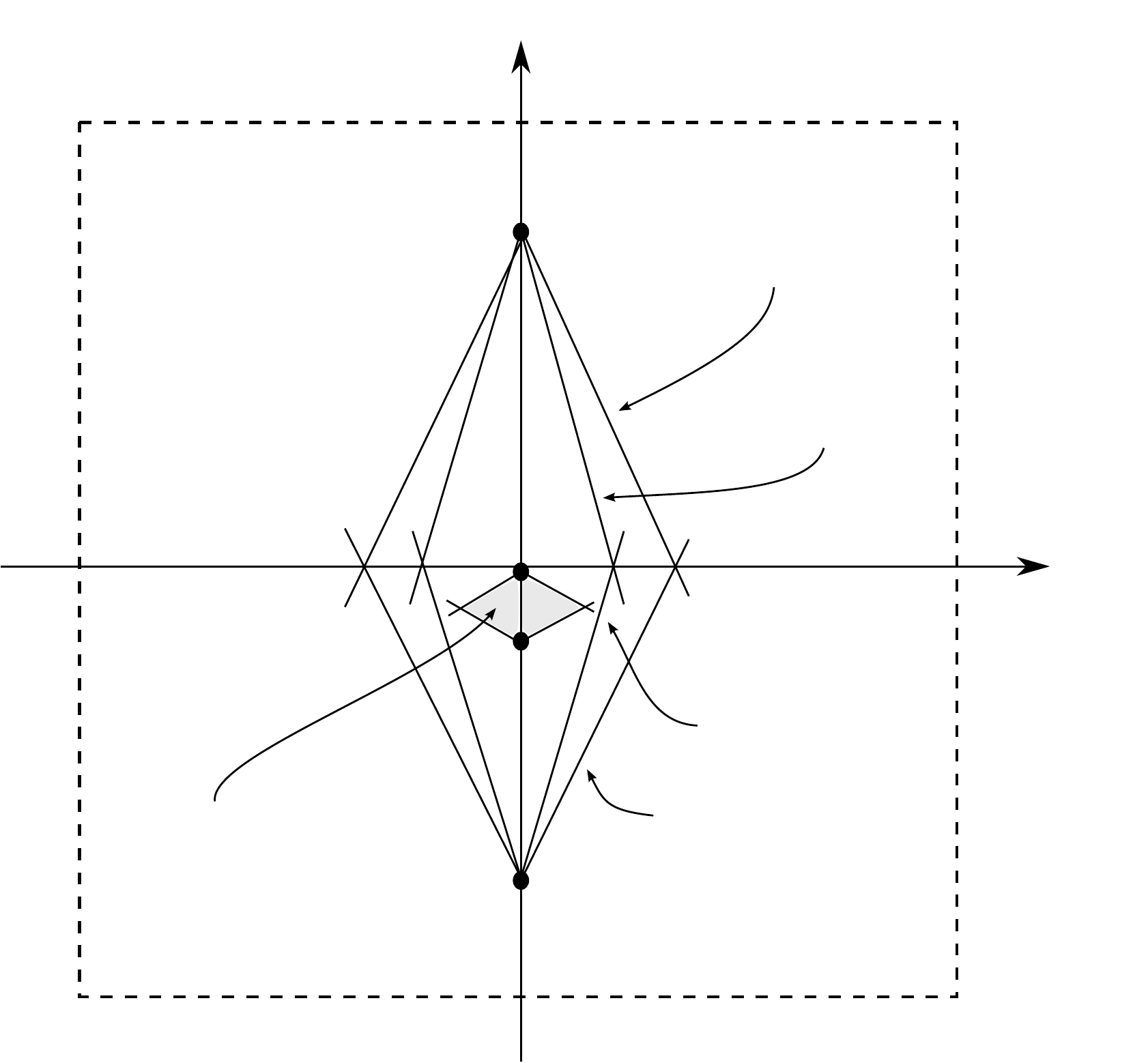
\end{center}
We claim that for all $y^-_0 < s < 0$ we have 
\begin{equation}
\label{FuturePastPreserving}
\begin{split}
I^-\big((s,0, \ldots, 0), \ed\big) &\cap I^+\big(y^-, \ed\big)  \\
&= (\tilde{\varphi} \circ \iota) \Big[I^-\big(\gamma(s), \Mint \big) \cap I^+\big( \gamma(y^-_0), \Mint\big) \Big] \;.
\end{split}
\end{equation}

The inclusion ``$\, \subseteq \,$'' follows from \eqref{PastPreserving}, since if $\sigma$ is a past directed timelike curve in $\ed$ from $(s,0, \ldots, 0)$ to $y^-$, then \eqref{PastPreserving} states that $\tilde{\varphi}^{-1} \circ \sigma$ is contained in $\iota(\Mint)$.

To prove ``$\, \supseteq \,$'', let $\sigma : [y_0^-, s] \to \Mint$ be a future directed timelike curve from $\gamma(x_0^-)$ to $\gamma(s)$. By Proposition \ref{IntFutCon}, there exists a timelike homotopy $\Gamma :  [0,1]\times [y^-_0,s] \to \Mint$ with fixed endpoints between $\gamma |_{[y^-_0,s]}$ and $\sigma$. It follows that $\iota \circ \Gamma: [0,1] \times [y^-_0,s]  \to \tilde{M}$ is a timelike homotopy with fixed endpoints in $\tilde{M}$. We need to show that $(\iota \circ \sigma) (\cdot)= (\iota \circ \Gamma)(1, \cdot)$ maps into $\tilde{O}$. We argue by continuity, i.e., we show that the interval $J := \{t \in [0,1] \,|\, \iota \circ \Gamma\big(t, [y^-_0,s]\big) \subseteq \tilde{O}\}$ is non-empty, open and closed in $[0,1]$.

Clearly, we have $0 \in J$, since $\iota \circ \Gamma(0, \cdot) = \tilde{\gamma}|_{[y^-_0,s]}$. The openness follows from the openness of $\tilde{O}$, and the closedness follows since $I^-\big(\tilde{\gamma}(s), \tilde{O}\big) \cap I^+ \big(\tilde{\gamma}(y^-_0),\tilde{O}\big)$ is precompact in $\tilde{O}$, i.e., in particular its closure in $\tilde{M}$ is contained in $\tilde{O}$.
This finishes the proof of \eqref{FuturePastPreserving}.

Together with Proposition \ref{UnionPasts}, we now deduce from \eqref{FuturePastPreserving} that
\begin{equation}
\label{FuturePastPreserving2}
\begin{split}
I^-\big(0, \ed\big) &\cap I^+\big(y^-, \ed\big)  \\
&= (\tilde{\varphi} \circ \iota) \Big[\Big(\bigcup_{-\varepsilon < s < 0} I^-\big(\gamma(s), \Mint \big)\Big) \cap I^+\big( \gamma(y^-_0), \Mint\big) \Big] \;.
\end{split}
\end{equation}
\vspace*{2mm}

\textbf{Step 1.3} In this step we switch back to the manifold $(\Mint,\gint)$. Choose a $y^+_0$ with $y^-_0 < y^+_0 <0$ and define 
\begin{equation}
\label{DefK}
K:=\Big[\Big(\bigcup_{-\varepsilon < s < 0} I^-\big(\gamma(s), \Mint \big)\Big) \cap I^+\big( \gamma(y^-_0), \Mint\big) \Big] \setminus I^+\big( \gamma(y^+_0), \Mint\big)  \;.
\end{equation}
\vspace*{2mm}

\textbf{Step 1.3.1}  We show that the set $K$ timelike separates the set $\gamma\big((y^+_0,0)\big)$ from $I^-\big(\gamma(y^-_0),\Mint\big)$.
\vspace*{2mm}

So let $\sigma : [0,1] \to \Mint$ be a past directed timelike curve with $\sigma (0) \in \gamma\big((y^+_0,0)\big)$ and $\sigma(1) \in I^-\big(\gamma(y^-_0),\Mint\big)$. 

\begin{center}
\def\svgwidth{4.8cm}
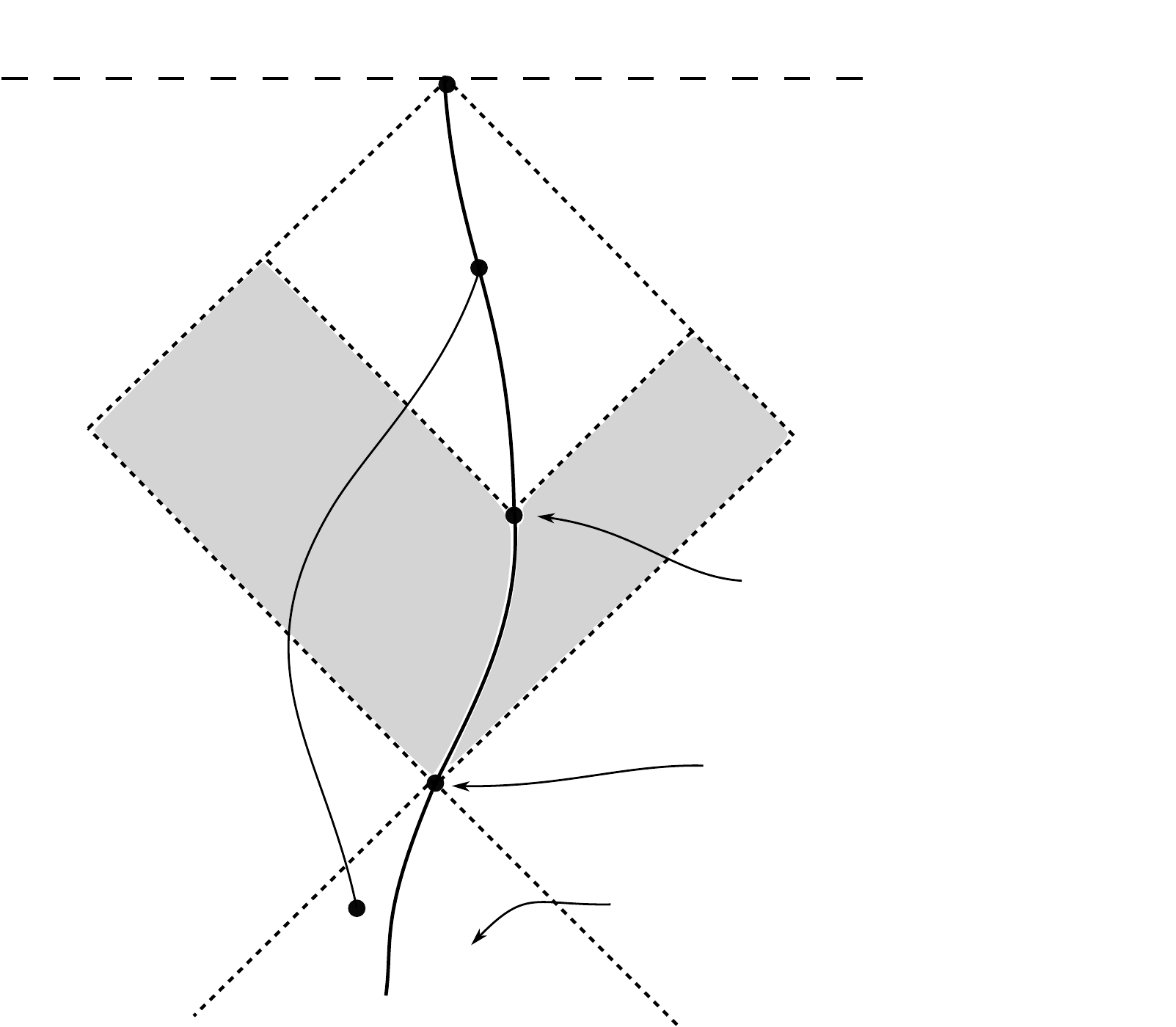
\end{center}

We claim that there exists a $\Delta_- \in (0,1)$ such that 
\begin{equation*}
\sigma^{-1}\big[I^+\big(\gamma(y^-_0),\Mint\big)\big] = [0,\Delta_-)\;.
\end{equation*} 
This is seen as follows: To begin with, it is clear that $0 \in \sigma^{-1}\big[I^+\big(\gamma(y^-_0),\Mint\big)\big]$. Moreover, by the continuity of $\sigma$ and the openness of $I^+\big(\gamma(y^-_0),\Mint\big)$, we know that $\sigma^{-1}\big[I^+\big(\gamma(y^-_0),\Mint\big)\big]$ is open in $[0,1]$. Moreover, since $\sigma$ is a past directed timelike curve, it follows that if $s_0 \in \sigma^{-1}\big[I^+\big(\gamma(y^-_0),\Mint\big)\big]$, then we also have $[0,s_0] \in \sigma^{-1}\big[I^+\big(\gamma(y^-_0),\Mint\big)\big]$. And finally, since $(\Mint,\gint)$ satisfies the chronology condition, there are no closed timelike curves in $\Mint$, and hence $I^-\big(\gamma(y^-_0),\Mint\big)$ is disjoint from $I^+\big(\gamma(y^-_0),\Mint\big)$. This implies that $\Delta_- <1$.

In the same way we deduce that there exists a $\Delta_+ \in (0,1)$ such that
\begin{equation*}
\sigma^{-1}\big[I^+\big(\gamma(y^+_0),\Mint\big)\big] = [0,\Delta_+)\;.
\end{equation*} 
In the following we show that $\Delta_+ < \Delta_-$.

Since $(\Mint,\gint)$ is globally hyperbolic and $\gint$ is smooth (!), we have $\overline{I^+\big(\gamma(y^+_0),\Mint\big)} = J^+\big(\gamma(y^+_0),\Mint\big)$.\footnote{This follows from 6.\ Lemma and 22.\ Lemma of Chapter 14 of \cite{ONeill}. Also note that $J^+\big(\gamma(y^+_0),\Mint\big)$ denotes the causal future of $\gamma(y^+_0)$ in $\Mint$, the definition of which, together with that of a null and causal curve, can be also found in Chapter 14 of \cite{ONeill}. Let us moreover emphasise here that for the purposes of this paper there is no need to define the causal future for Lorentzian manifolds with merely continuous metric.} Together with $\gamma(y^+_0) \in I^+\big(\gamma(y^-_0),\Mint \big)$, we now obtain\footnote{See 1.\ Corollary in Chapter 14 of \cite{ONeill}.}
\begin{equation*}
\overline{I^+\big(\gamma(y^+_0),\Mint\big)} = J^+\big(\gamma(y^+_0),\Mint\big) \subseteq J^+\Big(I^+\big(\gamma(y^-_0),\Mint\big),\Mint\Big) = I^+\big(\gamma(y^-_0),\Mint\big) \;.
\end{equation*}
Hence, we have $\sigma(\Delta_+) \in I^+\big(\gamma(y^-_0),\Mint\big)$, from which it follows that $\Delta_+ < \Delta_-$.

Choosing $s_0 \in (\Delta_+,\Delta_-)$, it follows that $\sigma(s_0) \in   I^+\big(\gamma(y^-_0),\Mint\big)  \setminus I^+\big(\gamma(y^+_0),\Mint\big)$. Moreover, it is clear that $\sigma(s_0) \in \bigcup_{-\varepsilon < s < 0} I^-\big(\gamma(s), \Mint \big)$, which concludes Step 1.4.1.
\vspace*{2mm}

\textbf{Step 1.3.2} We show that $\overline{K}$ is compact.
\vspace*{2mm}

We claim that for any $s_0 \in (y^+_0,0)$ there exists a $\delta >0$ and a neighbourhood $V$ of $\mathrm{id} \in SO(d)$ such that
\begin{equation*}
\big(\gamma_t(s) - \delta, \gamma_t(s) + \delta\big) \times \big{\{}\gamma_r(s)\big{\}} \times \big{\{}f \cdot \gamma_\omega(s) \, | \, f \in V\big{\}} \subseteq I^+\big(\gamma(y^+_0),\Mint\big)
\end{equation*}
holds for all $s \in (s_0,0)$.

In order to prove this claim, we first note that since $\gamma$ is timelike and $\gint$ is continuous, there exists a $\mu >0$ such that 
\begin{equation}
\label{UniformTimelike}
\gint\big(\dot{\gamma}(s), \dot{\gamma}(s)\big) < - \mu
\end{equation}
holds for all $s \in [y^+_0,s_0]$.

Let now $\lambda \in C^\infty\big( [y^+_0,s_0], \R\big)$ and $ h \in C^\infty\big([y^+_0,s_0], SO(d) \subseteq \mathrm{Mat}(d \times d, \R)\big)$, and define $\sigma : [y^+_0,s_0] \to \Mint$ by
\begin{equation*}
\sigma(s) := \Big(\gamma_t(s) + \lambda(s), \gamma_r(s), h(s)\big(\gamma_\omega(s)\big)\Big) \;.
\end{equation*}
We compute
\begin{equation}
\label{NormSigma}
\begin{split}
\gint\big(\dot{\sigma},\dot{\sigma}\big) = - \Big(1 - \frac{2m}{\big(\gamma_r(s)\big)^{d-2}}\Big)\, \big(\dot{\gamma}_t(s) + \dot{\lambda}(s)\big)^2 &+ \Big(1 - \frac{2m}{\big(\gamma_r(s)\big)^{d-2}}\Big)^{-1} \, \big(\dot{\gamma}_r(s)\big)^2 \\
&+ \big(\gamma_r(s)\big)^2 \, \big|\big|\dot{h}(s) \gamma_\omega(s) + h(s) \dot{\gamma}_\omega(s)\big|\big|^2_{\R^d} \;,
\end{split}
\end{equation}
where $|| \cdot ||_{\R^d}$ denotes the Euclidean norm on $\R^d$, and we think of $\gamma_\omega$ as mapping into $\Sd \subseteq \R^d$.
Since $\gamma_r(s)$ is bounded away from $0$ for $s \in [y^+_0, s_0]$, we can infer from \eqref{UniformTimelike} and \eqref{NormSigma} that there exists an $\eta >0$ such that whenever 
\begin{equation}
\label{BoundsOnSlopes}
||\dot{\lambda}||_{L^\infty\big([y^+_0,s_0]\big)} + ||\dot{h}||_{L^\infty\big([y^+_0,s_0]\big)} < \eta
\end{equation}
holds, the curve $\sigma : [y^+_0,s_0] \to \Mint$ is timelike\footnote{One can for example define $||\dot{h}||_{L^\infty\big([y^+_0,s_0]\big)} := \sup\limits_{s \in [y^+_0,s_0]} ||\dot{h}(s)||_{\R^{d\times d}}$.}. This in turn implies the existence of a $\delta >0$ and a neighbourhood $V$ of $\mathrm{id} \in SO(d)$ such that for every $\lambda_{s_0} \in (-\delta, \delta)$ and for every $h_{s_0} \in V$ there are smooth functions 
\begin{equation*}
\lambda \in C^\infty\big( [y^+_0,s_0], \R\big) \quad \textnormal{ with } \; \lambda(y^+_0) = 0 \; \textnormal{ and } \; \lambda(s_0) = \lambda_{s_0}
\end{equation*}
and
\begin{equation*}
h \in C^\infty\big([y^+_0,s_0], SO(d) \subseteq \mathrm{Mat}(d \times d, \R)\big) \quad \textnormal{ with }\; h(y^+_0) = \mathrm{id} \; \textnormal{ and } \; h(s_0) = h_{s_0}
\end{equation*}
such that moreover \eqref{BoundsOnSlopes} is satisfied. The claim now follows from concatenating $\sigma$ with the timelike\footnote{Recall that the Schwarzschild metric \eqref{SchwarzschildMetricInt} is spherically symmetric and invariant under translations in $t$.} curve  $\tau : [s_0, 0) \to \Mint$ given by
\begin{equation*}
\tau(s) = \Big(\gamma_t(s) + \lambda_{s_0}, \gamma_r(s), h_{s_0} \big(\gamma_\omega(s)\big)\Big) \;.
\end{equation*}

We now fix $s_0 \in (y^+_0, 0)$ and obtain $\delta >0$ and a neighbourhood $V \subseteq SO(d)$ of $\mathrm{id} \in SO(d)$ as in the claim. It follows from Lemma \ref{BoundsOnReach} that there exists $s_1 \in (s_0 ,0)$ (close to $0$) such that
\begin{equation*}
\begin{split}
I^-\big(\gamma(s),\Mint\big) \cap \Big{\{} r \leq \gamma_r(s_1)\Big{\}} &\subseteq \bigcup_{s_0  \leq s' <0} \Big[ \big(\gamma_t(s') - \delta, \gamma_t(s') + \delta\big) \times \big{\{}\gamma_r(s')\big{\}} \times \big{\{}f \cdot \gamma_\omega(s') \, | \, f \in V\big{\}} \Big] \\
&\subseteq I^+\big(\gamma(y^+_0),\Mint\big)
\end{split}
\end{equation*}
holds for all $s \in (-\varepsilon, 0)$. This implies 
\begin{equation*}
K \subseteq \R \times \big(\gamma_r(s_1), \gamma_r(y^-_0)\big) \times \Sd \;.
\end{equation*}
Moreover, the bound \eqref{BoundTVel} implies that there are $t_0, t_1 \in \R$ such that 
\begin{equation*}
I^+\big(\gamma(y^-_0),\Mint\big) \subseteq (t_0,t_1) \times \big(0, \gamma_r(y^-_0)\big) \times \Sd \;.
\end{equation*}
It follows that 
\begin{equation*}
K \subseteq (t_0,t_1) \times \big(\gamma_r(s_1), \gamma_r(y^-_0)\big) \times \Sd\;,
\end{equation*}
which implies that $\overline{K}$ is compact.
\vspace*{2mm}

\textbf{Step 1.4} First note that by the continuity of $\tilde{\varphi} \circ \iota$ we have\footnote{Indeed, since $\overline{K} \subseteq \Mint$ is compact, we actually have equality.}
\begin{equation}
\label{ClosureRel}
(\tilde{\varphi}\circ \iota)(\overline{K})  \subseteq \overline{(\tilde{\varphi} \circ \iota)(K)} \;.
\end{equation}
Moreover, it follows from the definition of $K$, \eqref{DefK}, that $K \subseteq \Big(\bigcup_{-\varepsilon < s < 0} I^-\big(\gamma(s), \Mint \big)\Big) \cap I^+\big( \gamma(y^-_0), \Mint\big)$, and thus, together with \eqref{FuturePastPreserving2}, we obtain
\begin{equation*}
\overline{(\tilde{\varphi} \circ \iota)(K)}  \subseteq \overline{I^-\big(0, \ed\big) \cap I^+\big(y^-, \ed\big)} \;.
\end{equation*}
By the choice of $y^- \in \ed$ in Step 1.3, together with \eqref{EstimatesOnPastAndFuture}, it follows that
\begin{equation}
\label{Neighbourhood}
\overline{(\tilde{\varphi} \circ \iota)(K)}  \subseteq  \Big(x^+ + C^-_{\nicefrac{6}{7}}\Big) \cap \Big(x^- + C^+_{\nicefrac{6}{7}}\Big)\;.
\end{equation}
Hence, from \eqref{ClosureRel} and \eqref{Neighbourhood}, it follows that 
\begin{equation*}
U:= (\tilde{\varphi} \circ \iota)^{-1}\Big(\big[x^+ + C^-_{\nicefrac{6}{7}}\big] \cap \big[x^- + C^+_{\nicefrac{6}{7}}\big]\Big) \subseteq \Mint
\end{equation*}
is an open neighbourhood of $\overline{K} \subseteq \Mint$.
\vspace*{2mm}

\textbf{Step 1.5} We show that there exists a $\mu >0$ such that $I^+\big(\gamma(-\mu), \Mint\big)$ is timelike separated from $\gamma(x^-_0)$ by $U$.
\vspace*{2mm}

We consider $\Mint = \R \times \big(0,(2m)^{\frac{1}{d-2}}\big) \times \Sd$ with the metric $d_{\Mint} : \Mint \times \Mint \to [0,\infty)$ given by
\begin{equation*}
d_{\Mint}\Big((t_1, r_1, \omega_1), (t_2, r_2, \omega_2)\Big) := |t_1 - t_2| + |r_1 - r_2| + d_{\Sd}\big(\omega_1, \omega_2\big) \;,
\end{equation*}
where $(t_i, r_i, \omega_i) \in \Mint$ for $i=1,2$. Since  
\begin{equation*}
\Mint \ni (t,r,\omega) \mapsto d_{\Mint}\big((t,r,\omega),\Mint \setminus U\big) = \inf\limits_{(t',r',\omega') \in \Mint \setminus U} d_{\Mint}\big((t,r,\omega),(t',r',\omega')\big)
\end{equation*}
is continuous, and $\overline{K}$ is compact and disjoint from the closed set $\Mint \setminus U$, we infer that $d_{\Mint}(\cdot, \Mint \setminus U)$ must attain its minimum on $\overline{K}$, which is moreover strictly positive. It follows that there exists a $\delta >0$ such that 
\begin{equation*}
\overline{K}_\delta := \Big{\{} (t,r,\omega) \in \Mint \, \big| \, d_{\Mint}\big((t,r,\omega), \overline{K}\big) < \delta \Big{\}} \subseteq U \;.
\end{equation*}
Moreover, by choosing $\delta $ slightly smaller if necessary, we can also assume that 
\begin{equation}
\label{RestrictionOnDelta}
B_\delta\big(\gamma(x^-_0)\big) \subseteq I^-\big(\gamma(y^-_0),\Mint\big)\;.
\end{equation}
\vspace*{2mm}

\textbf{Step 1.5.1}
We define a metric $d_{SO(d)} : SO(d) \times SO(d) \to [0,\infty)$ on $SO(d)$ by
\begin{equation*}
d_{SO(d)}(f,h) := \sup\limits_{\omega \in \Sd} d_{\Sd}\big(f (\omega), h(\omega)\big)\;,
\end{equation*}
where $f, h  \in SO(d)$, and denote with $B_{\eta}(\mathrm{id}) \subseteq SO(d)$ the ball of radius $\eta >0 $, centred at $\mathrm{id}$, with respect to this metric. Moreover, it is easy to see that 
\begin{equation}
\label{ChooseRotation}
\textnormal{for } \omega_0, \omega_1 \in \Sd \textnormal{ with } d_{\Sd}(\omega_0, \omega_1) < \eta \textnormal{,  there exists an } h \in B_\eta(\mathrm{id})  \textnormal{ with } h(\omega_0) = \omega_1\;. 
\end{equation}
In particular, $h$ can be defined as a rotation purely in the plane $\mathrm{span}\{\omega_0, \omega_1\} \subseteq \R^d$.
\vspace*{2mm}

Continuing the proof of Step 1.5, Lemma \ref{BoundsOnReach} implies that there exists a $\mu \in (0, -y^+_0)$ such that for all $(t_0,r_0, \omega_0) \in I^+\big(\gamma(-\mu),\Mint\big)$, we have
\begin{equation}
\label{ChoiceMu}
|t_0 - \gamma_t(-\mu)| < \frac{\delta}{2} \qquad \textnormal{ and } \qquad d_{\Sd}\big(\omega_0, \gamma_\omega(-\mu)\big) < \frac{\delta}{2} \;.
\end{equation}
In the following we will show that $I^+\big(\gamma(-\mu),\Mint\big)$ is timelike separated from $\gamma(x^-_0)$ even by $\overline{K}_\delta$ (which, of course, implies Step 1.5).

So let $\sigma : [0,1] \to \Mint$ be a past directed timelike curve with $\sigma(0) \in I^+\big(\gamma(-\mu),\Mint\big)$ and $\sigma (1) = \gamma(x^-_0)$. Also let $s_0 \in (-\mu ,0)$ be such that $\gamma_r(s_0) = \sigma_r (0)$.
 By \eqref{ChoiceMu} we have 
 \begin{equation*}
 |\sigma_t(0) - \gamma_t(s_0)| < \delta \qquad \textnormal{ and } \qquad d_{\Sd}\big(\sigma_\omega(0), \gamma_\omega(s_0)\big) < \delta \;.
 \end{equation*}
Thus, by \eqref{ChooseRotation} there exists an $h \in B_\delta(\mathrm{id}) \subseteq SO(d)$ such that $h\big(\sigma_\omega(0)\big) = \gamma_\omega(s_0)$. It now follows that the curve $\hat{\sigma} : [0,1] \to \Mint$, given by 
\begin{equation*}
\hat{\sigma}(s) = \Big(\sigma_t(s) + [\gamma_t(s_0) - \sigma_t(0)], \sigma_r(s), h\big(\sigma_\omega(s)\big)\Big)\;,
\end{equation*}
is past directed timelike with $\hat{\sigma}(0) = \gamma(s_0)$ and, using the fact that $h \in B_\delta(\mathrm{id})$ together with \eqref{RestrictionOnDelta},  $\hat{\sigma}(1) \in I^-\big(\gamma(y^-_0),\Mint\big)$. By Step 1.4.1, there exists an $\hat{s} \in [0,1]$ with $\hat{\sigma}(\hat{s}) \in K$. It now follows that $\sigma(\hat{s}) \in \overline{K}_\delta$, which concludes Step 1.5.
\vspace*{2mm}

We now finish the proof of Step 1. We set $w := \gamma(-\mu)$. We will first show that $(\tilde{\varphi} \circ \iota) \big(I^+(w, \Mint)\big) \subseteq  \big[x^+ + C^-_{\nicefrac{6}{7}}\big] \cap \big[x^- + C^+_{\nicefrac{6}{7}}\big]$.  

The proof is by contradiction. So let $\sigma : [0,1] \to \Mint$ be a future directed timelike curve with $\sigma(0) = \gamma(-\mu)$ and assume that there exists $\tilde{s} \in [0,1]$ such that $(\tilde{\varphi} \circ \iota \circ \sigma)(\tilde{s}) \notin \big[x^+ + C^-_{\nicefrac{6}{7}}\big] \cap \big[x^- + C^+_{\nicefrac{6}{7}}\big]$. Let
\begin{equation*}
s_0 := \sup \big{\{}s' \in [0,1] \, | \, (\tilde{\varphi} \circ \iota \circ \sigma)(s) \in \big[x^+ + C^-_{\nicefrac{6}{7}}\big] \cap \big[x^- + C^+_{\nicefrac{6}{7}}\big] \textnormal{ for all } s \in [0,s')\big{\}} \;.
\end{equation*}
Clearly, we have $0 < s_0 \leq 1$ and from our assumption it follows that $(\tilde{\varphi} \circ \iota \circ \sigma)(s_0)  \in \partial \big(\big[x^+ + C^-_{\nicefrac{6}{7}}\big] \cap \big[x^- + C^+_{\nicefrac{6}{7}}\big]\big)$. 
Since all vectors in $C^-_{\nicefrac{5}{6}}$ are past directed timelike, we can find a past directed timelike curve $\tau : [0,1] \to \ed$ with $\tau(0) = (\tilde{\varphi} \circ \iota \circ \sigma) (s_0)$ and $\tau(1) = x^-$, which does not intersect $\big[x^+ + C^-_{\nicefrac{6}{7}}\big] \cap \big[x^- + C^+_{\nicefrac{6}{7}}\big]$. For example, this curve can be chosen to lie in $\partial \Big(\big[x^+ + C^-_{\nicefrac{6}{7}}\big] \cap \big[x^- + C^+_{\nicefrac{6}{7}}\big]\Big)$.
\begin{center}
\def\svgwidth{12cm}
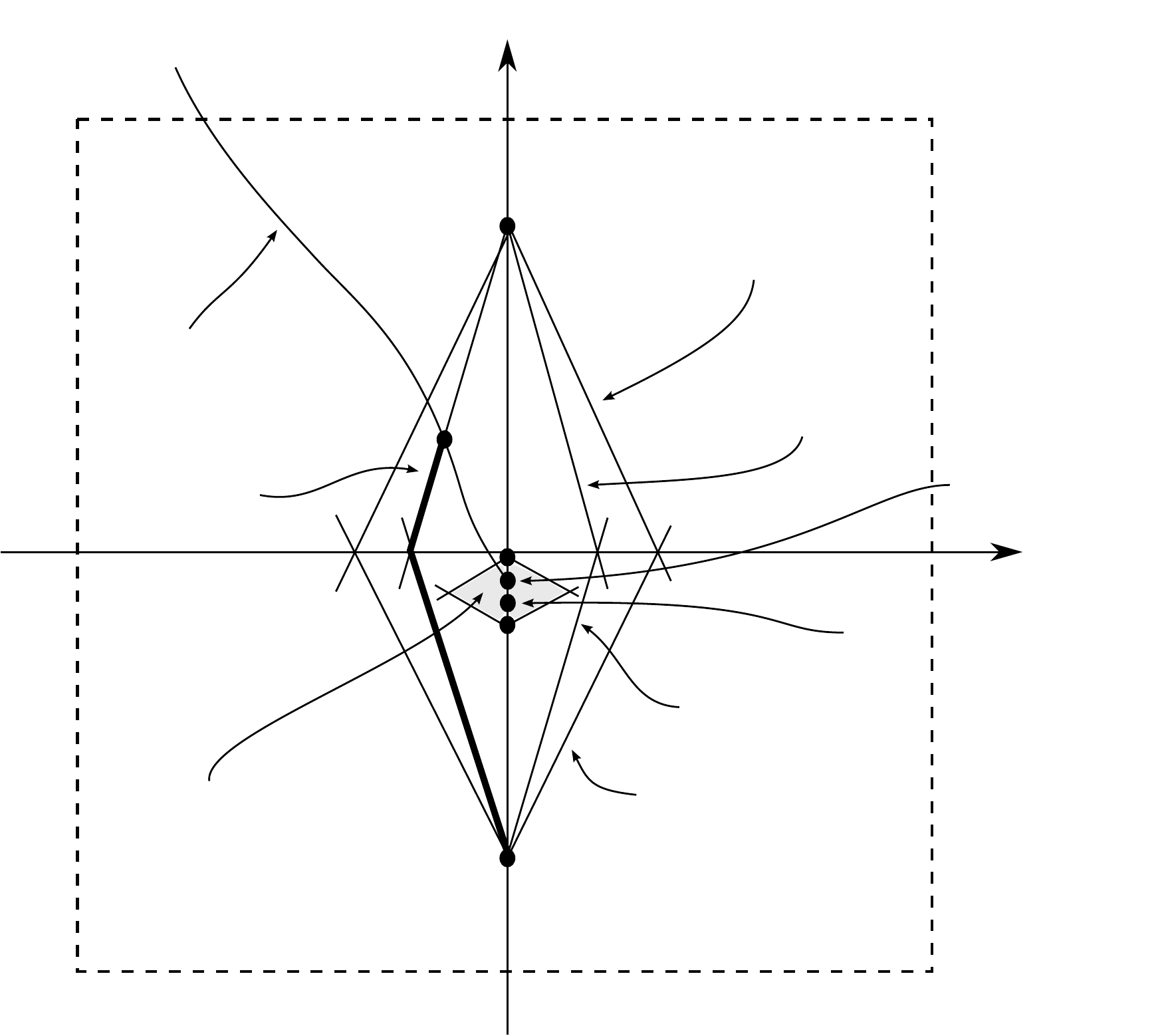
\end{center}
It now follows from $\eqref{PastPreserving}$ that $\tau$ maps in $\tilde{\varphi}\big(\iota(\Mint)\big)$ (since $\sigma(s_0) \in \Mint$). Hence, $(\tilde{\varphi} \circ \iota)^{-1} \circ \tau$ is a past directed timelike curve in $\Mint$ with $\big((\tilde{\varphi} \circ \iota)^{-1} \circ \tau\big)(0) = \sigma(s_0) \in I^+\big(\gamma(-\mu),\Mint\big)$ and $\big((\tilde{\varphi} \circ \iota)^{-1} \circ \tau\big)(1) =\gamma(x^-_0)$, which does not intersect $U= (\tilde{\varphi} \circ \iota)^{-1}\Big(\big[x^+ + C^-_{\nicefrac{6}{7}}\big] \cap \big[x^- + C^+_{\nicefrac{6}{7}}\big]\Big)$. This, however, is a contradiction to Step 1.5. Hence, we have shown that $(\tilde{\varphi} \circ \iota) \big(I^+(w, \Mint)\big) \subseteq  \big[x^+ + C^-_{\nicefrac{6}{7}}\big] \cap \big[x^- + C^+_{\nicefrac{6}{7}}\big]$. 

We now define $\rho >0$ by $\big(x^+ + C^-_{\nicefrac{6}{7}}\big) \cap \{x_0 = 0\} = \{0\} \times B^d_\rho(0)$.
Recalling that $-\frac{1}{5} \varepsilon < y_0^- < -\mu < 0$, elementary geometry shows that there exists a $\delta >0$ such that
\begin{equation*}
\{-\frac{19}{20} \varepsilon\} \times B^d_{\rho + \delta}(0) \subseteq \big((-\frac{1}{5}\varepsilon, 0, \ldots, 0) + C^-_{\nicefrac{5}{6}}\big) \cap \{x_0 = -\frac{19}{20} \varepsilon\} \;,
\end{equation*}
and thus by  \eqref{EstimatesOnPastAndFuture} and \eqref{PastPreserving} we obtain $\tilde{\psi}^{-1} \Big((-\varepsilon, - \frac{19}{20} \varepsilon] \times B^d_{\rho + \delta}(0) \Big) \subseteq I^-\big(\gamma(-\mu), \Mint\big)$.

We now set $\tilde{U} := \tilde{\varphi}^{-1} \big( (-\varepsilon, \varepsilon) \times B_{\rho + \delta}^d(0)\big)$ and $\tilde{\psi} := \tilde{\varphi}\big|_{\tilde{U}}$. Clearly, the first three properties of Step 1 are satisfied.
It remains to prove the fourth property. In fact, for $\underline{x} \in B^d_{\rho + \delta}(0)$ we cannot have $\tilde{\varphi}^{-1}(s,\underline{x}) \in \iota(\Mint)$ for all $s \in (-\varepsilon, \frac{3}{4} \varepsilon]$. Since otherwise we can connect $(\frac{3}{4} \varepsilon,  \underline{x})$ to $0$ via a past directed timelike curve, which, by \eqref{PastPreserving}, would imply the contradiction $\tilde{\varphi}^{-1}(0) \in \iota(\Mint)$.
This concludes the proof of Step 1.
\vspace*{2mm}

\underline{\textbf{Step 2:}} We construct a connected and globally hyperbolic subset $N \subseteq \Mint$ with $\diam_s(N) = \infty$.
\vspace*{5mm}

Let $\omega_0$ be the projection of $\gamma(-\frac{\mu}{2})$ to $\Sd$, $t_0 := t\big(\gamma(-\frac{\mu}{2})\big)$, and $r_0 := r\big(\gamma(-\frac{\mu}{2})\big)$. By the openness of $I^+\big(\gamma(-\mu), \Mint\big)$ there exist $\lambda, \kappa >0$ such that $[t_0 - \lambda, t_0 + \lambda] \times \{r_0\} \times B^{\Sd}_\kappa(\omega_0) \subseteq I^+\big(\gamma(-\mu), \Mint\big)$, where $B^{\Sd}_\kappa(\omega_0)$ denotes the ball of radius $\kappa$ around $\omega_0$ in $\Sd$.  Since $-\partial_r$ is future directed timelike, we also have
\begin{equation}
\label{BlockContained}
[t_0 - \lambda, t_0 + \lambda] \times (0,r_0] \times B^{\Sd}_\kappa(\omega_0)  \subseteq I^+\big(\gamma(-\mu), \Mint\big) \;.
\end{equation}
We define $r^*_\mathrm{int}(r) := \int_0 ^r \frac{1}{D(r')} \, dr'$ and set $v^*_\mathrm{int} := r^*_\mathrm{int} + t$ and $u^*_\mathrm{int} := r^*_\mathrm{int} - t$. It is easy to check that $v^*_\mathrm{int}$ and $u^*_\mathrm{int}$ are null coordinates - indeed, one could add a suitable constant $b \in \R$ to the definition of $r^*_\mathrm{int}(r)$ to obtain the relations $v^*_\mathrm{int} = \frac{2r_+}{d-2} \log (v)$ and $u^*_\mathrm{int} = \frac{2r_+}{d-2} \log(u)$ with the null coordinates $u,v$ defined in Section \ref{MaxAna}. We choose $r_1 \in (0,r_0)$ such that $0>r^*_\mathrm{int}(r_1) > -\frac{\lambda}{2}$ and set
\begin{equation*}
\begin{split}
N:=\Big[\{r \leq r_1\} &\cap  \{u^*_\mathrm{int} < r^*_\mathrm{int}(r_1) - (t_0 - \lambda)\} \cap \{v^*_\mathrm{int} < r^*_\mathrm{int}(r_1) +(t_0 + \lambda)\}\Big] \\
&\cup \Big[ \{r_1 < r < r_0\} \cap \{u^*_\mathrm{int} > r^*_\mathrm{int}(r_1) - (t_0 + \lambda)\} \cap \{v^*_\mathrm{int} > r^*_\mathrm{int}(r_1) +(t_0 - \lambda)\}\Big] \;.
\end{split}
\end{equation*}
Clearly $N$ is connected and it is also not difficult to verify that $N$ is globally hyperbolic with Cauchy hypersurface $\Sigma = (t_0 - \lambda, t_0 + \lambda) \times \{r_1\} \times \Sd$. It remains to show that $\diam_s(N) = \infty$.

For $n > \frac{1}{r_1}$ let $f_n : \R \to (0, r_1]$ be a smooth function with
\begin{equation*}
f_n(t) = \begin{cases} r_1 \textnormal{ for } t < t_0 - \frac{7}{8}\lambda \\
\frac{1}{n} \textnormal { for } t_0 - \frac{1}{4}\lambda \leq t \leq t_0 + \frac{1}{4}\lambda \\
r_1 \textnormal{ for } t_0 + \frac{7}{8}\lambda < t \end{cases}
\end{equation*}
and such that $\{r = f_n(t)\}$ is spacelike. Again, it is straightforward to verify that $\Sigma_n := \{r = f_n(t)\} \cap \{t_0 - \lambda < t < t_0 + \lambda\}$ is a Cauchy hypersurface for $N$.
\begin{center}
\def\svgwidth{12cm}
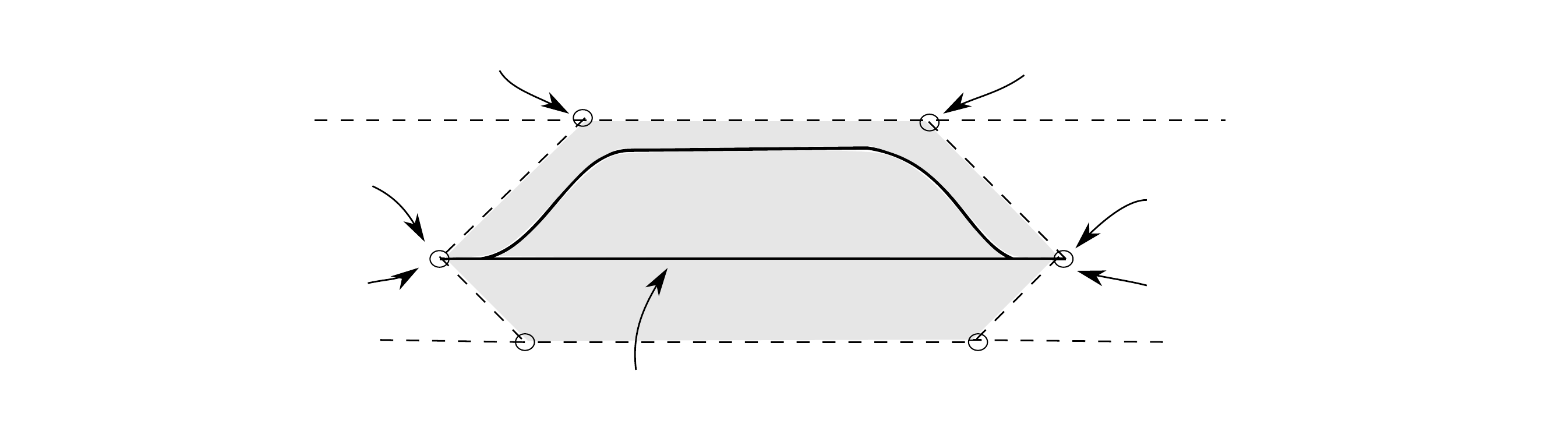
\end{center}
Let $\overline{g}_n$ denote the induced metric on $\Sigma_n$ and, using \eqref{SchwarzschildMetricInt}, we see that in the region $\Sigma_n \cap \{t_0 - \frac{1}{4}\lambda < t < t_0 + \frac{1}{4} \lambda\}$ the metric $\overline{g}_n$ is given by
\begin{equation*}
\overline{g}_n = -\big(1 - 2m\cdot n^{(d-2)}\big) \,dt^2 + n^{-2}\, \mathring{\gamma}_{d-1} \;.
\end{equation*}
We now consider the pairs of points $p_n = (t_0 - \frac{1}{4}\lambda, \frac{1}{n}, \omega_0) \in \Sigma_n$ and $q_n = (t_0 + \frac{1}{4}\lambda, \frac{1}{n}, \omega_0) \in \Sigma_n$.
The shortest curve connecting $p_n$ with $q_n$ in $\Sigma_n$ is given by
 $\gamma_n : [-\frac{1}{4}\lambda, \frac{1}{4}\lambda] \to \Sigma_n$, 
\begin{equation*}
\gamma_n (s) = ( t_0 + s, \frac{1}{n}, \omega_0) \;.
\end{equation*}
The length $L(\gamma_n)$ of $\gamma_n$ is given by
\begin{equation*}
\begin{split}
L(\gamma_n) &= \int_{-\frac{1}{4}\lambda}^{\frac{1}{4}\lambda} \sqrt{\overline{g}_n \big(\dot{\gamma}_n(s), \dot{\gamma}_n(s)\big)} \, ds \\
&= \int_{-\frac{1}{4}\lambda}^{\frac{1}{4}\lambda} \sqrt{2m \cdot n^{(d-2)} -1} \, ds \\ 
&= \frac{1}{2} \lambda \sqrt{2m \cdot n^{(d-2)} -1} \;.
\end{split}
\end{equation*}
Since $d \geq 3$, it thus follows that
\begin{equation*}
\diam(\Sigma_n) \geq d_{\Sigma_n}(p_n, q_n) = L(\gamma_n) = \frac{1}{2} \lambda \sqrt{2m \cdot n^{(d-2)} -1} \to \infty \quad \textnormal{ for } n \to \infty \;,
\end{equation*}
and hence
\begin{equation*}
\diam_s(N) = \sup_{\substack{\Sigma \textnormal{ Cauchy} \\ \textnormal{hypersurface of } N}} \diam(\Sigma) \; \geq \; \sup_{n \in \N_{> \frac{1}{r_1}}} \,\diam (\Sigma_n) = \infty \;.
\end{equation*}
\vspace*{2mm}

\underline{\textbf{Step 3:}} We show that Step 1 implies that $N$ can be covered by finitely many regular flow charts for $N$. A contradiction then follows from Theorem \ref{DiamFinite}.
\vspace*{5mm}

Consider the chart $\tilde{\psi} : \tilde{U} \to (-\varepsilon, \varepsilon) \times B^d_{\rho + \delta} (0)$ constructed in Step 1. For $\underline{x} \in B^d_{\rho + \delta}(0)$ let \begin{equation*}
I_{\underline{x}} := \big(- \varepsilon , \sup \big{\{} s_0 \in (-\varepsilon, \varepsilon) \; \big| \; \tilde{\psi}^{-1} (s,\underline{x}) \in \iota(\Mint) \;\; \forall  s \in (-\varepsilon, s_0)\big{\}}\big)
\end{equation*}
and set 
\begin{equation*}
D_{\rho + \delta} := \bigcup_{\underline{x} \in B^d_{\rho + \delta}(0)} I_{\ux} \times \{ \ux\} \subseteq (-\varepsilon, \varepsilon) \times B^d_{\rho + \delta}(0)\;.
\end{equation*} 
Clearly, $D_{\rho + \delta}$ is a connected component of $\tilde{\psi}\big(\iota(\Mint) \cap \tilde{U}\big)$ and hence open. Similarly set $D_\rho := \bigcup_{\underline{x} \in B^d_{\rho}(0)} I_{\ux} \times \{ \ux\}$. The corresponding regions in $\Mint$ we denote with $V_{\rho + \delta} := (\iota^{-1} \circ \tilde{\psi}^{-1})(D_{\rho + \delta})$ and $V_\rho := (\iota^{-1} \circ \tilde{\psi}^{-1})(D_{\rho})$. Setting $\psi_{\rho + \delta} := (\tilde{\psi} \circ \iota)|_{V_{\rho + \delta}}$ and $\psi_\rho := (\tilde{\psi} \circ \iota)|_{V_\rho}$, this yields the charts
\begin{equation*}
\psi_{\rho + \delta} : V_{\rho + \delta} \to D_{\rho + \delta} \quad \textnormal{ and } \quad \psi_{\rho} : V_{\rho} \to D_{\rho } 
\end{equation*}
of $\Mint$.  For $\ux \in B^d_{\rho + \delta}(0)$ the future directed timelike curves $\sigma_{\ux} : I_{\ux} \to \Mint$, $\sigma_{\ux}(s) = \psi^{-1}_{\rho + \delta}(s, \ux)$ are initially in $I^-\big(\{r = r_0\}, \Mint\big)$ (by property 3 of Step 1) and are future inextendible in $\Mint$ (by property 4 of Step 1). Hence, each of these curves intersects $\{r = r_1\}$ exactly once. As in Step 2 of the proof of Theorem \ref{DiamFinite} it follows that there exists a smooth $f : B^d_{\rho + \delta}(0) \to (-\varepsilon, \varepsilon)$ such that $\omega(\ux) = \big(f(\ux), \ux\big)$ is a smooth parametrisation of $\psi_{\rho + \delta}(\{r = r_1\} \cap V_{\rho + \delta})$. 

Recall now the definition of the Cauchy hypersurface $\Sigma$ of $N$, i.e., $\Sigma := (t_0 - \lambda, t_0 + \lambda) \times \{r_1\} \times \Sd$. In the following we construct finitely many charts $\mathring{\varphi}_\ell : W_\ell \to \mathring{A}_1$ for $\Sigma$, where each chart is contained in $\Sigma \cap V_{\rho + \delta}$ and together they cover $\Sigma \cap V_\rho$. Moreover, the charts extend regularly to their closures, which will guarantee that the diameter of $\psi_{\rho + \delta}(\mathring{W}_\ell)$ with respect to the Euclidean metric on $D_{\rho + \delta}$ is finite. Restricting $\psi_{\rho + \delta} : V_{\rho + \delta} \to D_{\rho + \delta}$ to the flow-out of $\mathring{W}_\ell$ in $N$ under the vector field $\partial_0$ will give a collection of regular flow charts for $N$ which cover $N \cap V_\rho$. 

First we observe that $\overline{\Sigma}$ is a smooth manifold with boundary which is, moreover, compact. It follows that $\overline{\Sigma} \cap \psi^{-1}_{\rho + \delta} \big( \bigcup_{\ux \in \overline{B^d_{\rho}(0)}} I_{\ux} \times \{\ux\}\big) = \overline{\Sigma} \cap \overline{V_\rho}$ is a compact subset of $\overline{\Sigma}$ and $\overline{\Sigma} \cap V_{\rho + \delta}$ is an open neighbourhood thereof (open in $\overline{\Sigma}$). 
\begin{center}
\def\svgwidth{8cm}
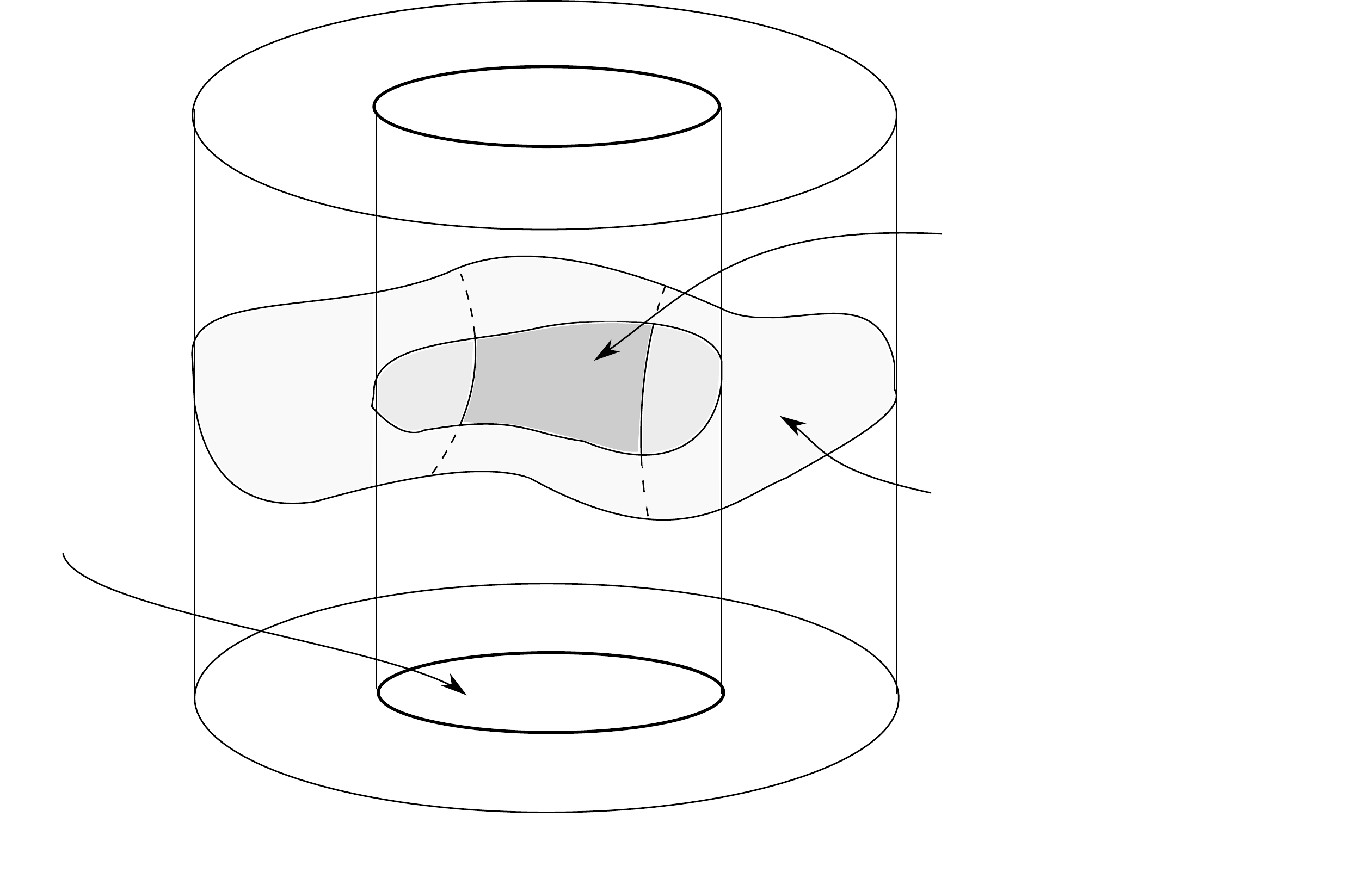
\end{center}
Let $A_i$ stand either for the open ball $B_i^d(0)$ or for the upper hemisphere $B_i^d(0) \cap \{x_d \geq 0\}$, where $i = 1, 2$. Consider now a collection of charts
\begin{equation*}
\varphi_\alpha : W_{2, \alpha} \to A_2
\end{equation*}
for $\overline{\Sigma}$ such that
\begin{enumerate}[i)]
\item $W_{2, \alpha} \subseteq \overline{\Sigma} \cap V_{\rho + \delta}$
\item $\overline{\Sigma} \cap \overline{V_{\rho}} \subseteq \bigcup_{\alpha} W_{1,\alpha}$, where $W_{1, \alpha} := \varphi_{\alpha}^{-1}(A_1)$.
\end{enumerate}
By the compactness of $\overline{\Sigma} \cap \overline{V_\rho}$ we can extract a finite subfamily $\alpha_1, \ldots, \alpha_L$ such that $ \overline{\Sigma} \cap \overline{V_{\rho}}\subseteq \bigcup_{1 \leq \ell \leq L} W_{1, \alpha_\ell}$. We set $\varphi_{\alpha_\ell} =: \varphi_\ell$.

Let $\mathring{A}_1 \subseteq A_1$ now stand for either $B^d_1(0)$ or $B^d_1(0) \cap \{x_d >0\}$ and define $\mathring{W}_\ell := \varphi_\ell^{-1}(\mathring{A}_1) \subseteq \Sigma$. By the definition of a manifold with boundary we have 
\begin{equation}
\label{Cover}\Sigma \cap V_\rho \subseteq \bigcup_{1 \leq \ell \leq L} \mathring{W}_\ell\;.
\end{equation} 
By slight abuse of notation, we have thus shown the existence of finitely many charts
\begin{equation*}
\varphi_\ell : \mathring{W}_\ell \to \mathring{A}_1\;, \quad  1 \leq \ell \leq L \;,
\end{equation*}
for $\Sigma$ that cover $\Sigma \cap V_\rho$ and extend regularly to $\overline{\varphi}_\ell : \overline{\mathring{W}_\ell} \to \overline{\mathring{A}_1}$, where the first closure is in $\{r = r_1\}$ and the latter in $\R^d$.

Since $\mathring{W}_\ell$ is an open subset of $\Sigma \subseteq \{r = r_1\}$, $B_\ell := \omega^{-1} \big( \psi_{\rho + \delta}(\mathring{W}_\ell)\big)$ is an open subset of $B^d_{\rho + \delta}(0)$. For $\ux \in B_\ell$ denote with $J_{\ux}$ the maximal open subinterval of $I_{\ux}$ on which $\sigma_{\ux}$ is contained in $N$. Since $\mathring{W}_\ell \subseteq \Sigma \subseteq N$, we have $J_{\ux} \neq \emptyset$. We now set $D_\ell := \bigcup_{\ux \in B_\ell} J_{\ux} \times \{\ux\}$ and $U_\ell := \psi_{\rho + \delta}^{-1}(D_\ell) \subseteq N$ and claim that
\begin{equation}
\label{RegFlowChartApp}
\psi_{\rho + \delta}\big|_{U_\ell} : U_\ell \to D_\ell
\end{equation} 
is a regular flow chart for $N$.

To prove this claim we need to show that the coordinate diameter $\diam_e(B_\ell)$ is finite - it is obvious that $\psi_{\rho + \delta}\big|_{U_\ell}$ satisfies the other properties of Definition \ref{DefRegFlowChart}. Clearly, the coordinate diameter of $B_\ell$ is less or equal to the diameter of $\psi_{\rho + \delta}(\mathring{W}_\ell)$ with respect to the metric $\eta$ induced by the Euclidean metric on $(-\varepsilon, \varepsilon) \times B_{\rho + \delta}^d(0)$ - since one ignores the distance traversed in the $x_0$-direction. To estimate the latter diameter, consider $\psi_{\rho + \delta} \circ \varphi_\ell^{-1} : \mathring{A}_1 \to \psi_{\rho + \delta}(\mathring{W}_\ell)$ and recall that $\varphi_\ell$ extends regularly to $\overline{\varphi}_\ell : \overline{\mathring{W}}_\ell \to \overline{\mathring{A}_1}$. Hence, the metric components of the pull-back metric $(\psi_{\rho + \delta} \circ \varphi_\ell^{-1})^*\eta$ with respect to the standard coordinates on $\mathring{A}_1$ are uniformly bounded on $\mathring{A}_1$. This, together with the finite coordinate diameter of $\mathring{A}_1$, implies, as in the last paragraph of the proof of Theorem \ref{DiamFinite}, the finiteness of the diameter of $\psi_{\rho + \delta}(\mathring{W}_\ell)$ with respect to the metric $\eta$. Hence, \eqref{RegFlowChartApp} is indeed a regular flow chart for $N$.
\vspace*{2mm}

\textbf{Step 3.1:} We show that $ N \cap \{ \omega \in B^{\Sd}_\kappa(\omega_0)\} \subseteq \bigcup_{1 \leq \ell \leq L} U_\ell$.
\vspace*{2mm}

By property 2 of Step 1 and \eqref{BlockContained} we have
\begin{equation*}
N \cap \{\omega \in B^{\Sd}_\kappa (\omega_0)\} \subseteq [t_0 - \lambda, t_0 + \lambda] \times (0, r_0] \times B^{\Sd}_\kappa(\omega_0) \subseteq V_\rho \;.
\end{equation*}
Given a $p \in N \cap \{ \omega \in B^{\Sd}_\kappa(\omega_0)\}$ consider the timelike curve $\sigma_{\underline{x}_0}$ that goes through $p$. The curve $\sigma_{\underline{x}_0}$ has to intersect $\Sigma$ in $V_\rho$, and hence, by \eqref{Cover}, $\sigma_{\underline{x}_0}$ intersects $W_\ell$ for some $1 \leq \ell \leq L$. Thus, $p$ is contained in $U_\ell$.
\vspace*{2mm}

We now construct a cover of regular flow charts for $N$. There is a finite number of rotations $R_k \in SO(d)$, $1 \leq k \leq K$, such that $\Sd \subseteq \bigcup_{1 \leq k \leq K} R_k \big(B^{\Sd}_\kappa(\omega_0)\big)$. Let us denote the symmetry action of $R_k$ on $\Mint$ by the same symbol. It follows that 
\begin{equation*}
\big{\{}\psi_{\rho \circ \delta}\big|_{U_\ell} \circ R_k^{-1} : R_k(U_\ell) \to D_\ell \quad \big| \quad 1 \leq \ell \leq L, \; 1 \leq k \leq K\big{\}}
\end{equation*}
is a finite cover of regular flow charts for $N$. The contradiction to Step 2 is now obtained from Theorem \ref{DiamFinite}.
This concludes the proof of Theorem \ref{Inex}.
\end{proof}

\section*{Acknowledgements}
I would like to thank Mihalis Dafermos for bringing this problem to my attention, and also for helpful comments on a preliminary version of this paper. Moreover, I am grateful to  Jo\~ao Costa, Dejan Gajic, and David Garfinkle for stimulating remarks. Special thanks goes to Paul Klinger and Melanie Graf for pointing out a mistake in a first version of this manuscript.
Finally, I would like to thank Magdalene College, Cambridge, for their financial support.

\bibliographystyle{acm}
\bibliography{Bibly}

\end{document}

%% file: StartingPoint.pdf_tex
%% Creator: Inkscape inkscape 0.48.2, www.inkscape.org
%% Accompanies image file 'StartingPoint.pdf' (pdf, eps, ps)
\begingroup%
  \makeatletter%
  \providecommand\color[2][]{%
    \errmessage{(Inkscape) Color is used for the text in Inkscape, but the package 'color.sty' is not loaded}%
    \renewcommand\color[2][]{}%
  }%
  \providecommand\transparent[1]{%
    \errmessage{(Inkscape) Transparency is used (non-zero) for the text in Inkscape, but the package 'transparent.sty' is not loaded}%
    \renewcommand\transparent[1]{}%
  }%
  \providecommand\rotatebox[2]{#2}%
  \ifx\svgwidth\undefined%
    \setlength{\unitlength}{487.85375977bp}%
    \ifx\svgscale\undefined%
      \relax%
    \else%
      \setlength{\unitlength}{\unitlength * \real{\svgscale}}%
    \fi%
  \else%
    \setlength{\unitlength}{\svgwidth}%
  \fi%
  \global\let\svgwidth\undefined%
  \global\let\svgscale\undefined%
  \makeatother%
  \begin{picture}(1,0.60299102)%
    \put(0,0){\includegraphics[width=\unitlength]{StartingPoint.pdf}}%
    \put(0.52453607,0.50970472){\color[rgb]{0,0,0}\makebox(0,0)[lb]{\smash{$\tilde{U}$}}}%
    \put(0.15612718,0.33000889){\color[rgb]{0,0,0}\makebox(0,0)[lb]{\smash{$\tilde{\gamma}(0)$}}}%
    \put(0.03163866,0.03798498){\color[rgb]{0,0,0}\makebox(0,0)[lb]{\smash{$\Mmax$}}}%
    \put(0.34134063,0.01908583){\color[rgb]{0,0,0}\makebox(0,0)[lb]{\smash{$\tilde{\gamma}(s_0)$}}}%
    \put(0.51950256,0.12567755){\color[rgb]{0,0,0}\makebox(0,0)[lb]{\smash{$I^+\big(\tilde{\gamma}(s_0), \tilde{U}\big) \cap I^-\big(\tilde{\gamma}(0), \tilde{U}\big)$}}}%
    \put(0.33573868,0.42523386){\color[rgb]{0,0,0}\makebox(0,0)[lb]{\smash{$\tilde{\gamma}$}}}%
  \end{picture}%
\endgroup%

%% file: PicN1.pdf_tex
%% Creator: Inkscape inkscape 0.48.2, www.inkscape.org
%% Accompanies image file 'PicN1.pdf' (pdf, eps, ps)
\begingroup%
  \makeatletter%
  \providecommand\color[2][]{%
    \errmessage{(Inkscape) Color is used for the text in Inkscape, but the package 'color.sty' is not loaded}%
    \renewcommand\color[2][]{}%
  }%
  \providecommand\transparent[1]{%
    \errmessage{(Inkscape) Transparency is used (non-zero) for the text in Inkscape, but the package 'transparent.sty' is not loaded}%
    \renewcommand\transparent[1]{}%
  }%
  \providecommand\rotatebox[2]{#2}%
  \ifx\svgwidth\undefined%
    \setlength{\unitlength}{689.56611328bp}%
    \ifx\svgscale\undefined%
      \relax%
    \else%
      \setlength{\unitlength}{\unitlength * \real{\svgscale}}%
    \fi%
  \else%
    \setlength{\unitlength}{\svgwidth}%
  \fi%
  \global\let\svgwidth\undefined%
  \global\let\svgscale\undefined%
  \makeatother%
  \begin{picture}(1,0.37010592)%
    \put(0,0){\includegraphics[width=\unitlength]{PicN1.pdf}}%
    \put(0.22729099,0.27726291){\color[rgb]{0,0,0}\makebox(0,0)[lb]{\smash{$\tilde{U}$}}}%
    \put(0.02675079,0.08832419){\color[rgb]{0,0,0}\makebox(0,0)[lb]{\smash{$\Mmax$}}}%
    \put(0.31181619,0.14135965){\color[rgb]{0,0,0}\makebox(0,0)[lb]{\smash{$N$}}}%
    \put(-0.00971108,0.34852921){\color[rgb]{0,0,0}\makebox(0,0)[lb]{\smash{$I^+\big(\tilde{\gamma}(-\mu), \Mmax\big)$}}}%
    \put(0.27369698,0.00545632){\color[rgb]{0,0,0}\makebox(0,0)[lb]{\smash{$t= const$}}}%
    \put(0.69798039,0.03528883){\color[rgb]{0,0,0}\makebox(0,0)[lb]{\smash{$\tilde{U} \cap \Mmax$}}}%
  \end{picture}%
\endgroup%

%% file: FutureOpen.pdf_tex
%% Creator: Inkscape inkscape 0.48.2, www.inkscape.org
%% Accompanies image file 'FutureOpen.pdf' (pdf, eps, ps)
\begingroup%
  \makeatletter%
  \providecommand\color[2][]{%
    \errmessage{(Inkscape) Color is used for the text in Inkscape, but the package 'color.sty' is not loaded}%
    \renewcommand\color[2][]{}%
  }%
  \providecommand\transparent[1]{%
    \errmessage{(Inkscape) Transparency is used (non-zero) for the text in Inkscape, but the package 'transparent.sty' is not loaded}%
    \renewcommand\transparent[1]{}%
  }%
  \providecommand\rotatebox[2]{#2}%
  \ifx\svgwidth\undefined%
    \setlength{\unitlength}{427.37993164bp}%
    \ifx\svgscale\undefined%
      \relax%
    \else%
      \setlength{\unitlength}{\unitlength * \real{\svgscale}}%
    \fi%
  \else%
    \setlength{\unitlength}{\svgwidth}%
  \fi%
  \global\let\svgwidth\undefined%
  \global\let\svgscale\undefined%
  \makeatother%
  \begin{picture}(1,0.9379848)%
    \put(0,0){\includegraphics[width=\unitlength]{FutureOpen.pdf}}%
    \put(0.80851443,0.42217369){\color[rgb]{0,0,0}\makebox(0,0)[lb]{\smash{$x_1$}}}%
    \put(0.34856908,0.87944496){\color[rgb]{0,0,0}\makebox(0,0)[lb]{\smash{$x_0$}}}%
    \put(0.45311033,0.13211124){\color[rgb]{0,0,0}\makebox(0,0)[lb]{\smash{$(-\frac{\varepsilon}{2},0,\ldots,0)$}}}%
    \put(0.20599888,0.32811338){\color[rgb]{0,0,0}\makebox(0,0)[lb]{\smash{$x$}}}%
    \put(0.4493286,0.62562961){\color[rgb]{0,0,0}\makebox(0,0)[lb]{\smash{$\varphi(q)$}}}%
    \put(0.54198146,0.23232762){\color[rgb]{0,0,0}\makebox(0,0)[lb]{\smash{$\varphi \circ \gamma$}}}%
    \put(0.28963309,0.66465653){\color[rgb]{0,0,0}\makebox(0,0)[lb]{\smash{$B_\rho(0)$}}}%
  \end{picture}%
\endgroup%

%% file: FlowChart.pdf_tex
%% Creator: Inkscape inkscape 0.48.2, www.inkscape.org
%% Accompanies image file 'FlowChart.pdf' (pdf, eps, ps)
\begingroup%
  \makeatletter%
  \providecommand\color[2][]{%
    \errmessage{(Inkscape) Color is used for the text in Inkscape, but the package 'color.sty' is not loaded}%
    \renewcommand\color[2][]{}%
  }%
  \providecommand\transparent[1]{%
    \errmessage{(Inkscape) Transparency is used (non-zero) for the text in Inkscape, but the package 'transparent.sty' is not loaded}%
    \renewcommand\transparent[1]{}%
  }%
  \providecommand\rotatebox[2]{#2}%
  \ifx\svgwidth\undefined%
    \setlength{\unitlength}{560.49741211bp}%
    \ifx\svgscale\undefined%
      \relax%
    \else%
      \setlength{\unitlength}{\unitlength * \real{\svgscale}}%
    \fi%
  \else%
    \setlength{\unitlength}{\svgwidth}%
  \fi%
  \global\let\svgwidth\undefined%
  \global\let\svgscale\undefined%
  \makeatother%
  \begin{picture}(1,0.91485703)%
    \put(0,0){\includegraphics[width=\unitlength]{FlowChart.pdf}}%
    \put(0.51468873,0.78073916){\color[rgb]{0,0,0}\makebox(0,0)[lb]{\smash{$\Delta$}}}%
    \put(0.49788194,0.09569344){\color[rgb]{0,0,0}\makebox(0,0)[lb]{\smash{$-\Delta$}}}%
    \put(0.79691603,0.4825497){\color[rgb]{0,0,0}\makebox(0,0)[lb]{\smash{$E$}}}%
    \put(0.11889189,0.47931627){\color[rgb]{0,0,0}\makebox(0,0)[lb]{\smash{$-E$}}}%
    \put(0.80795568,0.71499626){\color[rgb]{0,0,0}\makebox(0,0)[lb]{\smash{new $E$}}}%
    \put(0.01012494,0.28068816){\color[rgb]{0,0,0}\makebox(0,0)[lb]{\smash{new $-E$}}}%
    \put(0.18471968,0.86588268){\color[rgb]{0,0,0}\makebox(0,0)[lb]{\smash{$I^+\big(0,(-\Delta, \Delta) \times (-E,E)^d\big)$}}}%
    \put(0.18675869,0.0237735){\color[rgb]{0,0,0}\makebox(0,0)[lb]{\smash{$I^-\big(0,(-\Delta, \Delta) \times (-E,E)^d\big)$}}}%
  \end{picture}%
\endgroup%

%% file: Homotopy.pdf_tex
%% Creator: Inkscape inkscape 0.48.2, www.inkscape.org
%% Accompanies image file 'Homotopy.pdf' (pdf, eps, ps)
\begingroup%
  \makeatletter%
  \providecommand\color[2][]{%
    \errmessage{(Inkscape) Color is used for the text in Inkscape, but the package 'color.sty' is not loaded}%
    \renewcommand\color[2][]{}%
  }%
  \providecommand\transparent[1]{%
    \errmessage{(Inkscape) Transparency is used (non-zero) for the text in Inkscape, but the package 'transparent.sty' is not loaded}%
    \renewcommand\transparent[1]{}%
  }%
  \providecommand\rotatebox[2]{#2}%
  \ifx\svgwidth\undefined%
    \setlength{\unitlength}{422.03139648bp}%
    \ifx\svgscale\undefined%
      \relax%
    \else%
      \setlength{\unitlength}{\unitlength * \real{\svgscale}}%
    \fi%
  \else%
    \setlength{\unitlength}{\svgwidth}%
  \fi%
  \global\let\svgwidth\undefined%
  \global\let\svgscale\undefined%
  \makeatother%
  \begin{picture}(1,0.63200483)%
    \put(0,0){\includegraphics[width=\unitlength]{Homotopy.pdf}}%
    \put(0.46281847,0.57634398){\color[rgb]{0,0,0}\makebox(0,0)[lb]{\smash{$\lambda(1)$}}}%
    \put(0.2502714,0.50932461){\color[rgb]{0,0,0}\makebox(0,0)[lb]{\smash{$\lambda(t)$}}}%
    \put(0.67153589,0.24316205){\color[rgb]{0,0,0}\makebox(0,0)[lb]{\smash{$\lambda$}}}%
    \put(0.2713346,0.3867749){\color[rgb]{0,0,0}\makebox(0,0)[lb]{\smash{$\sigma_{\lambda(t)}$}}}%
    \put(0.36324688,0.58400335){\color[rgb]{0,0,0}\makebox(0,0)[lb]{\smash{$x_0$}}}%
    \put(0.88599782,0.1225272){\color[rgb]{0,0,0}\makebox(0,0)[lb]{\smash{$\underline{x}$}}}%
  \end{picture}%
\endgroup%

%% file: Kruskal.pdf_tex
%% Creator: Inkscape inkscape 0.48.2, www.inkscape.org
%% Accompanies image file 'Kruskal.pdf' (pdf, eps, ps)
\begingroup%
  \makeatletter%
  \providecommand\color[2][]{%
    \errmessage{(Inkscape) Color is used for the text in Inkscape, but the package 'color.sty' is not loaded}%
    \renewcommand\color[2][]{}%
  }%
  \providecommand\transparent[1]{%
    \errmessage{(Inkscape) Transparency is used (non-zero) for the text in Inkscape, but the package 'transparent.sty' is not loaded}%
    \renewcommand\transparent[1]{}%
  }%
  \providecommand\rotatebox[2]{#2}%
  \ifx\svgwidth\undefined%
    \setlength{\unitlength}{360.05891113bp}%
    \ifx\svgscale\undefined%
      \relax%
    \else%
      \setlength{\unitlength}{\unitlength * \real{\svgscale}}%
    \fi%
  \else%
    \setlength{\unitlength}{\svgwidth}%
  \fi%
  \global\let\svgwidth\undefined%
  \global\let\svgscale\undefined%
  \makeatother%
  \begin{picture}(1,0.98618783)%
    \put(0,0){\includegraphics[width=\unitlength]{Kruskal.pdf}}%
    \put(-0.00590182,0.84700429){\color[rgb]{0,0,0}\makebox(0,0)[lb]{\smash{$u$}}}%
    \put(0.87331934,0.82795976){\color[rgb]{0,0,0}\makebox(0,0)[lb]{\smash{$v$}}}%
    \put(0.73048565,0.46928832){\color[rgb]{0,0,0}\makebox(0,0)[lb]{\smash{$I$}}}%
    \put(0.44164396,0.65021107){\color[rgb]{0,0,0}\makebox(0,0)[lb]{\smash{$II$}}}%
    \put(0.40990308,0.28836544){\color[rgb]{0,0,0}\makebox(0,0)[lb]{\smash{$III$}}}%
    \put(0.10201706,0.46928832){\color[rgb]{0,0,0}\makebox(0,0)[lb]{\smash{$IV$}}}%
  \end{picture}%
\endgroup%

%% file: PInN.pdf_tex
%% Creator: Inkscape inkscape 0.48.2, www.inkscape.org
%% Accompanies image file 'PInN.pdf' (pdf, eps, ps)
\begingroup%
  \makeatletter%
  \providecommand\color[2][]{%
    \errmessage{(Inkscape) Color is used for the text in Inkscape, but the package 'color.sty' is not loaded}%
    \renewcommand\color[2][]{}%
  }%
  \providecommand\transparent[1]{%
    \errmessage{(Inkscape) Transparency is used (non-zero) for the text in Inkscape, but the package 'transparent.sty' is not loaded}%
    \renewcommand\transparent[1]{}%
  }%
  \providecommand\rotatebox[2]{#2}%
  \ifx\svgwidth\undefined%
    \setlength{\unitlength}{530.61474609bp}%
    \ifx\svgscale\undefined%
      \relax%
    \else%
      \setlength{\unitlength}{\unitlength * \real{\svgscale}}%
    \fi%
  \else%
    \setlength{\unitlength}{\svgwidth}%
  \fi%
  \global\let\svgwidth\undefined%
  \global\let\svgscale\undefined%
  \makeatother%
  \begin{picture}(1,0.62612011)%
    \put(0,0){\includegraphics[width=\unitlength]{PInN.pdf}}%
    \put(0.13763345,0.14587558){\color[rgb]{0,0,0}\makebox(0,0)[lb]{\smash{$p$}}}%
    \put(-0.00400479,0.43219801){\color[rgb]{0,0,0}\makebox(0,0)[lb]{\smash{$\gamma(s_{n+1})$}}}%
    \put(0.76053717,0.10780064){\color[rgb]{0,0,0}\makebox(0,0)[lb]{\smash{$\gamma(s_n)$}}}%
  \end{picture}%
\endgroup%

%% file: PInBoundary.pdf_tex
%% Creator: Inkscape inkscape 0.48.2, www.inkscape.org
%% Accompanies image file 'PInBoundary.pdf' (pdf, eps, ps)
\begingroup%
  \makeatletter%
  \providecommand\color[2][]{%
    \errmessage{(Inkscape) Color is used for the text in Inkscape, but the package 'color.sty' is not loaded}%
    \renewcommand\color[2][]{}%
  }%
  \providecommand\transparent[1]{%
    \errmessage{(Inkscape) Transparency is used (non-zero) for the text in Inkscape, but the package 'transparent.sty' is not loaded}%
    \renewcommand\transparent[1]{}%
  }%
  \providecommand\rotatebox[2]{#2}%
  \ifx\svgwidth\undefined%
    \setlength{\unitlength}{530.61474609bp}%
    \ifx\svgscale\undefined%
      \relax%
    \else%
      \setlength{\unitlength}{\unitlength * \real{\svgscale}}%
    \fi%
  \else%
    \setlength{\unitlength}{\svgwidth}%
  \fi%
  \global\let\svgwidth\undefined%
  \global\let\svgscale\undefined%
  \makeatother%
  \begin{picture}(1,0.62612011)%
    \put(0,0){\includegraphics[width=\unitlength]{PInBoundary.pdf}}%
    \put(0.13763345,0.14587558){\color[rgb]{0,0,0}\makebox(0,0)[lb]{\smash{$p$}}}%
    \put(0.76053717,0.10780064){\color[rgb]{0,0,0}\makebox(0,0)[lb]{\smash{$\gamma(s_n)$}}}%
    \put(0.79500182,0.4782195){\color[rgb]{0,0,0}\makebox(0,0)[lb]{\smash{$\Gamma$}}}%
  \end{picture}%
\endgroup%

%% file: Case1.pdf_tex
%% Creator: Inkscape inkscape 0.48.2, www.inkscape.org
%% Accompanies image file 'Case1.pdf' (pdf, eps, ps)
\begingroup%
  \makeatletter%
  \providecommand\color[2][]{%
    \errmessage{(Inkscape) Color is used for the text in Inkscape, but the package 'color.sty' is not loaded}%
    \renewcommand\color[2][]{}%
  }%
  \providecommand\transparent[1]{%
    \errmessage{(Inkscape) Transparency is used (non-zero) for the text in Inkscape, but the package 'transparent.sty' is not loaded}%
    \renewcommand\transparent[1]{}%
  }%
  \providecommand\rotatebox[2]{#2}%
  \ifx\svgwidth\undefined%
    \setlength{\unitlength}{596.52324219bp}%
    \ifx\svgscale\undefined%
      \relax%
    \else%
      \setlength{\unitlength}{\unitlength * \real{\svgscale}}%
    \fi%
  \else%
    \setlength{\unitlength}{\svgwidth}%
  \fi%
  \global\let\svgwidth\undefined%
  \global\let\svgscale\undefined%
  \makeatother%
  \begin{picture}(1,0.78035215)%
    \put(0,0){\includegraphics[width=\unitlength]{Case1.pdf}}%
    \put(0.35470418,0.00630741){\color[rgb]{0,0,0}\makebox(0,0)[lb]{\smash{$U$}}}%
    \put(0.72829759,0.36457393){\color[rgb]{0,0,0}\makebox(0,0)[lb]{\smash{$h_f(\underline{x}_\infty)$}}}%
    \put(0.72063413,0.6787755){\color[rgb]{0,0,0}\makebox(0,0)[lb]{\smash{$h_f(\underline{x}_\infty) + \Delta$}}}%
    \put(0.71680235,0.06186739){\color[rgb]{0,0,0}\makebox(0,0)[lb]{\smash{$h_f(\underline{x}_\infty) - \Delta >0 $}}}%
    \put(0.45049738,0.17490333){\color[rgb]{0,0,0}\makebox(0,0)[lb]{\smash{$\mathcal{D}_{\psi,1} \cap W$}}}%
    \put(0.36428349,0.46228289){\color[rgb]{0,0,0}\makebox(0,0)[lb]{\smash{$\sigma$}}}%
    \put(0.09606262,0.25153791){\color[rgb]{0,0,0}\makebox(0,0)[lb]{\smash{$\big(h_f(\underline{x}_\infty), \underline{x}_\infty\big)$}}}%
    \put(0.16120197,0.74774655){\color[rgb]{0,0,0}\makebox(0,0)[lb]{\smash{$\big(h_f(\underline{x}_\infty) + \frac{\Delta}{2}, \underline{x}_{n_{k_0}}\big)$}}}%
    \put(-0.00356231,0.67685961){\color[rgb]{0,0,0}\makebox(0,0)[lb]{\smash{$W$}}}%
  \end{picture}%
\endgroup%

%% file: BoundaryContinuity.pdf_tex
%% Creator: Inkscape inkscape 0.48.2, www.inkscape.org
%% Accompanies image file 'BoundaryContinuity.pdf' (pdf, eps, ps)
\begingroup%
  \makeatletter%
  \providecommand\color[2][]{%
    \errmessage{(Inkscape) Color is used for the text in Inkscape, but the package 'color.sty' is not loaded}%
    \renewcommand\color[2][]{}%
  }%
  \providecommand\transparent[1]{%
    \errmessage{(Inkscape) Transparency is used (non-zero) for the text in Inkscape, but the package 'transparent.sty' is not loaded}%
    \renewcommand\transparent[1]{}%
  }%
  \providecommand\rotatebox[2]{#2}%
  \ifx\svgwidth\undefined%
    \setlength{\unitlength}{701.00537109bp}%
    \ifx\svgscale\undefined%
      \relax%
    \else%
      \setlength{\unitlength}{\unitlength * \real{\svgscale}}%
    \fi%
  \else%
    \setlength{\unitlength}{\svgwidth}%
  \fi%
  \global\let\svgwidth\undefined%
  \global\let\svgscale\undefined%
  \makeatother%
  \begin{picture}(1,0.66404369)%
    \put(0,0){\includegraphics[width=\unitlength]{BoundaryContinuity.pdf}}%
    \put(0.56454632,0.00536731){\color[rgb]{0,0,0}\makebox(0,0)[lb]{\smash{$U$}}}%
    \put(0.88245708,0.3102356){\color[rgb]{0,0,0}\makebox(0,0)[lb]{\smash{$0$}}}%
    \put(0.87593582,0.57760664){\color[rgb]{0,0,0}\makebox(0,0)[lb]{\smash{$\Delta$}}}%
    \put(0.87267516,0.05264629){\color[rgb]{0,0,0}\makebox(0,0)[lb]{\smash{$-\Delta$}}}%
    \put(0.72594714,0.38686016){\color[rgb]{0,0,0}\makebox(0,0)[lb]{\smash{$\sigma$}}}%
    \put(0.34445427,0.21404716){\color[rgb]{0,0,0}\makebox(0,0)[lb]{\smash{$(0, \underline{x}_\infty)$}}}%
    \put(0.37216956,0.63629783){\color[rgb]{0,0,0}\makebox(0,0)[lb]{\smash{$\big(\frac{\Delta}{2}, \underline{x}_{n_0}\big)$}}}%
    \put(0.25967806,0.57597631){\color[rgb]{0,0,0}\makebox(0,0)[lb]{\smash{$W$}}}%
    \put(-0.00303136,0.45485692){\color[rgb]{0,0,0}\makebox(0,0)[lb]{\smash{$\big(\frac{\Delta}{3}, \underline{x}_\infty\big)$}}}%
  \end{picture}%
\endgroup%

%% file: Neighbourhood1.pdf_tex
%% Creator: Inkscape inkscape 0.48.2, www.inkscape.org
%% Accompanies image file 'Neighbourhood1.pdf' (pdf, eps, ps)
\begingroup%
  \makeatletter%
  \providecommand\color[2][]{%
    \errmessage{(Inkscape) Color is used for the text in Inkscape, but the package 'color.sty' is not loaded}%
    \renewcommand\color[2][]{}%
  }%
  \providecommand\transparent[1]{%
    \errmessage{(Inkscape) Transparency is used (non-zero) for the text in Inkscape, but the package 'transparent.sty' is not loaded}%
    \renewcommand\transparent[1]{}%
  }%
  \providecommand\rotatebox[2]{#2}%
  \ifx\svgwidth\undefined%
    \setlength{\unitlength}{535.85180664bp}%
    \ifx\svgscale\undefined%
      \relax%
    \else%
      \setlength{\unitlength}{\unitlength * \real{\svgscale}}%
    \fi%
  \else%
    \setlength{\unitlength}{\svgwidth}%
  \fi%
  \global\let\svgwidth\undefined%
  \global\let\svgscale\undefined%
  \makeatother%
  \begin{picture}(1,0.61001677)%
    \put(0,0){\includegraphics[width=\unitlength]{Neighbourhood1.pdf}}%
    \put(0.41436352,0.56684792){\color[rgb]{0,0,0}\makebox(0,0)[lb]{\smash{$\tilde{O}$}}}%
    \put(0.29290188,0.54526138){\color[rgb]{0,0,0}\makebox(0,0)[lb]{\smash{$q$}}}%
    \put(0.63404045,0.56409047){\color[rgb]{0,0,0}\makebox(0,0)[lb]{\smash{$\iota^{-1}\Big(I^-\big(\iota(q),\tilde{O}\big)\Big)$}}}%
    \put(0.48506697,0.29036193){\color[rgb]{0,0,0}\makebox(0,0)[lb]{\smash{$\gamma(-1)$}}}%
    \put(0.6922579,0.35586134){\color[rgb]{0,0,0}\makebox(0,0)[lb]{\smash{$S_\mathrm{aux}$}}}%
    \put(0.56303393,0.11271609){\color[rgb]{0,0,0}\makebox(0,0)[lb]{\smash{$M_\mathrm{int}$}}}%
    \put(0.20173559,0.1611012){\color[rgb]{0,0,0}\makebox(0,0)[lb]{\smash{$\gamma$}}}%
    \put(0.00960732,0.22400539){\color[rgb]{0,0,0}\makebox(0,0)[lb]{\smash{$I^-(q,M_\mathrm{int})$}}}%
  \end{picture}%
\endgroup%

%% file: Cones.pdf_tex
%% Creator: Inkscape inkscape 0.48.2, www.inkscape.org
%% Accompanies image file 'Cones.pdf' (pdf, eps, ps)
\begingroup%
  \makeatletter%
  \providecommand\color[2][]{%
    \errmessage{(Inkscape) Color is used for the text in Inkscape, but the package 'color.sty' is not loaded}%
    \renewcommand\color[2][]{}%
  }%
  \providecommand\transparent[1]{%
    \errmessage{(Inkscape) Transparency is used (non-zero) for the text in Inkscape, but the package 'transparent.sty' is not loaded}%
    \renewcommand\transparent[1]{}%
  }%
  \providecommand\rotatebox[2]{#2}%
  \ifx\svgwidth\undefined%
    \setlength{\unitlength}{484.35976563bp}%
    \ifx\svgscale\undefined%
      \relax%
    \else%
      \setlength{\unitlength}{\unitlength * \real{\svgscale}}%
    \fi%
  \else%
    \setlength{\unitlength}{\svgwidth}%
  \fi%
  \global\let\svgwidth\undefined%
  \global\let\svgscale\undefined%
  \makeatother%
  \begin{picture}(1,0.92620726)%
    \put(0,0){\includegraphics[width=\unitlength]{Cones.pdf}}%
    \put(0.74879405,0.84789396){\color[rgb]{0,0,0}\makebox(0,0)[lb]{\smash{$(-\varepsilon, \varepsilon)^{d+1}$}}}%
    \put(0.47048912,0.88605115){\color[rgb]{0,0,0}\makebox(0,0)[lb]{\smash{$x_0$}}}%
    \put(0.90101789,0.39336514){\color[rgb]{0,0,0}\makebox(0,0)[lb]{\smash{$\underline{x}$}}}%
    \put(0.46246695,0.33127055){\color[rgb]{0,0,0}\makebox(0,0)[lb]{\smash{$y^-$}}}%
    \put(0.46413534,0.13105841){\color[rgb]{0,0,0}\makebox(0,0)[lb]{\smash{$x^-$}}}%
    \put(0.46747224,0.72835759){\color[rgb]{0,0,0}\makebox(0,0)[lb]{\smash{$x^+$}}}%
    \put(0.60595225,0.69498889){\color[rgb]{0,0,0}\makebox(0,0)[lb]{\smash{$x^+ + C^-_{\nicefrac{5}{6}}$}}}%
    \put(0.64766309,0.55650893){\color[rgb]{0,0,0}\makebox(0,0)[lb]{\smash{$x^+ + C^-_{\nicefrac{6}{7}}$}}}%
    \put(0.57758881,0.20280109){\color[rgb]{0,0,0}\makebox(0,0)[lb]{\smash{$x^- + C^+_{\nicefrac{5}{6}}$}}}%
    \put(0.6159628,0.28288586){\color[rgb]{0,0,0}\makebox(0,0)[lb]{\smash{$x^- + C^+_{\nicefrac{6}{7}}$}}}%
    \put(0.08039568,0.17610619){\color[rgb]{0,0,0}\makebox(0,0)[lb]{\smash{$C^-_{\nicefrac{5}{8}} \cap \big(y^- + C^+_{\nicefrac{5}{8}}\big)$}}}%
  \end{picture}%
\endgroup%

%% file: Separation.pdf_tex
%% Creator: Inkscape inkscape 0.48.2, www.inkscape.org
%% Accompanies image file 'Separation.pdf' (pdf, eps, ps)
\begingroup%
  \makeatletter%
  \providecommand\color[2][]{%
    \errmessage{(Inkscape) Color is used for the text in Inkscape, but the package 'color.sty' is not loaded}%
    \renewcommand\color[2][]{}%
  }%
  \providecommand\transparent[1]{%
    \errmessage{(Inkscape) Transparency is used (non-zero) for the text in Inkscape, but the package 'transparent.sty' is not loaded}%
    \renewcommand\transparent[1]{}%
  }%
  \providecommand\rotatebox[2]{#2}%
  \ifx\svgwidth\undefined%
    \setlength{\unitlength}{461.55961914bp}%
    \ifx\svgscale\undefined%
      \relax%
    \else%
      \setlength{\unitlength}{\unitlength * \real{\svgscale}}%
    \fi%
  \else%
    \setlength{\unitlength}{\svgwidth}%
  \fi%
  \global\let\svgwidth\undefined%
  \global\let\svgscale\undefined%
  \makeatother%
  \begin{picture}(1,0.87301212)%
    \put(0,0){\includegraphics[width=\unitlength]{Separation.pdf}}%
    \put(0.15242731,0.50898242){\color[rgb]{0,0,0}\makebox(0,0)[lb]{\smash{$K$}}}%
    \put(0.44955655,0.52879096){\color[rgb]{0,0,0}\makebox(0,0)[lb]{\smash{$\gamma$}}}%
    \put(0.20937706,0.20690094){\color[rgb]{0,0,0}\makebox(0,0)[lb]{\smash{$\sigma$}}}%
    \put(0.53621919,0.09052538){\color[rgb]{0,0,0}\makebox(0,0)[lb]{\smash{$I^-\big(\gamma(y_0^-),M_\mathrm{int}\big)$}}}%
    \put(0.63773837,0.36289379){\color[rgb]{0,0,0}\makebox(0,0)[lb]{\smash{$\gamma(y_0^+)$}}}%
    \put(0.60802545,0.20937715){\color[rgb]{0,0,0}\makebox(0,0)[lb]{\smash{$\gamma(y_0^-)$}}}%
    \put(0.47926947,0.83087239){\color[rgb]{0,0,0}\makebox(0,0)[lb]{\smash{$\{r=0\}$}}}%
  \end{picture}%
\endgroup%

%% file: Cones2.pdf_tex
%% Creator: Inkscape inkscape 0.48.2, www.inkscape.org
%% Accompanies image file 'Cones2.pdf' (pdf, eps, ps)
\begingroup%
  \makeatletter%
  \providecommand\color[2][]{%
    \errmessage{(Inkscape) Color is used for the text in Inkscape, but the package 'color.sty' is not loaded}%
    \renewcommand\color[2][]{}%
  }%
  \providecommand\transparent[1]{%
    \errmessage{(Inkscape) Transparency is used (non-zero) for the text in Inkscape, but the package 'transparent.sty' is not loaded}%
    \renewcommand\transparent[1]{}%
  }%
  \providecommand\rotatebox[2]{#2}%
  \ifx\svgwidth\undefined%
    \setlength{\unitlength}{505.09819336bp}%
    \ifx\svgscale\undefined%
      \relax%
    \else%
      \setlength{\unitlength}{\unitlength * \real{\svgscale}}%
    \fi%
  \else%
    \setlength{\unitlength}{\svgwidth}%
  \fi%
  \global\let\svgwidth\undefined%
  \global\let\svgscale\undefined%
  \makeatother%
  \begin{picture}(1,0.88817884)%
    \put(0,0){\includegraphics[width=\unitlength]{Cones2.pdf}}%
    \put(0.7180499,0.81308095){\color[rgb]{0,0,0}\makebox(0,0)[lb]{\smash{$(-\varepsilon, \varepsilon)^{d+1}$}}}%
    \put(0.45117168,0.84967148){\color[rgb]{0,0,0}\makebox(0,0)[lb]{\smash{$x_0$}}}%
    \put(0.8640237,0.37721428){\color[rgb]{0,0,0}\makebox(0,0)[lb]{\smash{$\underline{x}$}}}%
    \put(0.44347888,0.31766917){\color[rgb]{0,0,0}\makebox(0,0)[lb]{\smash{$y^-$}}}%
    \put(0.44507878,0.12567738){\color[rgb]{0,0,0}\makebox(0,0)[lb]{\smash{$x^-$}}}%
    \put(0.44827867,0.69845253){\color[rgb]{0,0,0}\makebox(0,0)[lb]{\smash{$x^+$}}}%
    \put(0.58107294,0.66645389){\color[rgb]{0,0,0}\makebox(0,0)[lb]{\smash{$x^+ + C^-_{\nicefrac{5}{6}}$}}}%
    \put(0.62107121,0.53365967){\color[rgb]{0,0,0}\makebox(0,0)[lb]{\smash{$x^+ + C^-_{\nicefrac{6}{7}}$}}}%
    \put(0.55387405,0.19447444){\color[rgb]{0,0,0}\makebox(0,0)[lb]{\smash{$x^- + C^+_{\nicefrac{5}{6}}$}}}%
    \put(0.59067247,0.27127108){\color[rgb]{0,0,0}\makebox(0,0)[lb]{\smash{$x^- + C^+_{\nicefrac{6}{7}}$}}}%
    \put(0.07709478,0.17905751){\color[rgb]{0,0,0}\makebox(0,0)[lb]{\smash{$C^-_{\nicefrac{5}{8}} \cap \big(y^- + C^+_{\nicefrac{5}{8}}\big)$}}}%
    \put(0.82020826,0.46044797){\color[rgb]{0,0,0}\makebox(0,0)[lb]{\smash{$(\tilde{\varphi} \circ \iota)\big(\gamma(-\mu)\big)$}}}%
    \put(0.73309648,0.33826513){\color[rgb]{0,0,0}\makebox(0,0)[lb]{\smash{$y^+$}}}%
    \put(0.09503103,0.5792367){\color[rgb]{0,0,0}\makebox(0,0)[lb]{\smash{$\tilde{\varphi}\circ \iota \circ \sigma$}}}%
    \put(0.1911934,0.45479129){\color[rgb]{0,0,0}\makebox(0,0)[lb]{\smash{$\tau$}}}%
  \end{picture}%
\endgroup%

%% file: ConstructionN.pdf_tex
%% Creator: Inkscape inkscape 0.48.2, www.inkscape.org
%% Accompanies image file 'ConstructionN.pdf' (pdf, eps, ps)
\begingroup%
  \makeatletter%
  \providecommand\color[2][]{%
    \errmessage{(Inkscape) Color is used for the text in Inkscape, but the package 'color.sty' is not loaded}%
    \renewcommand\color[2][]{}%
  }%
  \providecommand\transparent[1]{%
    \errmessage{(Inkscape) Transparency is used (non-zero) for the text in Inkscape, but the package 'transparent.sty' is not loaded}%
    \renewcommand\transparent[1]{}%
  }%
  \providecommand\rotatebox[2]{#2}%
  \ifx\svgwidth\undefined%
    \setlength{\unitlength}{775.23037109bp}%
    \ifx\svgscale\undefined%
      \relax%
    \else%
      \setlength{\unitlength}{\unitlength * \real{\svgscale}}%
    \fi%
  \else%
    \setlength{\unitlength}{\svgwidth}%
  \fi%
  \global\let\svgwidth\undefined%
  \global\let\svgscale\undefined%
  \makeatother%
  \begin{picture}(1,0.28349404)%
    \put(0,0){\includegraphics[width=\unitlength]{ConstructionN.pdf}}%
    \put(0.16974215,0.24808521){\color[rgb]{0,0,0}\makebox(0,0)[lb]{\smash{$t = - u^*_\mathrm{int} > t_0 - \frac{\lambda}{2}$}}}%
    \put(0.54271887,0.24661101){\color[rgb]{0,0,0}\makebox(0,0)[lb]{\smash{$t= v^*_\mathrm{int} > t_0 + \frac{\lambda}{2}$}}}%
    \put(-0.07940037,0.15078692){\color[rgb]{0,0,0}\makebox(0,0)[lb]{\smash{$u^*_\mathrm{int} = r^*_\mathrm{int}(r_1) - (t_0 - \lambda)$}}}%
    \put(-0.07940037,0.09329255){\color[rgb]{0,0,0}\makebox(0,0)[lb]{\smash{$v^*_\mathrm{int} = r^*_\mathrm{int}(r_1) + (t_0 - \lambda)$}}}%
    \put(0.74616067,0.14931272){\color[rgb]{0,0,0}\makebox(0,0)[lb]{\smash{$u^*_\mathrm{int} = r^*_\mathrm{int}(r_1) - (t_0 + \lambda)$}}}%
    \put(0.74173803,0.09329248){\color[rgb]{0,0,0}\makebox(0,0)[lb]{\smash{$v^*_\mathrm{int} = r^*_\mathrm{int}(r_1) + (t_0 + \lambda)$}}}%
    \put(0.25819513,0.01810752){\color[rgb]{0,0,0}\makebox(0,0)[lb]{\smash{$(t_0 - \lambda, t_0 + \lambda) \times \{r_1\} \times \mathbb{S}^{d-1}$}}}%
    \put(0.66950144,0.03874654){\color[rgb]{0,0,0}\makebox(0,0)[lb]{\smash{$\{r = r_0\}$}}}%
    \put(0.4557401,0.15963227){\color[rgb]{0,0,0}\makebox(0,0)[lb]{\smash{$\Sigma_n$}}}%
    \put(0.48817284,0.08002457){\color[rgb]{0,0,0}\makebox(0,0)[lb]{\smash{$N$}}}%
  \end{picture}%
\endgroup%

%% file: Charts.pdf_tex
%% Creator: Inkscape inkscape 0.48.2, www.inkscape.org
%% Accompanies image file 'Charts.pdf' (pdf, eps, ps)
\begingroup%
  \makeatletter%
  \providecommand\color[2][]{%
    \errmessage{(Inkscape) Color is used for the text in Inkscape, but the package 'color.sty' is not loaded}%
    \renewcommand\color[2][]{}%
  }%
  \providecommand\transparent[1]{%
    \errmessage{(Inkscape) Transparency is used (non-zero) for the text in Inkscape, but the package 'transparent.sty' is not loaded}%
    \renewcommand\transparent[1]{}%
  }%
  \providecommand\rotatebox[2]{#2}%
  \ifx\svgwidth\undefined%
    \setlength{\unitlength}{725.99462891bp}%
    \ifx\svgscale\undefined%
      \relax%
    \else%
      \setlength{\unitlength}{\unitlength * \real{\svgscale}}%
    \fi%
  \else%
    \setlength{\unitlength}{\svgwidth}%
  \fi%
  \global\let\svgwidth\undefined%
  \global\let\svgscale\undefined%
  \makeatother%
  \begin{picture}(1,0.64956455)%
    \put(0,0){\includegraphics[width=\unitlength]{Charts.pdf}}%
    \put(0.69758976,0.46955536){\color[rgb]{0,0,0}\makebox(0,0)[lb]{\smash{$\psi_{\rho + \delta}(\overline{\Sigma} \cap \overline{V_\rho})$}}}%
    \put(0.68971876,0.28222615){\color[rgb]{0,0,0}\makebox(0,0)[lb]{\smash{$\psi_{\rho + \delta}\big(\{r = r_1\} \cap V_{\rho + \delta}\big)$}}}%
    \put(0.46460888,0.01933553){\color[rgb]{0,0,0}\makebox(0,0)[lb]{\smash{$B^d_{\rho + \delta}(0)$}}}%
    \put(-0.00292702,0.26491007){\color[rgb]{0,0,0}\makebox(0,0)[lb]{\smash{$B^d_{\rho}(0)$}}}%
  \end{picture}%
\endgroup%

%% file: InexSchwarzschild2.bbl
\begin{thebibliography}{10}

\bibitem{Beem80}
{\sc Beem, J.}
\newblock {Minkowski Space-Time is Locally Extendible}.
\newblock {\em Comm. Math. Phys. 72\/} (1980), 273--275.

\bibitem{Chris91}
{\sc Christodoulou, D.}
\newblock {The Formation of Black Holes and Singularities in Spherically
  Symmetric Gravitational Collapse}.
\newblock {\em Comm. Pure Appl. Math. 44\/} (1991), 339--373.

\bibitem{Chris93}
{\sc Christodoulou, D.}
\newblock {Bounded Variation Solutions of the Spherically Symmetric
  Einstein-Scalar Field Equations}.
\newblock {\em Comm. Pure Appl. Math. 46\/} (1993), 1131--1220.

\bibitem{Chris99}
{\sc Christodoulou, D.}
\newblock {The instability of naked singularities in the gravitational collapse
  of a scalar field}.
\newblock {\em Ann. of Math. 149\/} (1999), 183--217.

\bibitem{Chris09}
{\sc Christodoulou, D.}
\newblock {\em {The formation of black holes in general relativity}}.
\newblock European Mathematical Society, 2009.

\bibitem{ChrusGra12}
{\sc Chru\'sciel, P., and Grant, J.}
\newblock {On Lorentzian causality with continuous metrics}.
\newblock {\em Class. Quantum Grav. 29\/} (2012).

\bibitem{Daf03}
{\sc Dafermos, M.}
\newblock {Stability and instability of the Cauchy horizon for the spherically
  symmetric Einstein-Maxwell-scalar field equations}.
\newblock {\em Ann. of Math. 158\/} (2003), 875--928.

\bibitem{Daf05a}
{\sc Dafermos, M.}
\newblock {The interior of charged black holes and the problem of uniqueness in
  general relativity}.
\newblock {\em Comm. Pure Appl. Math. 58\/} (2005), 445--504.

\bibitem{DafLuk}
{\sc Dafermos, M., and Luk, J.}
\newblock {Stability of the Kerr Cauchy horizon}.
\newblock {\em in preparation\/}.

\bibitem{HawkEllis}
{\sc Hawking, S., and Ellis, G.}
\newblock {\em The large scale structure of space-time}.
\newblock Cambridge University Press, 1973.

\bibitem{Hirsch12}
{\sc Hirsch, M.}
\newblock {\em {Differential Topology}}.
\newblock Springer, 2012.

\bibitem{KlRodSzef12}
{\sc Klainerman, S., Rodnianski, I., and Szeftel, J.}
\newblock {The Bounded L2 Curvature Conjecture}.
\newblock {\em arXiv:1204.1767v2\/} (2012).

\bibitem{Krus60}
{\sc Kruskal, M.}
\newblock {Maximal extension of the Schwarzschild metric}.
\newblock {\em Phys. Rev. 119\/} (1960), 1743--1745.

\bibitem{LeeRiem}
{\sc Lee, J.}
\newblock {\em {Riemannian Manifolds: An Introduction to Curvature}}.
\newblock Springer, 1997.

\bibitem{LeeSmooth}
{\sc Lee, J.}
\newblock {\em {Introduction to Smooth Manifolds}}.
\newblock Springer, 2002.

\bibitem{LukRod12}
{\sc Luk, J., and Rodnianski, I.}
\newblock {Local propagation of impulsive gravitational waves}.
\newblock {\em arXiv:1209.1130v1\/} (2012).

\bibitem{ONeill}
{\sc O'Neill, B.}
\newblock {\em Semi-Riemannian geometry}.
\newblock Academic Press Inc, 1983.

\bibitem{Ori00}
{\sc Ori, A.}
\newblock {Strength of curvature singularities}.
\newblock {\em Phys. Rev. D 61\/} (2000).

\bibitem{Schw16}
{\sc Schwarzschild, K.}
\newblock {\"Uber das Gravitationsfeld eines Massenpunktes nach der
  Einsteinschen Theorie}.
\newblock {\em {Sitzungsber. d. Preuss. Akad. d. Wissenschaften} 1\/} (1916),
  189--196.

\bibitem{Synge50}
{\sc Synge, J.}
\newblock {The gravitational field of a particle}.
\newblock {\em Proc. Roy. Irish Acad. 53\/} (1950), 83--114.

\bibitem{Tang63}
{\sc Tangherlini, F.}
\newblock {Schwarzschild Field in n Dimensions and the Dimensionality of Space
  Problem}.
\newblock {\em II Nuovo Cimento XXVII}, 3 (1963), 636--651.

\end{thebibliography}
